\documentclass[12pt,usenames,dvipsnames]{article}
\usepackage[margin=1in]{geometry} 
\usepackage{amsmath,amsthm,amssymb,tikz}
\usepackage{mathpartir}
\usepackage{stix}
\usepackage{tabularx}
\usepackage{booktabs,colortbl}
\usepackage{makecell}
\usepackage{tgpagella}
\usepackage[normalem]{ulem}
\usepackage{xargs}
\usepackage{enumitem}
\usepackage{wasysym}
\usepackage{todonotes}
\usepackage{ifmtarg}
\usepackage{stmaryrd}
\usepackage{G} 
\usepackage{authblk}
\counterwithin{equation}{section}
\numberwithin{equation}{section}

\begin{document}
 
 
 
\title{Cost-Aware Type Theory}
\author[1]{Yue Niu\footnote{yuen@andrew.cmu.edu}} 
\author[1]{Robert Harper\footnote{rwh@cs.cmu.edu}}
\affil[1]{Carnegie Mellon University}

\maketitle

\begin{abstract}
  Although computational complexity is a fundamental aspect of program behavior,
  it is often at odds with common type theoretic principles such as function
  extensionality, which identifies all functions with the same \emph{input-output} behavior.
  We present a computational type theory called \cctt{} that has a primitive
  notion of cost (the number of evaluation steps). We introduce a new dependent function type
  ``funtime'' whose semantics can be viewed as a cost-aware version of function
  extensionality. We prove a collection of lemmas for \cctt{},
  including a novel introduction rule for the new funtime type. \cctt{} can be simultaneously viewed as a framework for analyzing
  computational complexity of programs and as the beginnings of a semantic foundation for
  characterizing feasible mathematical proofs. 
\end{abstract}

\section{Introduction}

Beginning with mechanized time complexity analysis of Lisp
programs~\cite{10.1145/361002.361016}, computational complexity in the context
of programming language research has been studied from various perspectives,
including definability of complexity classes inside type
theory~\cite{CONSTABLE1998137}\cite{CC01}, syntactic characterizations of
complexity
classes~\cite{10.1145/129712.129740}\cite{HOFMANN200357}\cite{HOFMANN2000113},
recurrence extraction~\cite{BENZINGER200479}\cite{10.1145/3371083},
substructural \& resource aware type
systems~\cite{Hofmann:2003:SPH:604131.604148}\cite{Jost:2010:SDQ:1707801.1706327}\cite{Hoffmann:2017:TAR:3009837.3009842},
and program logics for interactive verification of cost
bounds~\cite{Chargueraud:2019:VCA:3315672.3315720}\cite{10.1007/978-3-030-17184-1_1}\cite{10.1007/978-3-319-89884-1_19}.

What seemed to be missing from the discussion was a semantic theory
with a rich specification language in which computational complexity is a
primitive notion. This is the main motivation behind the work in this article.
The shift in perspective is subtle: instead of analyzing programs after the
fact, we develop a theory for simultaneous construction of cost-aware
specification and programs satisfying those specifications. This perspective on
computational complexity resonates with a distinctive feature in dependent type
theory: specifications \emph{are} programs, and the construction of one is
inextricably intertwined with the other. 

We develop a
computational dependent type theory, called \cctt{} (cost-aware type theory), for analyzing the
time complexity of programs. The central move borrows from the Martin-Löf
ethos in the formulation of computational type theories: type structure
internalizes prior judgmental notions. Following this idea, we introduce new
judgment forms to describe computational complexity. A new dependent function
type ``funtime'' is introduced to internalize the meaning of open judgments.
We prove a novel introduction rule for the funtime type. The proof relies on the
fact that a terminating recursive function induces a well-founded ordering on
the recursive calls (a similar idea was introduced in~\cite{cite-key} to add general
recursion to Martin-Löf type theories). 

A perhaps surprising fact is that the judgments in \cctt{} only speak about total
terms: there is no judgment form that acknowledges partiality. However, this is
not achieved through syntactic means (such as Gödel's T or system F).
While syntactic restrictions alleviate the 
burden of the termination proof from the user, total languages are susceptible to arbitrary blow
up in code size when compared to languages with general
recursion~\cite{BLUM1967257}.\footnote{An article by Robert Harper discusses
this in detail:
https://existentialtype.wordpress.com/2014/03/20/old-neglected-theorems-are-still-theorems/.}

Practically, programs in total languages can also have unexpected computational
complexity. For example, a more liberal form of recursion can be simulated in
Gödel's T by proving course-of-value (COV) recursion (the dependently typed version is
strong induction). One can then program via the COV-recursion interface, but the
resulting algorithm will not have the expected recursive
behavior: \emph{every} recursive argument is considered, regardless
of whether it was requested by the algorithm. This seems to be a fundamental
limitation of programming with primitive recursion. Hence, while programs
in Gödel's T can be \emph{written} in a more natural manner using an interface such as
COV-recursion, their \emph{computational behavior} and as a result \emph{complexity}
does not coincide with textbook algorithm analysis. 

The key innovation in \cctt{} leverages the synergetic relation between general
recursion and cost specification: general recursion is needed to write
algorithms with the expected cost, and the cost
specification in turn enables one to prove that a general recursive function is
total (and satisfies its functional specification).
This enables the programmer to write
general recursive functions that are not a priori terminating and
\emph{simultaneously} construct the proof that it is in fact terminating using the
cost specification. As a
consequence, the membership judgment \isComp{M}{A} immediately implies $M$
terminates, and all well-typed functions are total (an immediate consequence of
adhering to a cost specification).\\

\noindent\textbf{Language vs. machine}
Before diving into the specifics of \cctt{}, we take a moment to frame our work
in the bigger landscape. While the programming languages theory encompasses nearly all
imaginable computational behavior -- e.g., state, randomization, parallelism,
concurrency -- language-based semantics for computational complexity has
received comparatively little attention. Instead, machine models such as the Turing machine and
RAM are incumbent in algorithmic complexity literature. Part of the reason for
this discrepancy could be attributed to the belief that machine models are
``lower-level'' and so the induced computational complexity is a more accurate 
description of program performance in real life. 

However, it should be noted that computational complexity defined in a
language-based model (such as \cctt{}) \emph{do} reflect real-life performance,
as long as the implementation is \emph{resource
bounded}~\cite{lambda95}\cite{GREINER96psl}\cite{BLELLOCH96nesl}\cite{BlCh99}\cite{10.1145/1411204.1411240} by the cost
semantics. This means that each stage of compilation (from the source language
to the target language) is cost-preserving: the cost semantics of the higher-level
language induces an upper bound on the cost semantics of the lower-level language.
This process separates the concern of the algorithm analysis from low-level,
implementation specific concerns that are asymptotically irrelevant in any case.

Furthermore, language-based models have some clear advantages compared to
machine-models such as compositionality and modularity. Hence, a possible
direction for future work  
would be designing intermediate languages for complexity analysis and the
associated cost-preserving compilations, which would lead to a comprehensive
framework for describing both high-level and low-level computational complexity and
the effects of different compilation/scheduling strategies. \\

\noindent\textbf{Outline} 
In the following sections, we motivate and develop \cctt{} and demonstrate the applicability
of the theory with examples. In section~\ref{sec:setup}, we set the stage by
defining the underlying programming language \pl{} and discuss some considerations
when choosing an appropriate language for analyzing computational complexity. In
section~\ref{sec:meaning}, we give some intuition for \cctt{} by extending
Martin-Löf's meaning explanation for type theories to account for computational
complexity. In section~\ref{sec:construction}, we realize \cctt{} via Allen's
PER-semantics~\cite{Allen1987ANS} for computational type theories, which
involves constructing a relation $\tau : \ps{\val \times \val \times \ps{\val
\times \val}}$
satisfying some coherence conditions. A coherent relation $\tau$ is then a
\emph{type system} (in a technical sense), and the judgments of \cctt{} are
defined relative to a given type system $\tau$. In section~\ref{sec:cctt}, we
define the judgments of \cctt{} and prove some expected lemmas, including head
expansion and the structural rules of hypothesis, weakening, and substitution.
In section~\ref{sec:proof}, we prove a collection of lemmas about the semantics
of \cctt{} constituting a simple proof theory that we later use to verify some
example programs. In section~\ref{sec:machines}, we discuss the usual
assumptions made in algorithmic complexity and how they can be dealt with inside
type theory. In section~\ref{sec:gcd}, we analyze the time complexity of Euclid's
algorithm for computing the greatest common divisor. In
section~\ref{sec:parallel} and~\ref{sec:fib}, we extend \cctt{} to account for
parallelism and illustrate the extension by verifying the parallel complexity
of computing the Fibonacci numbers. In section~\ref{sec:related}
and~\ref{sec:concl}, we discuss some related work and promising directions for
investigation.
\section{Programming language for complexity analysis}\label{sec:setup}

\begin{figure}
\begin{align*}
    \mathsf{Exp} \quad E &::= a 
    \mid \arrtimev{\isOf{a}{E_1}}{E_2}{E_3}
    \mid \pity{\isOf{a}{E_1}}{E_2}
    \mid \fun{f}{a}{E} 
    \mid \ap{E_1}{E_2}
    \mid \nat
    \mid \zero
    \mid \suc{E}\\
    &\mid \ifz{E}{E_0}{a}{E_1}
    \mid \sigmaty{\isOf{a}{E_1}}{E_2}
    \mid \pair{E_1}{E_2}
    \mid \fst{E}
    \mid \snd{E}\\
    &\mid \eqty{E}{E_1}{E_2}
    \mid \triv 
    \mid \subsetty{\isOf{a}{E_1}}{E_2}
    \mid \letcst{E_1}{a}{E_2}
    \mid \univ{i}\\
    &\mid \rel{r}{E_1}{E_2}
    \mid \tern{r}{E_1}{E_2}{E_3}
    \mid \cffone{f}{E}
    \mid \cfftwo{f}{E_1}{E_2}
\end{align*}
\caption{Syntax of \pl{}.}
\label{fig:syntax}
\end{figure}

Although there are various abstract notions of computational complexity, 
we focus on analyzing a simple and concrete cost measure -- the number of steps a closed program
takes to evaluate to a value. Under this view of complexity, the analysis is fundamentally about the 
specification and verification of \emph{behavioral} properties of programs. 
Hence, the very nature of computational complexity suggests itself as an object of study inside
\emph{computational type theories}, in which typehood and membership are semantic properties, i.e., given by 
how programs compute. 
In the influential paper \emph{Constructive Mathematics and Computer Programming}~\cite{MARTINLOF1982153}, Martin-Löf introduced
a computational semantics for programming languages through what is known as the ``meaning explanation'',  
a notion closely related to Tait's method of logical relations~\cite{10.2307/2271658} and 
Girard's method of candidates of reducibility~\cite{10.5555/64805}.
For dependent type theory, the same perspectives can also be found in the computational type theory
of \nuprl~\cite{ALLEN2006428}, a program/proof refinement framework which has been used to verify a substantial amount of mathematics.\\

\noindent\textbf{Programming language \pl{}}
Following these computational type-theoretic approaches, we develop a dependent type theory
called \cctt{} (complexity-aware type theory) for analyzing computational complexity of programs. We begin by
defining a programming language \pl{} and an evaluation relation. The syntax of \pl{} is given in Figure~\ref{fig:syntax}.
Note that since typehood is a semantic property, there is no syntactic distinction between types and elements of a type.
The value relation $\final{\_} : \ps{\mathsf{Exp}}$ and step relation $\step{\_}{\_} : \ps{\mathsf{Exp} \times \mathsf{Exp}}$
are defined in Figure~\ref{fig:opersem}. Write $\val = \{E \mid \final{E}\}$ for the set of closed programs that are values.
Define evaluation as follows:
\begin{enumerate}
  \item \evalCost{M}{c}{V} when $\stepIn{M}{c}{V}$ and \final{V}.
  \item \eval{M}{V} when there is a $c$ s.t. \evalCost{M}{c}{V}.
\end{enumerate}

Where $\stepIn{M}{c}{V}$ means $M$ transition to $V$ in exactly $c$ steps.\\

\noindent\textbf{Complexity-aware function types}
In anticipation of analyzing complexity properties of programs, we introduce the
new funtime type
\arrtimev{\isOf{a}{A}}{B}{P}, which classify functions that comply with the cost bound stated by $P$.
Crucially, $P$ is an open term of \pl{} and \emph{not} a natural number since the cost $P$ varies in the input value.
Furthermore, it is also not a meta-theoretical function, because the variable $a$ ranges over terms of \pl{}. This means 
we cannot directly use the ambient arithmetic functions in defining the cost component $P$. One route would be to define
all the usual arithmetic functions internally. This turns out to burden the mechanization process with many secondary 
proof obligations which hold in the ambient mathematics. Instead, we import a part of the ambient arithmetic world as 
\emph{foreign functions}.\\

\noindent\textbf{Foreign functions}
We introduce several syntactic operators that mediate the use of meta-theoretical functions or relations 
on natural numbers.
The expression \rel{r}{E_1}{E_2} reifies a (meta-)relation $r : \N \times \N$ into a type, and \cffone{f_1}{E} 
and \cfftwo{f_2}{E_1}{E_2} form applications of (meta-)functions $f_1 : \N \to \N$ and $f_2 : \N \to \N \to \N$, respectively.
In principle, we can define these constructs internally. However, we found it more practical to import these
arithmetical facts and use them directly in the verification.\\

\noindent\textbf{Evaluation strategy}
Usually, computational type theories are agnostic to the evaluation strategy employed by the programming language, i.e.,
one can define a type theory with either call by value or call by name evaluation. However, 
when considering computational complexity, different evaluation strategies can lead to wildly different evaluation costs.
Often, call by name evaluation is favored because fewer rules are needed to specify the transition relation. This poses a barrier for 
complexity analysis, since if arbitrary computations can be applied to a function, it becomes difficult to prove uniform or 
general complexity properties about the function. Therefore, we opt for call by
value evaluation; as a result, variables only range
over \emph{values} (elements of \val) in open term of \pl{}.

\begin{figure}
\begin{mathpar}
\inferrule{
}{
    \final{\arrtimev{\isOf{a}{E_1}}{E_2}{E_3}}
}

\inferrule{
}{
    \final{\pity{\isOf{a}{E_1}}{E_2}}
}

\inferrule{
}{
    \final{\lam{x}{}{E}}
}

\inferrule{
}{
    \final{\nat}
}

\inferrule{
}{
    \final{\zero}
}

\inferrule{
    \final{E}
}{
    \final{\suc{E}}
}

\inferrule{
}{
    \final{\sigmaty{\isOf{a}{E_1}}{E_2}}
}

\inferrule{
    \final{E_1}\\
    \final{E_2}
}{
    \final{\pair{E_1}{E_2}}
}

\inferrule{
}{
    \final{\triv}
}

\inferrule{
}{
    \final{\eqty{E}{E_1}{E_2}}
}

\inferrule{
}{
    \final{\subsetty{\isOf{a}{E_1}}{E_2}}
}

\inferrule{
}{
    \final{\univ{i}}
}

\inferrule{
}{
    \final{\rel{r}{E_1}{E_2}}
}

\inferrule{
}{
    \final{\tern{r}{E_1}{E_2}{E_2}}
}

\inferrule{
    \step{E_1}{E_1'}
}{
    \step{\ap{E_1}{E_2}}{\ap{E_1'}{E_2}}
}

\inferrule{
    \final{E_1}\\
    \step{E_2}{E_2'}
}{
    \step{\ap{E_1}{E_2}}{\ap{E_1}{E_2'}}
}

\inferrule{
    \final{E_2}
}{
    \step{(\fun{f}{a}{E}) E_2}{[\fun{f}{a}{E}/f,E_2/a]E}
}

\inferrule{
    \step{E}{E'}
}{
    \step{\suc{E}}{\suc{E'}}
}

\inferrule{
    \step{E}{E'}
}{
    \step{\ifz{E}{E_0}{a}{E_1}}{\ifz{E'}{E_0}{a}{E_1}}
}

\inferrule{
}{
    \step{\ifz{\zero}{E_0}{a}{E_1}}{E_0}
}

\inferrule{
    \final{\suc{E}}
}{
    \step{\ifz{\suc{E}}{E_0}{a}{E_1}}{[E/a]E_1}
}

\inferrule{
    \step{E_1}{E_1'}
}{
    \step{\pair{E_1}{E_2}}{\pair{E_1'}{E_2}}
}

\inferrule{
    \final{E_1}\\
    \step{E_2}{E_2'}
}{
    \step{\pair{E_1}{E_2}}{\pair{E_1}{E_2'}}
}

\inferrule{
    \step{E}{E'}
}{
    \step{\fst{E}}{\fst{E'}}
}

\inferrule{
    \step{E}{E'}
}{
    \step{\snd{E}}{\snd{E'}}
}

\inferrule{
    \final{E_1}\\
    \final{E_2}
}{
    \step{\fst{\pair{E_1}{E_2}}}{E_1}
}

\inferrule{
    \final{E_1}\\
    \final{E_2}
}{
    \step{\snd{\pair{E_1}{E_2}}}{E_2}
}

\inferrule{
    \step{E_1}{E_1'}
}{
    \step{\letcst{E_1}{x}{E_2}}{\letcst{E_1'}{x}{E_2}}
}

\inferrule{
    \final{E_1}
}{
    \step{\letcst{E_1}{x}{E_2}}{[E_1/x]E_2}
}

\inferrule{
    \step{E}{E'}
}{
    \step{\cffone{f}{E}}{\cffone{f}{E'}}
}

\inferrule{
}{
    \step{\cffone{f}{\bar{m}}}{\overline{f(m)}}
}

\inferrule{
  \step{E_1}{E_1'}
}{
  \step{\cfftwo{f}{E_1}{E_2}}{\cfftwo{f}{E_1'}{E_2}}
}

\inferrule{
  \final{E_1}\\
  \step{E_2}{E_2'}
}{
  \step{\cfftwo{f}{E_1}{E_2}}{\cfftwo{f}{E_1}{E_2'}}
}

\inferrule{
}{
    \step{\cfftwo{f}{\bar{m}}{\bar{n}}}{\overline{f(m,n)}}
}
\end{mathpar}
\caption{Structural operational semantics of \pl{}. Note that $\bar{m} = \mathtt{suc}^m(\zero)$.}
\label{fig:opersem}
\end{figure}

\section{Meaning explanation}\label{sec:meaning}
Given a programming language, the meaning explanation propounded by Martin-Löf~\cite{MARTINLOF1982153} specifies the meaning 
of types and membership via the execution behavior of the underlying programs; we give a brief introduction here.
At the outset, there is the notion of \emph{canonical} types. A canonical type is defined by specifying its equal canonical elements.
This specification of equality is dependent on the canonical type in question, but in general it must be an 
equivalence relation on all members.
In its basic form, the meaning explanation consists of four judgments:
\begin{gather*}
\isTypeComp{A} \qquad \isComp{M}{A}\\
\eqType{A}{A'} \qquad \eqComp{M}{M'}{A}
\end{gather*}
The first states that $A$ evaluates to a value $A_0$, and that $A_0$ is a canonical type. 
Given that it is known that $A$ is a type (so that $A$ is known to evaluate to a canonical type $A_0$), 
the second states that $M$ evaluates to $M_0$ and $M_0$ is a canonical element of $A_0$.
The third states that $A$ and $A'$ evaluates to equal canonical types. Lastly, given that $A$ is a type, 
the fourth states that $M$ and $M'$ evaluates to equal canonical members of the canonical type that $A$ evaluates to.

Following the philosophy that type constructors merely internalize prior judgmental structure, we add two new judgment forms,
\isCComp{M}{A}{P} and \eqCComp{M}{M'}{A}{P}, which can be read as ``$M$ in $A$ with cost $P$'' and 
``$M$ equal to $M'$ in $A$ with cost $P$''. 
The first states that, presupposing $A$ is a type and $P$ is in \nat (by the way we specify $\nat$, this will mean $P$ evaluates to a numeral $\bar{p}$),
$M$ evaluates to $M_0$ in $c$ steps, $M_0$ is in $A$, and that $p \ge c$. Similarly, the second states that (with the same presuppositions)
$M$ evaluates to $M_0$ in $c$ steps, $M'$ evaluates to $M'_0$ in $d$ steps, $M_0$ and $M'_0$ are equal in $A$, and that $p \ge \max{(c,d)}$.
These two judgments express the meaning of the previous membership and equality judgments and in addition prescribe 
the evaluation costs of the subjects. Here, the evaluation costs of $A$ and $P$ are not mentioned, since we are analyzing the 
computational complexity of \emph{running code} and not its
specifications.\footnote{If one desires, it is also possible to analyze the
complexity of types by shifting the perspective and viewing them as elements of a universe.}

Note that a careful staging of definitions is needed to specify the cost-bounded judgmental forms. Prior to even stating 
\isCComp{M}{A}{P}, we need to know that $P$ is in \nat, and hence that $\nat$ is a type. This also means it is not possible to 
define the judgment \isComp{M}{A} in terms of \isCComp{M}{A}{P} by existentially quantifying $P$, which would lead to an infinite 
regress of presuppositions.

\section{Constructing a type system}\label{sec:construction}

The procedure given in the previous section of defining equal types and type membership 
is an instance of definition 
by induction-recursion, i.e., defining a structure inductively while simultaneously defining a function out of that structure.
However, for mechanization purposes, it is more practical to 
follow the approach given by Allen~\cite{Allen1987ANS} and define the judgments of type theory relative to a \emph{type system}. 
Structurally, a type system is a relation $\tau : \ps{\val \times \val \times \ps{\val \times \val}}$. 
The idea is that $A$ and $B$ are equal
types with membership relation $\phi$ exactly when $\tau(A,B,\phi)$. Observe that in the meaning explanation, 
1) type equality and membership are partial equivalence relations, and 2) equal types classify the same terms.
Hence, $\tau$ must satisfy the same coherence properties
in order to qualify as a type system, which can be summarized as follows.

\begin{definition}[Type Systems]
Write $\pts = \ps{\val \times \val \times \ps{\val \times \val}}$ for the space of all \emph{possible type systems}.
Let $\tau : \pts$. $\tau$ is a type system when the following conditions hold:
\begin{enumerate}
\item Unicity: $\tau(A,B,\phi)$ and $\tau(A,B,\phi')$ implies $\phi = \phi'$.
\item PER valuation: $\tau(A,B,\phi)$ implies $\phi$ is a PER.
\item Symmetry: $\tau(A,B,\phi)$ implies $\tau(B,A,\phi)$. 
\item Transitivity: $\tau(A,B,\phi)$ and $\tau(B,C,\phi)$ implies $\tau(A,C,\phi)$.
\end{enumerate}
\end{definition}

Now, to define a possible type system $\tau$ closed under type constructors and to show that it is a type system, we follow 
the construction presented by Angiuli~\cite{DBLP:journals/corr/abs-1712-01800} and define two hierarchies of possible type systems
$\tau_n$ and $v_n$, the latter of which forms a cumulative universe hierarchy.
Roughly, each $\tau_n$ contains all types from the previous level and the universes in $v_n$ and closes them 
under the type constructors, and $v_n$ contains all universes \univ{i} with $i < n$. 
In Figure~\ref{fig:construction}, 
we define a function $\text{Types}(v,\tau)$
that closes $\tau$ under type constructors and is monotone for fixed $v$. 
For convenience, we introduce some notations that lift the type equality and membership relations 
along evaluation and to open terms.

\begin{enumerate}
  \item $A \sim A' \downarrow \alpha \in \tau$ when \eval{A}{A_0}, \eval{A'}{A_0'}, and 
  $\tau(A_0,A_0',\alpha)$.
  \item \sameComp{}{M}{M'}{\alpha} when 
  \eval{M}{V}, \eval{M'}{V'}, and $\alpha(V, V')$.
  \item $a : \alpha \rhd B \sim B' \downarrow \beta \in \tau$ when 
  for all $V,V'$ s.t. $\alpha(V,V')$,  \sameType{[V/a]B}{[V'/a]B'}{\beta_{V,V'}}{\tau}. 
  \item \sameComp{\isOf{a}{\alpha}}{N}{N'}{\beta} when 
  for all $V,V'$ s.t. $\alpha(V,V')$, 
  \sameComp{}{[V/a]N}{[V'/a]N'}{\beta_{V,V'}}.
\end{enumerate}

Note that these definitions are slightly different from those introduced in the ``Idealized Nuprl'' construction 
by Angiuli~\cite{DBLP:journals/corr/abs-1712-01800}: in the last two definitions, the open terms are only required to respect
functionality for \emph{values}, i.e. elements of \val. As mentioned in Section~\ref{sec:setup}, 
this is because if arbitrary computations were also included,
substitution instances of those open terms will not, in general, have uniform or predictable evaluation costs.
We also define the following notations that are relevant for cost analysis:

\begin{enumerate}
\item $\omega = \mu \alpha.\, \{(\zero,\zero)\} \cup \{(\suc{v},\suc{v'}) \mid \alpha(v,v')\}$
  \item \sameCComp{}{M}{M'}{\alpha}{P}{} when 
  \begin{gather}
    \sameComp{}{P}{P}{\omega},\\
    \evalCost{M}{c}{V}, \evalCost{M'}{c'}{V'}, \alpha(V,V'),\\
    \eval{P}{\bar{p}}, \text{ and } p \ge \max(c,c')
  \end{gather}
  \item $\sameCComp{\isOf{a}{\alpha}}{N}{N'}{\beta}{P}{}$ when \sameComp{\isOf{a}{\alpha}}{P}{P}{\omega}
and for all $V,V'$ s.t. $\alpha(V,V')$, \sameCComp{}{[V/a]N}{[V'/a]N'}{\beta_{V,V'}}{[V/a]P}{}.
\end{enumerate}

Where we write $\mu \alpha. S (\alpha)$ for the least fixed point of a
monotone function $S : \ps{\val \times \val} \to \ps{\val \times \val}$.
Note that natural numbers are strict ($\omega$ contain only numerals), since we need to directly read off the numeral 
and compare it against the evaluation cost. The definition \sameCComp{}{M}{M'}{\alpha}{P}{} mirrors the 
judgmental form \eqCComp{M}{M'}{A}{P} mentioned in the previous section. The last notation 
\sameCComp{\isOf{a}{\alpha}}{N}{N'}{\beta}{P}{} expresses a ``complexity-aware'' functionality condition:
$N$ and $N'$ are functional in the relation $\alpha$ and in addition adhere to the evaluation cost bound $P$.

Furthermore, these definitions are closed under head expansion and respects equality of complexities: 
\begin{lemma}[Head expansion and replacement]\label{lemma:headexp1}
  \ 
  \begin{enumerate}
    \item If \sameComp{}{M}{M'}{\alpha} and \step{M''}{M}, then \sameComp{}{M''}{M'}{\alpha}.
    \item If \sameCComp{}{M}{M'}{\alpha}{P}, \step{M_1}{M}, and \step{M_2}{M'},  then \sameCComp{}{M_1}{M_2}{\alpha}{\suc{P}}.
    \item If \sameCComp{}{M}{M'}{\alpha}{P} and \step{M''}{M}, then \sameCComp{}{M''}{M'}{\alpha}{\suc{P}}.
    \item If \sameCComp{}{M}{M'}{\alpha}{P}, \stepIn{M_1}{c_1}{M}, \stepIn{M_2}{c_2}{M'}, \eval{Q}{\bar{q}}, and $q \ge \max(c_1,c_2)$,
     then \sameCComp{}{M_1}{M_2}{\alpha}{Q \hat{+} P}.
    \item If \sameCComp{}{M}{M'}{\alpha}{P} and \sameComp{}{P}{P'}{\omega}, then \sameCComp{}{M}{M'}{\alpha}{P'}.
  \end{enumerate}
\end{lemma}

\begin{proof}
  \ 
  \begin{enumerate}
    \item By assumption, \eval{M}{V}, \eval{M'}{V'}, and $\alpha(V,V')$. The result follows since \eval{M''}{V}. 
    \item By assumption, we know: 
    \begin{enumerate}
      \item \sameComp{}{P}{P}{\omega},
      \item \evalCost{M}{c}{V}, 
      \item \evalCost{M'}{c'}{V'}, 
      \item $\alpha(V,V')$,
      \item \eval{P}{\bar{p}}, 
      \item $p \ge \max{(c,c')}$
    \end{enumerate}
    First, we need to show \sameComp{}{\suc{P}}{\suc{P}}{\omega}. This holds since \eval{\suc{P}}{\suc{\bar{p}}} and 
    $\omega(\bar{p},\bar{p})$. Since \evalCost{M_1}{c+1}{V} and \evalCost{M_2}{c'+1}{V'}, it suffices to show 
    $p+1 \ge \max(c+1,c'+1)$, which holds.
    \item Similar to above.
    \item By assumption, we know: 
    \begin{enumerate}
      \item \sameComp{}{P}{P}{\omega},
      \item \evalCost{M}{c}{V}, 
      \item \evalCost{M'}{c'}{V'}, 
      \item $\alpha(V,V')$,
      \item \eval{P}{\bar{p}}, 
      \item $p \ge \max{(c,c')}$
    \end{enumerate} 
    Hence \evalCost{M_1}{c_1 + c}{V}, \evalCost{M_2}{c_2+c'}{V'}, and \eval{Q \hat{+} P}{\overline{q + p}}, and 
    it suffices to show $q + p \ge \max{c_1+c,c_2+c'}$, which holds.
    \item By assumption, \eval{P}{\bar{p}}, \eval{P}{\bar{p'}}, and $\omega(\bar{p},\bar{p'})$. 
    It suffices to show $p = p'$, which holds by definition of $\omega$.
    \qedhere
  \end{enumerate}
\end{proof}

\begin{figure}
\small
\begin{align*}
  \text{Nat} &= \{(\nat,\nat,\phi) \mid \phi = \omega\}\\
  \text{Funtime}(\tau) &=
  \{(\arrtimev{\isOf{a}{A}}{B}{P}, \arrtimev{\isOf{a}{A'}}{B'}{P'}, \phi) 
  \mid & \\
  &\exists \alpha.\, \sameType{A}{A'}{\alpha}{\tau} \\
  &\land\exists \beta.\, \sameTypeOne{a}{\alpha}{B}{B'}{\beta}{\tau} \\
  &\land\sameComp{\isOf{a}{\alpha}}{P}{P'}{\omega} \\
  \land\phi &= \{(\fun{f}{a}{N},\fun{f'}{a'}{N'}) \mid \sameCComp{\isOf{a}{\alpha}}{[\fun{f}{a}{N}/f]N}{[\fun{f}{a}{N'}/f]N'}{\beta}{P}{} \}
  \}\\
  \text{Sigma}(\tau) &=
  \{(\sigmaty{A}{a.B}, \sigmaty{A'}{a.B'}, \phi) 
  \mid & \\
  &\exists \alpha.\, \sameType{A}{A'}{\alpha}{\tau} \\
  &\land\exists \beta.\, \sameTypeOne{a}{\alpha}{B}{B'}{\beta}{\tau} \\
  \land\phi &= \{(\pair{U}{V}, \pair{U'}{V'}) \mid \alpha(U,U') \land \beta_{U,U'}(V,V') \}
  \}\\
  \text{Eq}(\tau) &= \{
    (\eqty{A}{M}{N}, \eqty{A'}{M'}{N'}, \phi) \mid \\
  &\exists \alpha.\, \sameType{A}{A'}{\alpha}{\tau}\\
  &\land\sameComp{}{M}{M'}{\alpha}\\
  &\land\sameComp{}{N}{N'}{\alpha}\\
  \land\phi &= \{(\triv, \triv) \mid \sameComp{}{M}{N}{\alpha}
  \}
  \}\\
  \text{Subset}(\tau) &= \{
    (\subsetty{\isOf{x}{A}}{B}, \subsetty{\isOf{x}{A'}}{B'},\phi) \mid \\
    &\exists \alpha.\, \sameType{A}{A'}{\alpha}{\tau} \\
    &\land\exists \beta.\, \sameTypeOne{a}{\alpha}{B}{B'}{\beta}{\tau} \\
    \land\phi &= \{(V,V') \mid \alpha(V,V') \land \exists U,U'.\beta_{V,V'}(U,U')\}
  \}\\
  \text{Rel2}(\tau) &= \{
    (\rel{r}{M}{N}, \rel{r}{M'}{N'}, \phi) \mid \\
    &r : \mathcal{P}(\N \times \N) \land\sameComp{}{M}{M'}{\omega} \land\sameComp{}{N}{N'}{\omega}\\ 
    \land\phi &= \{ (\triv,\triv) \mid \exists m,n.\, r(m,n) \land \sameComp{}{M}{\bar{m}}{\omega} \land \sameComp{}{N}{\bar{n}}{\omega} \}
    \}\\
   \text{Rel3}(\tau) &= \{
    (\tern{r}{M}{N}{O}, \tern{r}{M'}{N'}{O'}, \phi) \mid \\
    &r : \mathcal{P}(\N \times \N \times \N) \land\sameComp{}{M}{M'}{\omega} \land \sameComp{}{N}{N'}{\omega} \land \sameComp{}{O}{O'}{\omega}\\ 
    \land\phi &= \{ (\triv,\triv) \mid \exists m,n,o.\, r(m,n,o) \land \sameComp{}{M}{\bar{m}}{\omega} \land \sameComp{}{N}{\bar{n}}{\omega} \land \sameComp{}{O}{\bar{o}}{\omega}\}
    \}\\
  \text{Univ}(\tau) &= \{(\univ{i}, \univ{i}, \phi) \mid \exists \phi.\, \sameType{\univ{i}}{\univ{i}}{\phi}{\tau} \}\\
  v_n &= \{(\univ{i}, \univ{i}, \phi) \mid i < n \land \phi = \{(A,B) 
  \mid \exists \alpha.\, \tau_i(A,B,\alpha) \}\}\\
  \text{Types}(v,\tau) &= \text{Nat} \cup \text{Fun}(\tau) \cup \text{Funtime}(\tau) 
    \cup \text{Sigma}(\tau) \cup \text{Eq}(\tau) \cup \text{Subset}(\tau) \cup \text{Rel2}(\tau) \cup \text{Rel3}(\tau)\cup \text{Univ}(v)\\
  \tau_n &= \mu \tau.\, \text{Types}(v_n,\tau)\\
  v_{\omega} &= \{(\univ{i}, \univ{i}, \phi) \mid \phi = \{(A,B) 
  \mid \exists \alpha.\, \tau_i(A,B,\alpha)\}\}\\
  \tau_{\omega} &= \mu \tau.\, \text{Types}(v_{\omega},\tau)\\
\end{align*}
\caption{Monotone operators on \pts{} and type system hierarchies.}
\label{fig:construction}
\end{figure}

\begin{theorem}[Monotonicity]\label{lemma:mono}
  \ 
  \begin{enumerate}
    \item For all $v : \pts$, $\text{Types}(v,\_)$ is monotone.
    \item For all $\tau : \pts$, $\text{Types}(\_,\tau)$ is monotone.
  \end{enumerate}
\end{theorem}

\begin{proof}
  By case analysis on each clause of Types.
\end{proof}

\begin{lemma}[Closure under Types]\label{lemma:TYPES_TS}
  If $v : \pts$ is a type system, then $\mu \tau.\,\text{Types}(v,\tau)$ is also a type system.
\end{lemma}

\begin{proof}
  By case analysis on each clause of Types.
\end{proof}

\begin{lemma}[Empty type system]\label{lemma:empty_TS}
  $\emptyset : \pts$ is a type system.
\end{lemma}

\begin{proof}
  The criteria for type system hold vacuously.
\end{proof}

\begin{lemma}\label{lemma:tau0_TS}
  $\tau_0$ is a type system.
\end{lemma}

\begin{proof}
  By Lemma~\ref{lemma:empty_TS} and~\ref{lemma:TYPES_TS}.
\end{proof}

\begin{theorem}[Hierarchy]\label{lemma:hierachy}
  For all $n : \N$, $v_n$ and $\tau_n$ are type systems. 
\end{theorem}

\begin{proof}
  By strong induction on $n$, using Lemma~\ref{lemma:empty_TS} and~\ref{lemma:tau0_TS} in the base case
  and Lemma~\ref{lemma:TYPES_TS} in the inductive case.
\end{proof}

\begin{theorem}\label{lemma:vomega_TS}
  $v_{\omega}$ is a type system.
\end{theorem}

\begin{proof}
  Unicity, symmetry, and transitivity are immediate; per valuation follows from Lemma~\ref{lemma:hierachy}.
\end{proof}

\begin{theorem}\label{lemma:tauomega_TS}
  $\tau_{\omega}$ is a type system.
\end{theorem}

\begin{proof}
  By Lemma~\ref{lemma:vomega_TS} and~\ref{lemma:TYPES_TS}.
\end{proof}

\begin{lemma}\label{lemma:fp_le}
  Let $X$ be a complete lattice and $f,g : X \to X$ be monotone functions. If 
  $f(x) \subseteq g(x)$ for all $x \in X$, then $\mu x.\,f(x) \subseteq \mu x.\,g(x)$.
\end{lemma}

\begin{proof}
  $\mu x.\,g(x)$ is $f$-closed because $f(\mu x.\,g(x)) \subseteq g(\mu x.\,g(x)) = \mu x.\,g(x)$. 
  Since $\mu x.\, f(x)$ is the least $f$-closed point, $\mu x.\, f(x) \subseteq \mu x.\, g(x)$.
\end{proof}

\begin{lemma}[Cumulativity]\label{lemma:cumulativity}
  \ 
  \begin{enumerate}
    \item $v_i \subseteq v_{\omega}$ for all $i$. 
    \item $\tau_i \subseteq \tau_{\omega}$ for all $i$. 
  \end{enumerate}
\end{lemma}

\begin{proof}
  The first part holds by construction. 
  For the latter, by Lemma~\ref{lemma:fp_le}, it suffices to show 
  $\text{Types}(v_i,\tau) \subseteq \text{Types}(v_{\omega},\tau)$ for all $\tau$.
  It suffices to show $\text{Types}(\_,\tau)$ is monotone, 
  which holds by Lemma~\ref{lemma:mono}.
\end{proof}
\section{Computational Type Theory}~\label{sec:cctt}
\textbf{Closed Judgments} 
Having defined a type system $\tau_{\omega}$, we are now in a position to define the judgments of \cctt{}. 
In addition to the usual type equality and 
membership relation, we define a third membership relation which bounds the evaluation cost.

\begin{definition}[Types and Equality]
  Given a type system $\tau$, 
  \begin{enumerate}
    \item \eqType{A}{B} when \sameType{A}{B}{\alpha}{\tau} for some $\alpha$.
    \item \eqComp{M}{M'}{A} when 
      \sameType{A}{A}{\alpha}{\tau} and \sameComp{}{M}{M'}{\alpha} for some $\alpha$.
    \item \eqCComp{M}{M'}{A}{P}{} when 
      \sameType{A}{A}{\alpha}{\tau} and \sameCComp{}{M}{M'}{\alpha}{P}{} for some $\alpha$.
    \item \true{A}{B} when \sameType{A}{B}{\alpha}{\tau} and \sameComp{}{M}{M}{\alpha} for some $\alpha$ and $M$.
  \end{enumerate}
\end{definition}

\begin{remark}
  This definition is generic in a type system $\tau$. However, from now on we will implicitly assume 
  $\tau = \tau_{\omega}$ unless specified otherwise.
\end{remark}

The reflexive case for the 3 relations will be abbreviated as follows: 
\begin{enumerate}
  \item \isTypeComp{A} means \eqType{A}{A}.
  \item \isComp{M}{A} means \eqComp{M}{M}{A}.
  \item \isCComp{M}{A}{P} means \eqCComp{M}{M}{A}{P}.
\end{enumerate}

Furthermore, when an expression is a value, membership will be subscripted with a 0:
\eqVal{M}{M'}{A} when \final{M}, \final{M'}, and \eqComp{M}{M'}{A}, and \isVal{M}{A} when \final{M} and \isComp{M}{A}.

The closed judgments are PERs: 
\begin{lemma}[Closed PER]\label{lemma:per}
  \ 
\begin{enumerate}
  \item \eqType{\_}{\_} is a PER.
  \item \eqComp{\_}{\_}{A} is a PER. 
  \item \eqCComp{\_}{\_}{A}{P} is a PER.
\end{enumerate}
\end{lemma}

\begin{proof}
  Immediate consequence of the fact that $\tau_{\omega}$ is a type system.
\end{proof}

\begin{lemma}\label{lemma:resp_eq}
  Given \eqType{A}{A'}, 
  \begin{enumerate}
    \item If \eqComp{M}{M'}{A}, then \eqComp{M}{M'}{A'}.
    \item If \eqCComp{M}{M'}{A}{P}, then \eqCComp{M}{M'}{A'}{P}.
  \end{enumerate}
\end{lemma}

\begin{proof}
  Part 1) and 2) are similar, so we just show 1). Suppose \eqComp{M}{M'}{A}. This means
  \sameType{A}{A}{\alpha}{\tau_{\omega}} for some $\alpha$ and
  \sameComp{}{M}{M'}{\alpha}. We need to show \eqComp{M}{M'}{A'}.
  By assumption, \sameType{A}{A'}{\alpha'}{\tau_{\omega}} for some $\alpha'$. 
  By symmetry and transitivity, \sameType{A}{A}{\alpha'}{\tau_{\omega}}, so by
  unicity, $\alpha = \alpha'$. Similarly,
  \sameType{A'}{A'}{\alpha'}{\tau_{\omega}}, 
  and \sameComp{M}{M'}{\alpha'} by definition, so the result holds.
\end{proof}

\noindent\textbf{Head expansion} The closed judgments are closed under head expansion, and furthermore is a congruence with respect to complexity: 
\begin{lemma}[Head expansion and replacement]\label{lemma:headexp}
  \ 
  \begin{enumerate}
    \item If \eqType{A}{A'} and \step{A''}{A}, then \eqType{A''}{A}.
    \item If \eqComp{M}{M'}{A} and \step{M''}{M}, then \eqComp{M''}{M'}{A}.
    \item If \eqCComp{M}{M'}{A}{P} and \step{M''}{M}, then \eqCComp{M''}{M'}{A}{\suc{P}}.
    \item If \eqCComp{M}{M'}{A}{P}, \stepIn{M_1}{c_1}{M}, \stepIn{M_2}{c_2}{M'}, \eval{Q}{\bar{q}}, and $q \ge \max(c_1,c_2)$,
     then \eqCComp{M_1}{M_2}{A}{Q \hat{+} P}.
    \item If \eqCComp{M}{M'}{A}{P} and \eqComp{P}{P'}{\nat}, then \eqCComp{M}{M'}{A}{P'}.
  \end{enumerate}
\end{lemma}

\begin{proof}
  \ 
  \begin{enumerate}
    \item By assumption, \sameType{A}{A'}{\alpha}{\tau_{\omega}} for some $\alpha$. So \eval{A}{V}, \eval{A'}{V'}, 
    and $\tau_{\omega}(V,V',\alpha)$. Since $\_ \mapsto  \_$ is deterministic, \eval{A''}{V}, and the result holds.
    \item By assumption, \sameComp{}{M}{M'}{\alpha} for some $\alpha$ where \sameType{A}{A}{\alpha}{\tau_{\omega}}.
    It suffices to show \sameComp{}{M''}{M'}{\alpha}. By assumption, \eval{M}{V}, \eval{M'}{V'}, and $\alpha(V,V')$.
    Hence \eval{M''}{V}, and we are done. 
    \item By assumption, \sameComp{}{M}{M'}{\alpha}{P} for some $\alpha$ where \sameType{A}{A}{\alpha}{\tau_{\omega}}.
    This means that 
    \begin{enumerate}
      \item \evalCost{M}{c}{V}
      \item \evalCost{M'}{c'}{V'}
      \item \eval{P}{\bar{p}}
      \item $p \ge \max{(c,c')}$
      \item $\alpha(V,V')$
    \end{enumerate}
    It suffices to show \sameCComp{}{M''}{M'}{\alpha}{\suc{P}}. Hence we have \evalCost{M''}{c+1}{V} and 
    \eval{\suc{P}}{\bar{p+1}}. It suffices to show $p+1 \ge \max{(c+1,c')}$, which holds.
    \item Similar to above.
    \item By assumption, \sameComp{}{P}{P'}{\omega}, which means \eval{P}{\bar{p}}, \eval{P'}{\bar{p'}}, and 
    $\omega(\bar{p},\bar{p'})$. It suffices to show $p = p'$, which holds by definition of $\omega$.
  \end{enumerate}
\end{proof}

\noindent\textbf{Open Judgments}
To extend the closed judgments to open terms, we define \tel{\Gamma}, which represent
telescopes. 

\begin{definition}[Telescopes]
  \ 
  \begin{enumerate}
  \item \tel{\cdot}.
  \item \tel{\Gamma,\isOf{x}{A}} when \tel{\Gamma} and $A \in \mathsf{Exp}$.
  \end{enumerate}
\end{definition}

Open judgments are given meaning by functionality: open types are equal if all equal substitution instances are equal as 
closed types, and similarly for membership. As before, substitution instances range over values. 

\begin{definition}[Equal Instances]
  Given \tele{\Gamma}, define \eqInst{\gamma}{\gamma'}{\Gamma}:
  \begin{enumerate}
    \item \eqInst{\cdot}{\cdot}{\cdot}.
    \item \eqInst{\gamma,[a \mapsto V]}{\gamma',[a \mapsto V']}{\Gamma,\isOf{a}{A}} when
      \eqVal{V}{V'}{\hat\gamma A} and \eqInst{\gamma}{\gamma'}{\Gamma}.
  \end{enumerate}
\end{definition}

\begin{definition}[Open Judgments]
  \ 
  \begin{enumerate}
    \item \openEqType{\Gamma}{A}{A'} when for all \eqInst{\gamma}{\gamma'}{\Gamma}, 
      \eqType{\hat\gamma A}{\hat\gamma' A'}.
    \item \openEqComp{\Gamma}{M}{M'}{A} when for all \eqInst{\gamma}{\gamma'}{\Gamma}, 
      \eqComp{\hat\gamma M}{\hat\gamma' M'}{\hat\gamma A}.
    \item \openEqCComp{\Gamma}{M}{M'}{A}{P} when for all \eqInst{\gamma}{\gamma'}{\Gamma}, 
      \eqCComp{\hat\gamma M}{\hat\gamma' M'}{\hat\gamma A}{\hat\gamma P}.
    \item \openEqVal{\Gamma}{V}{V'}{A} when for all \eqInst{\gamma}{\gamma'}{\Gamma}, 
      \eqVal{\hat\gamma V}{\hat\gamma' V'}{\hat\gamma A}.
  \end{enumerate}
\end{definition}

\begin{proposition}[Conversion]
  \openEqVal{\Gamma}{V}{V'}{A} $\iff$ \openEqCComp{\Gamma}{V}{V'}{A}{\zero}.
\end{proposition}

In order to prove some structural lemmas about the open judgments (such as symmetry and transitivity), we need to 
restrict $\Gamma$ to telescopes that are also \emph{contexts}: 

\begin{definition}[Contexts]
  Given \tele{\Gamma}, and \tele{\Gamma'}, define \eqCtx{\Gamma}{\Gamma'}:
  \begin{enumerate}
    \item \eqCtx{\cdot}{\cdot}.
    \item \eqCtx{\Gamma,\isOf{a}{A}}{\Gamma',\isOf{a}{A'}} when
      \openEqType{\Gamma}{A}{A'} and \eqCtx{\Gamma}{\Gamma'}.
  \end{enumerate}
\end{definition}

\begin{lemma}[PER of contexts]\label{lemma:per_context}
  Given \eqCtx{\Gamma}{\Gamma}, 
  \begin{enumerate}
    \item If \eqInst{\gamma_1}{\gamma_2}{\Gamma}, then \eqInst{\gamma_2}{\gamma_1}{\Gamma}.
    \item If \eqInst{\gamma_1}{\gamma_2}{\Gamma} and \eqInst{\gamma_2}{\gamma_3}{\Gamma},
    then \eqInst{\gamma_1}{\gamma_3}{\Gamma}.
  \end{enumerate}
\end{lemma}

\begin{proof}
  Part 1): induction on the length of $\Gamma$. If $\Gamma = \cdot$, then the
  result follows by definition. Otherwise, $\Gamma = \Gamma',\isOf{a}{A}$, and suppose \eqInst{\gamma_1,[a \mapsto
  V_1]}{\gamma_2,[a \mapsto V_2]}{\Gamma',\isOf{a}{A}}. We need to show 
  \eqInst{\gamma_2,[a \mapsto V_2]}{\gamma_1,[a \mapsto V_1]}{\Gamma',\isOf{a}{A}}. By definition, it
  suffices to show 
  \begin{enumerate}
    \item \eqVal{V_2}{V_1}{\hat\gamma_2 A} 
    \item \eqInst{\gamma_2}{\gamma_1}{\Gamma'}
  \end{enumerate}
  By the assumption that \eqCtx{\Gamma}{\Gamma}, we have
  \openEqType{\Gamma'}{A}{A} and \eqCtx{\Gamma'}{\Gamma'}. By assumption, we
  have \eqVal{V_1}{V_2}{\hat\gamma_1 A} and
  \eqInst{\gamma_1}{\gamma_2}{\Gamma'}. Hence the second obligation follows by
  induction. The first obligation follows from Lemma~\ref{lemma:per} and~\ref{lemma:resp_eq}, given that
  we show \eqType{\hat\gamma_1 A}{\hat\gamma_2 A}, which in turn follows by
  definition of \openEqType{\Gamma'}{A}{A}.

  Part 2): induction on the length of $\Gamma$. If $\Gamma = \cdot$, then the
  result is immediate. Otherwise, $\Gamma = \Gamma',\isOf{a}{A}$. Suppose 
  \eqInst{\gamma_1,[a\mapsto V_1]}{\gamma_2,[a \mapsto
  V_2]}{\Gamma',\isOf{a}{A}} and \eqInst{\gamma_2,[a\mapsto V_2]}{\gamma_3,[a \mapsto
  V_3]}{\Gamma',\isOf{a}{A}}. We need to show \eqInst{\gamma_1,[a\mapsto V_1]}{\gamma_3,[a \mapsto
  V_3]}{\Gamma',\isOf{a}{A}}. By definition, it suffices to show
  \begin{enumerate}
    \item \eqVal{V_1}{V_3}{\hat\gamma_1 A}
    \item \eqInst{\gamma_1}{\gamma_3}{\Gamma'}
  \end{enumerate}
  By the assumptions, we have 
  \begin{enumerate}
    \item \eqCtx{\Gamma'}{\Gamma'}
    \item \openEqType{\Gamma'}{A}{A}
    \item \eqVal{V_1}{V_2}{\hat\gamma_1 A}
    \item \eqInst{\gamma_1}{\gamma_2}{\Gamma'}
    \item \eqVal{V_2}{V_3}{\hat\gamma_2 A}
    \item \eqInst{\gamma_2}{\gamma_3}{\Gamma'}
  \end{enumerate}
  Hence we have \eqType{\hat\gamma_1 A}{\hat\gamma_2 A}. By
  Lemma~\ref{lemma:resp_eq}, this implies \eqVal{V_2}{V_3}{\hat\gamma_1 A}, and
  by Lemma~\ref{lemma:per} we get \eqVal{V_1}{V_3}{\hat\gamma_1 A}. Since the latter obligation follows by
  induction, the result holds.
\end{proof}

\begin{lemma}[PER of open judgments]\label{lemma:per_open}
  Given \eqCtx{\Gamma}{\Gamma}, 
  \begin{enumerate}
    \item \openEqType{\Gamma}{\_}{\_} is a PER.
    \item If \openEqType{\Gamma}{A}{A}, then \openEqComp{\Gamma}{\_}{\_}{A} is a PER.
    \item If \openEqType{\Gamma}{A}{A} and \openEqComp{\Gamma}{P}{P}{\nat}, then \openEqCComp{\Gamma}{\_}{\_}{A}{P} is a PER.
  \end{enumerate}
\end{lemma}

\begin{proof}
  \ 
  For part 1), suppose \openEqType{\Gamma}{A_1}{A_2}. We need to show \openEqType{\Gamma}{A_2}{A_1}. 
  Let \eqInst{\gamma}{\gamma'}{\Gamma}. It suffices to show \eqType{\hat\gamma A_2}{\hat\gamma' A_1}. 
  By Lemma~\ref{lemma:per_context}, \eqInst{\gamma'}{\gamma}{\Gamma}, so by assumption, 
  \eqType{\hat\gamma' A_1}{\hat\gamma A_2}. Hence the result follows by Lemma~\ref{lemma:per}.

  Now suppose \openEqType{\Gamma}{A_1}{A_2} and \openEqType{\Gamma}{A_2}{A_3}. We need to show 
  \openEqType{\Gamma}{A_1}{A_3}.  
  Let \eqInst{\gamma}{\gamma'}{\Gamma}. It suffices to show \eqType{\hat\gamma A_1}{\hat\gamma' A_3}. 
  By Lemma~\ref{lemma:per_context}, \eqInst{\gamma}{\gamma}{\Gamma}, so by assumption, 
  \eqType{\hat\gamma A_1}{\hat\gamma A_2} and \eqType{\hat\gamma A_2}{\hat\gamma' A_3}. Hence the 
  result follows by Lemma~\ref{lemma:per}. Part 2) - 3) follow similarly.
\end{proof}

Head expansion can be extended to open judgments. Write \openStep{\Gamma}{E}{E'} if for all \eqInst{\gamma}{\gamma}{\Gamma},
\stepIn{\hat\gamma E}{*}{\hat\gamma E'}.

\begin{lemma}[Open head expansion]\label{lemma:open_headexp}
  Given \eqCtx{\Gamma}{\Gamma}, 
  \begin{enumerate}
    \item If \openEqType{\Gamma}{A}{A'} and \openStep{\Gamma}{A''}{A}, then \openEqType{\Gamma}{A''}{A'}.
    \item If \openEqComp{\Gamma}{M}{M'}{A} and \openStep{\Gamma}{M''}{M}, then \openEqComp{\Gamma}{M''}{M'}{A}.
  \end{enumerate}
\end{lemma}

\begin{proof}
  \ 
  \begin{enumerate}
    \item Let \eqInst{\gamma}{\gamma'}{\Gamma}. We need to show \eqType{\hat\gamma A''}{\hat\gamma' A'}. 
    By assumption, we know \eqType{\hat\gamma A}{\hat\gamma' A'}. Since \eqCtx{\Gamma}{\Gamma}, 
    \eqInst{\gamma}{\gamma}{\Gamma}, so we have \stepIn{\hat\gamma A''}{*}{\hat\gamma A}, 
    so the result follows from Lemma~\ref{lemma:headexp}.
    \item Let \eqInst{\gamma}{\gamma'}{\Gamma}. We need to show \eqComp{\hat\gamma M''}{\hat\gamma' M'}{\hat\gamma A}. 
    By assumption, we know \eqComp{\hat\gamma M}{\hat\gamma' M'}{\hat\gamma A}.
    Since \eqCtx{\Gamma}{\Gamma}, \eqInst{\gamma}{\gamma}{\Gamma}, so we have 
    \stepIn{\hat\gamma M''}{*}{\hat\gamma M}, 
    so the result follows from Lemma~\ref{lemma:headexp}.
  \end{enumerate}
\end{proof}

Open judgments admit the structural rules of identity, weakening, and substitution: 

\begin{lemma}[Hypothesis]\label{lemma:hypothesis}
  \openCComp{\Gamma,\isOf{a}{A}}{a}{A}{\zero}.
\end{lemma}

\begin{proof}
  Let \eqInst{\gamma,[a\mapsto V]}{\gamma',[a \mapsto V']}{\Gamma,\isOf{a}{A}}. We need to show 
  \eqCComp{V}{V'}{\hat\gamma A}{\zero}, which holds by definition of \eqInst{\gamma,[a\mapsto V]}{\gamma',[a \mapsto V']}{\Gamma,\isOf{a}{A}}. 
\end{proof}

\begin{lemma}[Weakening]\label{lemma:weaken}
  \ 
  \begin{enumerate}
    \item If \openEqType{\Gamma}{B}{B'}, then \openEqType{\Gamma,\isOf{a}{A}}{B}{B'}.
    \item If \openEqComp{\Gamma}{M}{M'}{B}, then \openEqComp{\Gamma,\isOf{a}{A}}{M}{M'}{B}. 
    \item If \openEqCComp{\Gamma}{M}{M'}{B}{P}, then \openEqCComp{\Gamma,\isOf{a}{A}}{M}{M'}{B}{P}. 
  \end{enumerate} 
\end{lemma}

\begin{proof}
  For part 1, let \eqInst{\gamma,[a \mapsto V]}{\gamma', [a \mapsto V']}{\Gamma,\isOf{a}{A}}. We need to show 
  \eqType{\hat\gamma B}{\hat\gamma B'}. By definition, \eqInst{\gamma}{\gamma'}{\Gamma}, so the result follows 
  from assumption. Part 2 and 3 follow similarly. 
\end{proof}

For sequencing operations occurring in types and cost specifications, we write
$\lsub{E_1}{a}{E_2}$ for \letcst{E_1}{a}{E_2}.
\begin{remark}
  ``leftist'' substitution is required because untyped principles
  such as direct computation~\cite{Howe-Equalityinlazycompu} does not hold in \cctt{}.
  In contrast to \nuprl{}, types can internalize
  cost-aware judgments, and so in general membership is not invariant under 
  untyped rewrite rules.
\end{remark}

\begin{lemma}[Substitution]\label{lemma:seq}
  \ 
  \begin{enumerate}
    \item If \openEqComp{\Gamma}{M}{M'}{A} and \openEqType{\Gamma,\isOf{a}{A}}{B}{B'}, then \openEqType{\Gamma}{\lsub{M}{a}{B}}{\lsub{M'}{a}{B'}}.
    \item If \openEqComp{\Gamma}{M}{M'}{A} and \openEqComp{\Gamma,\isOf{a}{A}}{N}{N'}{B}, then \openEqComp{\Gamma}{\letcst{M}{a}{N}}{\letcst{M'}{a}{N'}}{\lsub{M}{a}{B}}.
    \item If \openEqCComp{\Gamma}{M}{M'}{A}{P} and \openEqCComp{\Gamma,\isOf{a}{A}}{N}{N'}{B}{Q}, then \openEqCComp{\Gamma}{\letcst{M}{a}{N}}{\letcst{M'}{a}{N'}}{\lsub{M}{a}{B}}{P \hat{+} \suc{\lsub{M}{a}{Q}}}.
  \end{enumerate}
\end{lemma}

\begin{proof}
  \ 
  \begin{enumerate}
    \item Let \eqInst{\gamma}{\gamma'}{\Gamma}. We need to show \eqType{\hat\gamma \lsub{M}{a}{B}}{\hat\gamma' \lsub{M'}{a}{B'}}.
    By assumption, we have \eqComp{\hat\gamma M}{\hat\gamma' M'}{\hat\gamma A}. This means that 
    \begin{enumerate}
      \item \eval{\hat\gamma M}{V}
      \item \eval{\hat\gamma' M'}{V'} 
    \end{enumerate}
    and \eqVal{V}{V'}{\hat\gamma A}. By definition, \eqInst{\gamma, [a \mapsto V]}{\gamma', [a \mapsto V']}{\Gamma,\isOf{a}{A}}, 
    so by assumption, \eqType{\hat\gamma [V/a]B}{\hat\gamma [V'/a]B'}, and the result follows from Lemma~\ref{lemma:headexp}.
    \item Let \eqInst{\gamma}{\gamma'}{\Gamma}. We need to show \eqComp{\hat\gamma\letcst{M}{a}{N}}{\hat\gamma'\letcst{M'}{a}{N'}}{\hat\gamma\lsub{M}{a}{B}}.
    By assumption, we have \eqComp{\hat\gamma M}{\hat\gamma' M'}{\hat\gamma A}. This means that 
    \begin{enumerate}
      \item \eval{\hat\gamma M}{V}
      \item \eval{\hat\gamma' M'}{V'} 
    \end{enumerate}
    and \eqVal{V}{V'}{\hat\gamma A}. By definition, \eqInst{\gamma, [a \mapsto V]}{\gamma', [a \mapsto V']}{\Gamma,\isOf{a}{A}},
    so by assumption, \eqComp{\hat\gamma[V/a]N}{\hat\gamma'[V'/a]N'}{\hat\gamma[V/a]B}, and the result follows from Lemma~\ref{lemma:headexp}.
    \item Let \eqInst{\gamma}{\gamma'}{\Gamma}. 
    We need to show \eqCComp{\hat\gamma\letcst{M}{a}{N}}{\hat\gamma'\letcst{M'}{a}{N'}}{\hat\gamma\lsub{M}{a}{B}}{\hat\gamma P \hat{+} \suc{\lsub{M}{a}{Q}}}.
    By assumption, we have \eqCComp{\hat\gamma M}{\hat\gamma' M'}{\hat\gamma A}{\hat\gamma P}. This means that 
    \begin{enumerate}
      \item \evalCost{\hat\gamma M}{c}{V}
      \item \evalCost{\hat\gamma' M'}{c'}{V'} 
      \item \eval{\hat\gamma P}{\bar{p}}
      \item $p \ge \max(c,c')$
    \end{enumerate}
    Hence \eqVal{V}{V'}{\hat\gamma A}. By definition, \eqInst{\gamma, [a \mapsto V]}{\gamma', [a \mapsto V']}{\Gamma,\isOf{a}{A}}, 
    so by assumption, \eqCComp{\hat\gamma [V/a]N}{\hat\gamma' [V'/a]N'}{\hat\gamma [V/a] B}{\hat\gamma [V/a]Q}.
    By Lemma~\ref{lemma:headexp}, we have 
    \eqCComp{\hat\gamma [V/a]N}{\hat\gamma' [V'/a]N'}{\hat\gamma \lsub{M}{a}{B}}{\hat\gamma \lsub{M}{a}{Q}}.
    By Lemma~\ref{lemma:headexp}, we have 
    \eqCComp{\hat\gamma \letcst{M}{a}{N}}{\hat\gamma' \letcst{M'}{a}{N'}}{\hat\gamma \lsub{M}{a}{B}}{\hat\gamma P \hat{+} \suc{\hat\gamma \lsub{M}{a}{Q}}},
    as required.
  \end{enumerate}
\end{proof}

\noindent\textbf{Semantics of types} Since type systems satisfy unicity, it is sometimes useful to refer to the 
denotation of types, i.e. the unique PER designated as its membership relation. 

\begin{lemma}[Closed denotation]\label{lemma:closed_sem}
  If \isType{A}, then there is a unique $\alpha$ such that \sameType{A}{A}{\alpha}{\tau_{\omega}}. Write 
  this as $\alpha = \sem{A}$.
\end{lemma}

\begin{proof}
  By unicity of $\tau_{\omega}$.
\end{proof}

\begin{lemma}[Open denotation]\label{lemma:open_sem}
  Given \sameType{A}{A'}{\alpha}{\tau_{i}} and \openEqComp{\isOf{a}{A}}{B}{B'}{\univ{i}}, there is a $\beta$ such that
  \sameTypes{\isOf{a}{\alpha}}{B}{B'}{\beta}{\tau_{i}}. Write this as $\beta = \sem{\isOf{a}{\alpha}.(B,B')}$.
\end{lemma} 

\begin{proof}
  First, we need to construct $\beta$. Let $V,V'$ such that $\alpha(V,V')$. By assumption, we have 
  \eqComp{[V/a]B}{[V'/a]B'}{\univ{i}}. This means \sameType{[V/a]B}{[V'/a]B'}{\phi}{\tau_{i}} for some $\phi$.
  Let $\beta = (V,V') \mapsto \phi$. We need to show \sameTypes{\isOf{a}{\alpha}}{B}{B'}{\beta}{\tau_{i}}. 
  Let $V,V'$ such that $\alpha(V,V')$. It suffices to show \sameType{[V/a]B}{[V'/a]B'}{\beta_{V,V'}}{\tau_{i}}, 
  which holds by definition of $\beta$.
\end{proof}

\begin{lemma}[Functionality]\label{lemma:func_per}
  If \sameTypes{\isOf{a}{\alpha}}{B}{B'}{\beta}{\tau_{\omega}}, then $\beta_{V,V} = \beta_{V,V'}$ for all $\alpha(V,V)$ and 
  $\alpha(V,V')$.
\end{lemma}

\begin{proof}
  Let $V,V'$ such that $\alpha(V,V)$ and $\alpha(V,V')$. By assumption, we have 
  \begin{enumerate}
    \item \sameType{[V/a]B}{[V/a]B'}{\beta_{V,V}}{\tau_{\omega}} \label{fact:1}
    \item \sameType{[V/a]B}{[V'/a]B'}{\beta_{V,V'}}{\tau_{\omega}} \label{fact:2}
  \end{enumerate}
  Apply symmetry and transitivity: 
  \begin{align}
    &\sameType{[V/a]B'}{[V/a]B}{\beta_{V,V}}{\tau_{\omega}} \nonumber\\
    &\sameType{[V/a]B}{[V/a]B}{\beta_{V,V}}{\tau_{\omega}}  \label{fact:5}\\
    &\sameType{[V'/a]B'}{[V/a]B}{\beta_{V,V'}}{\tau_{\omega}} \nonumber\\
    &\sameType{[V/a]B}{[V/a]B}{\beta_{V,V'}}{\tau_{\omega}} \label{fact:6}
  \end{align}
  Thus by unicity with~(\ref{fact:5}) and~(\ref{fact:6}) we have $\beta_{V,V} = \beta_{V,V'}$.
\end{proof}

\begin{lemma}[Open computation]\label{lemma:open_comp}
  Given \sameType{A}{A'}{\alpha}{\tau_{\omega}}, \openEqType{\isOf{a}{A}}{B}{B}, and 
  \openEqComp{\isOf{a}{A}}{N}{N'}{B}, there is a $\beta$ such that 
  \sameTypes{\isOf{a}{\alpha}}{B}{B'}{\beta}{\tau_{\omega}} and \sameComp{\isOf{a}{\alpha}}{N}{N'}{\beta}.
\end{lemma}

\begin{proof}
  The first part is the same as in Lemma~\ref{lemma:open_sem}. For the latter, let $V,V'$ such that $\alpha(V,V')$. 
  We need to show \sameComp{}{[V/a]N}{[V'/a]N'}{\beta_{V,V'}}. By assumption, \eqComp{[V/a]N}{[V'/a]N'}{[V/a]B}. 
  So it suffices to show $\sem{[V/a]B} = \beta_{V,V'}$, which follows from Lemma~\ref{lemma:func_per}.
\end{proof}
\section{Semantic Proof Theory}\label{sec:proof}
In the context of computational type theories, a proof theory could be viewed as a collection of 
generally useful lemmas about the semantic construction. As such, it is not the least set of derivable 
facts from the collection, and issues of admissibility of rules do not arise in this setting. 
This is a pragmatic approach, as it is often convenient to extend the proof theory as needed with new lemmas. 
Here we present a particular proof theory for \cctt{}. Notice that even though the lemmas are written in an inference 
rule format, they should be read as ``if the premises are true, then the conclusions are true'', i.e., not an inductive 
definition.

In Figure~\ref{fig:prooftheory}, we present a collection of lemmas that is useful for complexity analysis. They can be grouped 
roughly by the type constructor involved, resulting in the usual formation, introduction, and elimination rules. 
For brevity, we write the rules in the local form; however, all rules (with the
exception of Lemma~\ref{lemma:subsete},~\ref{lemma:reli}, and~\ref{lemma:rele}) can be interpreted to have an ambient context prepended
to the judgments. 

\begin{figure}
  \begin{mathpar}
    \inferrule{
    }{
      \eqVal{\univ{i}}{\univ{i}}{\univ{i+1}}
    }(\ref{lemma:univf})

    \inferrule{
      \eqVal{A}{A'}{\univ{i}}
    }{
      \eqType{A}{A'}
    }(\ref{lemma:unive})

    \inferrule{
      \eqComp{A}{A'}{\univ{i}}\\
      \eqComp{M}{M'}{A}\\
      \eqComp{N}{N'}{A}
    }{
      \eqVal{\eqty{A}{M}{N}}{\eqty{A'}{M'}{N'}}{\univ{i}}
    }(\ref{lemma:eqf})
    
    \inferrule{
      \eqComp{M}{N}{A}
    }{
      \eqVal{\triv}{\triv}{\eqty{A}{M}{N}}
    }(\ref{lemma:eqi})

    \inferrule{
      \eqComp{P}{P'}{\eqty{A}{M}{N}}
    }{
      \eqComp{M}{N}{A}
    }(\ref{lemma:eqe})

    \inferrule{
    }{
      \eqVal{\nat}{\nat}{\univ{i}}
    }(\ref{lemma:natf})

    \inferrule{
    }{
      \eqComp{\zero}{\zero}{\nat}
    }(\ref{lemma:nati})

    \inferrule{
      \eqComp{M}{M'}{\nat}
    }{
      \eqComp{\suc{M}}{\suc{M'}}{\nat}
    }(\ref{lemma:nati})

    \inferrule{
      f : \N \to \N\\
      \eqComp{M}{M'}{\nat}
    }{
      \eqComp{\cffone{f}{M}}{\cffone{f}{M'}}{\nat}
    }(\ref{lemma:ffe})

    \inferrule{
    \openComp{\isOf{a}{\nat}}{A}{\univ{i}}\\
    \eqComp{M}{M'}{\nat}\\
    \openEqComp{\isOf{p}{\eqty{\nat}{\zero}{M}}}{M_0}{M_0'}{[\zero/a]A}\\
    \openEqComp{\isOf{a}{\nat},\isOf{p}{\eqty{\nat}{\suc{a}}{M}}}{M_1}{M_1'}{[\suc{a}/a]A}
    }{
    \eqComp{\ifz{M}{M_0}{a}{M_1}}{\ifz{M'}{M_0'}{a}{M_1'}}{\lsub{M}{a}{A}}
    }(\ref{lemma:nate1})

    \inferrule{
    \openComp{\isOf{a}{\nat}}{A}{\univ{i}}\\
    \eqVal{V}{V'}{\nat}\\
    \eqComp{P_0}{P_0}{\nat}\\
    \openEqComp{\isOf{a}{\nat}}{P_1}{P_1}{\nat}\\
    \openEqComp{\isOf{p}{\eqty{\nat}{\zero}{M}}}{M_0}{M_0'}{[\zero/a]A}{P_0}\\
    \openEqCComp{\isOf{a}{\nat},\isOf{p}{\eqty{\nat}{\suc{a}}{V}}}{M_1}{M_1'}{[\suc{a}/a]A}{P_1}
    }{
    \eqCComp{\ifz{V}{M_0}{a}{M_1}}{\ifz{V'}{M_0'}{a}{M_1'}}{[V/a]A}{\ifz{V}{\suc{P_0}}{a}{\suc{P_1}}}
    }(\ref{lemma:nate2})

    \inferrule{
      \eqComp{A}{A'}{\univ{i}}\\
      \openEqComp{\isOf{a}{A}}{B}{B'}{\univ{i}}
    }{
      \eqVal{\subsetty{\isOf{a}{A}}{B}}{\subsetty{\isOf{a}{A'}}{B'}}{\univ{i}}
    }(\ref{lemma:subsetf})

    \inferrule{
      \isComp{A}{\univ{i}}\\
      \openComp{\isOf{a}{A}}{B}{\univ{i}}\\
      \eqCComp{M}{M'}{A}{P}\\
      \eqComp{N}{N'}{\lsub{M}{a}{B}}
    }{
      \eqCComp{M}{M'}{\subsetty{\isOf{a}{A}}{B}}{P}
    }(\ref{lemma:subseti})

    \inferrule{
      \eqVal{V}{V'}{\subsetty{\isOf{a}{A}}{B}}
    }{
      \eqVal{V}{V'}{A}\\
      \exists U,U'.\, \eqVal{U}{U'}{[V/a]B}
    }(\ref{lemma:subsete})

    \inferrule{
    r : \mathcal{P}(\N \times \N)\\
    \eqComp{M}{M'}{\nat}\\
    \eqComp{N}{N'}{\nat}
    }{
      \eqVal{\rel{r}{M}{N}}{\rel{r}{M'}{N'}}{\univ{i}}
    }(\ref{lemma:relf})

    \inferrule{
     \exists m,n.\, r(m,n) \land \eqComp{M}{\bar{m}}{\nat} \land \eqComp{N}{\bar{n}}{\nat} 
    }{
     \eqVal{\triv}{\triv}{\rel{r}{M}{N}}
    }(\ref{lemma:reli})

    \inferrule{
    \eqComp{P}{P'}{\rel{r}{M}{N}}
    }{
     \exists m,n.\, r(m,n) \land \eqComp{M}{\bar{m}}{\nat} \land \eqComp{N}{\bar{n}}{\nat}
    }(\ref{lemma:rele})

  \end{mathpar}
\caption{A proof theory for \cctt{}.}
\label{fig:prooftheory}
\end{figure}

\begin{figure}
\begin{mathpar}
\inferrule{
    \eqComp{A}{A'}{\univ{i}}\\
    \openEqComp{\isOf{a}{A}}{B}{B'}{\univ{i}}
    }{
      \eqVal{\sigmaty{\isOf{a}{A}}{B}}{\sigmaty{\isOf{a}{A'}}{B'}}{\univ{i}}
    }(\ref{lemma:sigmaf})

    \inferrule{
    \isComp{A}{\univ{i}}\\
    \openComp{\isOf{a}{A}}{B}{\univ{i}}\\
    \eqVal{V}{V'}{A}\\
    \eqVal{U}{U'}{[V/a]B}
    }{
     \eqVal{\pair{V}{U}}{\pair{V'}{U'}}{\sigmaty{\isOf{a}{A}}{B}}.
    }(\ref{lemma:sigmai})

    \inferrule{
    \eqVal{V}{V'}{\sigmaty{\isOf{a}{A}}{B}}
    }{
    \eqCComp{\fst{V}}{\fst{V'}}{A}{\bar{1}}\\
    \eqCComp{\snd{V}}{\snd{V'}}{\lsub{\fst{V}}{a}{B}}{\bar{1}}
    }(\ref{lemma:sigmae})

    \inferrule{
      \eqComp{A}{A'}{\univ{i}}\\
      \openEqComp{\isOf{a}{A}}{B}{B'}{\univ{i}}\\
      \openEqComp{\isOf{a}{A}}{P}{P'}{\nat}\\
    }{
      \eqVal{\arrtimev{\isOf{a}{A}}{B}{P}}{\arrtimev{\isOf{a}{A'}}{B'}{P'}}{\univ{i}}
    }(\ref{lemma:pitimef})

    \inferrule{
      \eqComp{A}{A'}{\univ{i}}\\
      \openEqComp{\isOf{a}{A}}{B}{B'}{\univ{i}}\\
      \openEqComp{\isOf{a}{A}}{P}{P'}{\nat}\\
      \openEqCComp{\isOf{a}{A},\isOf{f}{\arrtimev{\isOf{a}{\subsetty{\isOf{a'}{A}}{\rel{<}{[a'/a]P}{P}}}}{B}{P}}}{N}{N'}{B}{P}
    }{
      \eqVal{\fun{f}{a}{N}}{\fun{f}{a}{N'}}{\arrtimev{\isOf{a}{A}}{B}{P}}
    }(\ref{lemma:pitimei})
 
    \inferrule{
    \eqVal{\fun{a}{f}{N}}{\fun{a}{f}{N'}}{\arrtimev{\isOf{a}{A}}{B}{P}}\\
    \eqVal{V}{V'}{A}
    }{
      \eqCComp{\ap{\fun{a}{f}{N}}{V}}{\ap{\fun{a}{f}{N'}}{V'}}{[V/a]B}{\suc{[V/a]P}}.
    }(\ref{lemma:pitimee})
\end{mathpar}
\caption{A proof theory for \cctt{}.}
\end{figure}
\subsection{Universes}

The introduction rules for \univ{i} are the formation rules for the ``small
types'' (
Lemma~\ref{lemma:natf}, \ref{lemma:sigmaf}, \ref{lemma:subsetf},
\ref{lemma:relf}, \ref{lemma:pitimef}).
These introduction rules capture the by-name evaluation behavior of elements of
the universe; consequently, there is no cost-aware version of the typecase
mechanism. 

\begin{lemma}[Universe formation]\label{lemma:univf}
  \ \\
  \eqVal{\univ{i}}{\univ{i}}{\univ{i+1}} for all $i$.
\end{lemma}

\begin{proof}
  First, we need to show there is a $\phi$ such that \sameType{\univ{i+1}}{\univ{i+1}}{\phi}{\tau_{\omega}}.
  By definition, $\tau_{\omega} = \text{Types}(v_{\omega},\tau_{\omega})$ and 
  $\text{Univ}(v_{\omega}) \subseteq \text{Types}(v_{\omega},\tau_{\omega})$, 
  so it suffices to show \sameType{\univ{i+1}}{\univ{i+1}}{\phi}{v_{\omega}} for some $\phi$. 
  By definition of $v_{\omega}$, take $\phi = \{(A,B) \mid \exists \alpha.\,
  \tau_{i+1}(A,B,\alpha)\}$.
  
  Next, we need to show $\phi(\univ{i},\univ{i})$. It suffices to show there is
  $\alpha$ such that $\tau_{i+1}(\univ{i},\univ{i},\alpha)$. Since $\tau_{i+1} =
  \text{Types}(v_{i+1},\tau_{i+1})$, it suffices to show
  \sameType{\univ{i}}{\univ{i}}{\alpha}{v_{i+1}}. By definition of $v_{i+1}$, it
  suffices to take $\alpha = \{(A,B) \mid \exists \beta.\,
  \tau_{i}(A,B,\beta)\}$. 
\end{proof}

The introduction rules for \univ{i} are the formation rules for the type
constructors. 

\begin{lemma}[Universe elimination]\label{lemma:unive}
 If
 \begin{enumerate}
   \item \eqComp{A}{A'}{\univ{i}}
 \end{enumerate}
 then \eqType{A}{A'}.
\end{lemma}

\begin{proof}
  We need to show \sameType{A}{A'}{\alpha}{\tau_{\omega}} for some $\alpha$. By assumption, 
  \sameComp{}{A}{A'}{\sem{\univ{i}}}. By definition,
  \[
  \sem{\univ{i}} =  \{(A,B) \mid \exists \alpha.\, \tau_i(A,B,\alpha)\} 
  \]
  Hence \sameType{A}{A'}{\alpha}{\tau_i} for some $\alpha$, and so the result holds by Lemma~\ref{lemma:cumulativity}.
\end{proof}
\subsection{Equality}

\begin{lemma}[Equality formation]\label{lemma:eqf}
  If 
  \begin{enumerate}
    \item \eqComp{A}{A'}{\univ{i}}
    \item \eqComp{M}{M'}{A}
    \item \eqComp{N}{N'}{A}
  \end{enumerate}
  then \eqVal{\eqty{A}{M}{N}}{\eqty{A'}{M'}{N'}}{\univ{i}}
\end{lemma}

\begin{proof}
  By definition, $\sem{\univ{i}} = \{(A,B) \mid \exists \phi.\,
  \tau_i(A,B,\phi)\}$, so 
  we need to show $\tau_i(\eqty{A}{M}{N}, \eqty{A'}{M'}{N'}, \phi)$ for some $\phi$.
  Since $\tau_{i} = \text{Types}(v_{i},\tau_{i})$, it suffices to show 
  $\text{Eq}(\tau_i)(\eqty{A}{M}{N}, \eqty{A'}{M'}{N'}, \phi)$. By assumption
  and definition of \sem{\univ{i}}, we have 
  \sameType{A}{A}{\alpha}{\tau_i} for some $\alpha$. Furthermore, by
  Lemma~\ref{lemma:cumulativity} and unicity of $\tau_{\omega}$, we have 
  \begin{enumerate}
    \item \sameComp{}{M}{M'}{\alpha}
    \item \sameComp{}{N}{N'}{\alpha}
  \end{enumerate}
  So the result holds by taking $\phi = \{(\triv, \triv) \mid \sameComp{}{M}{N}{\alpha} \}$.
\end{proof}

\begin{lemma}[Equality introduction]\label{lemma:eqi}
  If 
  \begin{enumerate}
    \item \eqComp{M}{N}{A}
  \end{enumerate}
  then \eqVal{\triv}{\triv}{\eqty{A}{M}{N}}.
\end{lemma}

\begin{proof}
  First, we need to show \sameType{\eqty{A}{M}{N}}{\eqty{A}{M}{N}}{\phi}{\tau_{\omega}} for some $\phi$, which holds by Lemma~\ref{lemma:eqf}.
  By definition, $\phi = \{(\triv, \triv) \mid \sameComp{}{M}{N}{\sem{A}} \}$. Next, we need to show 
  $\phi(\triv,\triv)$. It suffices to show \sameComp{}{M}{N}{\sem{A}}, which holds by assumption.
\end{proof}

\begin{lemma}[Equality elimination]\label{lemma:eqe}
  If 
  \begin{enumerate}
    \item \eqComp{P}{P'}{\eqty{A}{M}{N}}
  \end{enumerate}
  then \eqComp{M}{N}{A}.
\end{lemma}

\begin{proof}
  First, we need to show \sameType{A}{A}{\alpha}{\tau_{\omega}}. By assumption, we have 
  \sameType{\eqty{A}{M}{N}}{\eqty{A}{M}{N}}{\phi}{\tau_{\omega}} for some $\phi$. 
  This implies \sameType{A}{A}{\alpha}{\tau_{\omega}}. By unicity,
  we have $\phi = \{(\triv, \triv) \mid \sameComp{}{M}{N}{\alpha} \}$.
  It remains to show \sameComp{}{M}{N}{\alpha}.
  By assumption, \sameComp{}{P}{P'}{\phi}, so we have \sameComp{}{M}{N}{\alpha}, as required. 
\end{proof}
\subsection{Nat}

Before introducing judgments involving complexities, we need to prove some
lemmas governing the construction and use of natural numbers. As we alluded to in
section~\ref{sec:setup}, the complexity component of a judgment is simply a
natural number in the programming language itself. Aside from the rules for
eliminating foreign functions, there are also two elimination rules --
Lemma~\ref{lemma:nate1} is the ordinary eliminator, and Lemma~\ref{lemma:nate2}
is the eliminator that additionally tracks the complexity of the case analysis
given the complexity of each branch. Note that there is no syntactic form for
the recursor: it is an instance of general recursion. The associated induction
rule is also an instance of the elimination rule for the funtime type (see
Lemma~\ref{lemma:pitimei} and~\ref{lemma:pitimee}).

\begin{lemma}[Nat formation]\label{lemma:natf}
  \ \\
  \eqVal{\nat}{\nat}{\univ{i}}.
\end{lemma}

\begin{proof}
  By definition of $\sem{\univ{i}}$, it suffices to show 
  $\tau_i(\nat, \nat, \alpha)$ for some $\alpha$.
  Unfolding definitions, we have $\text{Nat}(\tau_{i}) \subseteq \text{Types}(v_{\omega},\tau_{i}) = \tau_{i}$.
  So take $\alpha = \omega$, since $\text{Nat}(\tau_i)(\nat,\nat,\omega)$. 
\end{proof}

\begin{lemma}[Nat introduction]\label{lemma:nati}
  \ 
  \begin{enumerate}
    \item \eqComp{\zero}{\zero}{\nat}.
    \item If \eqComp{M}{M'}{\nat}, then \eqComp{\suc{M}}{\suc{M'}}{\nat}.
  \end{enumerate}
\end{lemma}

\begin{proof}
  First, by Lemma~\ref{lemma:natf}, we know there is $\alpha$ such that \sameType{\nat}{\nat}{\alpha}{\tau}. 
  By unicity of $\tau_{\omega}$, $\alpha = \omega$. For part 1, it suffices to show \sameComp{}{\zero}{\zero}{\omega}.
  This holds by definition, as $\omega = \{(\zero,\zero)\} \cup \{(\suc{V}, \suc{V'}) \mid \omega(V,V')\}$. 
  Now suppose \eqComp{M}{M'}{\nat}, which by above, means \sameComp{}{M}{M'}{\omega}. We need to show 
  \sameComp{}{\suc{M}}{\suc{M'}}{\nat}. By assumption, \eval{M}{V}, \eval{M'}{V'}, and $\omega(V,V')$.
  Hence \eval{\suc{M}}{\suc{V}} and \eval{\suc{M'}}{\suc{V'}}. By definition, $\omega(\suc{V},\suc{V'})$,
  and so\sameComp{}{\suc{M}}{\suc{M'}}{\nat}.
\end{proof}

\begin{lemma}[FF elimination (1)]\label{lemma:ffe}
  \ 
  \begin{enumerate}
    \item If $f : \N \to \N$ and \eqComp{M}{M'}{\nat}, then \eqComp{\cffone{f}{M}}{\cffone{f}{M'}}{\nat}.
    \item If $f : \N \to \N \to \N$ and \eqComp{M_1}{M_1'}{\nat} and \eqComp{M_2}{M_2'}{\nat}, then 
    \eqComp{\cfftwo{f}{M_1}{M_2}}{\cfftwo{f}{M_1'}{M_2'}}{\nat}.
  \end{enumerate}
\end{lemma}

\begin{proof}
  We prove the unary case, as the binary case follows similarly. 
  Since \sameType{\nat}{\nat}{\omega}{\tau_{\omega}}, it suffices to show 
  \sameComp{}{\cffone{f}{M}}{\cffone{f}{M'}}{\omega}.
  By assumption and Lemma~\ref{lemma:closed_sem}, we know 
  \eval{M}{V}, \eval{M'}{V'}, and $\omega(V,V')$. This means $V = \bar{m}$ and $V' = \bar{m}$ for some $m$. 
  Hence, 
  \begin{align*}
    \cffone{f}{M} &\mapsto^* \cffone{f}{\bar{m}} \mapsto \overline{f(m)}\\
    \cffone{f}{M'} &\mapsto^* \cffone{f}{\bar{m}} \mapsto \overline{f(m)}
  \end{align*}
  By Lemma~\ref{lemma:headexp1}, it suffices to show \sameComp{}{\overline{f(m)}}{\overline{f(m)}}{\omega}, 
  which holds by definition of $\omega$.
\end{proof}

\begin{lemma}[FF elimination (2)]\label{lemma:ffe2}
  \ 
  \begin{enumerate}
    \item If $f : \N \to \N$ and \eqVal{V}{V'}{\nat}, then \eqCComp{\cffone{f}{V}}{\cffone{f}{V'}}{\nat}{\bar{1}}.
    \item If $f : \N \to \N \to \N$ and \eqVal{V_1}{V_1'}{\nat} and \eqVal{V_2}{V_2'}{\nat}, then 
    \eqCComp{\cfftwo{f}{V_1}{V_2}}{\cfftwo{f}{V_1'}{V_2'}}{\nat}{\bar{1}}.
  \end{enumerate}
\end{lemma}

\begin{proof}
  Similar to proof of Lemma~\ref{lemma:ffe}.
\end{proof}

\begin{lemma}[Nat elimination (1)]\label{lemma:nate1}
  If
  \begin{enumerate}
    \item \openComp{\isOf{a}{\nat}}{A}{\univ{i}},
    \item \eqComp{M}{M'}{\nat},
    \item \openEqComp{\isOf{p}{\eqty{\nat}{\zero}{M}}}{M_0}{M_0'}{[\zero/a]A},
    \item \openEqComp{\isOf{a}{\nat},\isOf{p}{\eqty{\nat}{\suc{a}}{M}}}{M_1}{M_1'}{[\suc{a}/a]A},
  \end{enumerate}
  then \eqComp{\ifz{M}{M_0}{a}{M_1}}{\ifz{M'}{M_0'}{a}{M_1'}}{\lsub{M}{a}{A}}.
\end{lemma}

\begin{proof}
  By assumption 2, 
  \sameComp{}{M}{M'}{\omega}, so 
  \begin{enumerate}
    \item \eval{M}{V}
    \item \eval{M'}{V'}
    \item $\omega(V,V')$
  \end{enumerate}
  One possibility is that $(V,V') = (\zero,\zero)$. 
  In this case, we know that 
  \begin{align*}
    \ifz{M}{M_0}{a}{M_1} &\mapsto^* \ifz{\zero}{M_0}{a}{M_1} \mapsto M_0\\
    \ifz{M'}{M_0'}{a}{M_1'} &\mapsto^* \ifz{\zero}{M_0'}{a}{M_1'} \mapsto M_0'
  \end{align*} 
  By Lemma~\ref{lemma:headexp}, it suffices to show \eqComp{M_0}{M_0'}{[\zero/a]A}, 
  which holds by the third assumption, given that we show \eqVal{\zero{}}{M}{\nat},
  which again follows from Lemma~\ref{lemma:headexp}.
  The other possibility is that $(V,V') = (\suc{U},\suc{U'})$ where
  $\omega(U,U')$. Then we have  
  \begin{align*}
    \ifz{M}{M_0}{a}{M_1} &\mapsto^* \ifz{\suc{U}}{M_0}{a}{M_1} \mapsto [U/a]M_1\\
    \ifz{M'}{M_0'}{a}{M_1'} &\mapsto^* \ifz{\suc{U'}}{M_0'}{a}{M_1'} \mapsto [U'/a]M_1'
  \end{align*} 
  By Lemma~\ref{lemma:headexp}, it suffices to show \eqComp{[U/a]M_1}{[U'/a]M_1'}{[\suc{U}/a]A},
  which holds by the fourth assumption, given we show that \isVal{P}{\eqty{\nat}{\suc{U}}{M}} for some $P$. 
  By Lemma~\ref{lemma:eqi}, it suffices to show \eqComp{\suc{U}}{M}{\nat}, which holds by Lemma~\ref{lemma:headexp}.
\end{proof}

\begin{lemma}[Nat elimination (2)]\label{lemma:nate2}
  If
  \begin{enumerate}
    \item \openComp{\isOf{a}{\nat}}{A}{\univ{i}},
    \item \eqVal{V}{V'}{\nat},
    \item \eqComp{P_0}{P_0}{\nat},
    \item \openEqComp{\isOf{a}{\nat}}{P_1}{P_1}{\nat},
    \item \openEqCComp{\isOf{p}{\eqty{\nat}{\zero}{V}}}{M_0}{M_0'}{[\zero/a]A}{P_0},
    \item \openEqCComp{\isOf{a}{\nat},\isOf{p}{\eqty{\nat}{\suc{a}}{V}}}{M_1}{M_1'}{[\suc{a}/a]A}{P_1},
  \end{enumerate}
  then \eqCComp{\ifz{V}{M_0}{a}{M_1}}{\ifz{V'}{M_0'}{a}{M_1'}}{[V/a]A}{\ifz{V}{\suc{P_0}}{a}{\suc{P_1}}}.
\end{lemma}

\begin{proof}
  By assumption 2, we know $\omega(V,V')$. One possibility is that $(V,V') = (\zero,\zero)$. 
  In this case, we know that 
  \begin{align*}
    \ifz{\zero}{M_0}{a}{M_1} & \mapsto M_0\\
    \ifz{\zero'}{M_0'}{a}{M_1'} & \mapsto M_0'\\
    \ifz{\zero}{\suc{P_0}}{a}{\suc{P_1}} &\mapsto^* \suc{P_0}
  \end{align*} 
  By Lemma~\ref{lemma:headexp}, we have \eqComp{\ifz{M}{\suc{P_0}}{a}{\suc{P_1}}}{\suc{P_0}}{\nat}.
  So by Lemma~\ref{lemma:headexp}, it suffices to show 
  \eqCComp{M_0}{M_0'}{[\zero/a]A}{P_0}. This holds by assumption 5, given 
  that we show \eqVal{\zero}{V}{\nat}, which holds by Lemma~\ref{lemma:nati}.
  The other possibility is that $(V,V') = (\suc{U},\suc{U'})$ where
  $\omega(U,U')$. Then we have  
  \begin{align*}
    \ifz{\suc{U}}{M_0}{a}{M_1} & \mapsto [U/a]M_1\\
    \ifz{\suc{U'}}{M_0'}{a}{M_1'} &\mapsto [U'/a]M_1'\\
    \ifz{\suc{U}}{\suc{P_0}}{a}{\suc{P_1}} &\mapsto^* [U/a]\suc{P_1}
  \end{align*} 
  As above, by Lemma~\ref{lemma:headexp}, we have \eqComp{\ifz{\suc{U}}{\suc{P_0}}{a}{\suc{P_1}}}{[U/a]\suc{P_1}}{\nat}.
  So by Lemma~\ref{lemma:headexp}, it suffices to show 
  \eqCComp{[U/a]M_1}{[U'/a]M_1'}{[\suc{U}/a]A}{[U/a]P_1}, which holds by assumption 6, given that we 
  show \isVal{P}{\eqty{\nat}{\suc{U}}{M}} for some $P$. 
  By Lemma~\ref{lemma:eqi}, it suffices to show \eqComp{\suc{U}}{M}{\nat}, which holds by Lemma~\ref{lemma:headexp}.
\end{proof}
\subsection{Subset}

\begin{lemma}[Subset formation]\label{lemma:subsetf}
If
\begin{enumerate}
\item \eqComp{A}{A'}{\univ{i}}
\item \openEqComp{\isOf{a}{A}}{B}{B'}{\univ{i}}
\end{enumerate}
then \eqVal{\subsetty{\isOf{a}{A}}{B}}{\subsetty{\isOf{a}{A'}}{B'}}{\univ{i}}
\end{lemma}

\begin{proof}
  By definition of \sem{\univ{i}}, 
  we need to show $\tau_i(\subsetty{\isOf{a}{A}}{B}, \subsetty{\isOf{a}{A}}{B}, \phi)$ 
  for some $\phi$.
  By definition, $\tau_i = \text{Types}(v_{i},\tau_{i})$, so it suffices to show 
  $\text{Subset}(\tau_{i})(\subsetty{\isOf{a}{A}}{B},\subsetty{\isOf{a}{A'}}{B'},\phi)$ for some $\phi$.
  By assumption, we have \sameType{A}{A'}{\alpha}{\tau_{i}} for some $\alpha$. It remains to show 
  \sameTypes{\isOf{a}{\alpha}}{B}{B'}{\beta}{\tau_{i}} for some $\beta$, which exists 
  by Lemma~\ref{lemma:open_sem}. By definition, 
  we can take $\phi = \{(V,V') \mid \alpha(V,V') \land \exists U,U'.\beta_{V,V'}(U,U')\}$.
\end{proof}

\begin{lemma}[Subset introduction]\label{lemma:subseti}
  If 
  \begin{enumerate}
    \item \isComp{A}{\univ{i}}
    \item \openComp{\isOf{a}{A}}{B}{\univ{i}}
    \item \eqCComp{M}{M'}{A}{P}
    \item \eqComp{N}{N'}{\lsub{M}{a}{B}}
  \end{enumerate}
  then \eqCComp{M}{M'}{\subsetty{\isOf{a}{A}}{B}}{P}
\end{lemma}

\begin{proof}
  First, we need to show \sameType{\subsetty{\isOf{a}{A}}{B}}{\subsetty{\isOf{a}{A}}{B}}{\phi}{\tau_{\omega}} for some $\phi$,
  which follows from Lemma~\ref{lemma:subsetf} and~\ref{lemma:cumulativity}. By unicity of $\tau_{\omega}$, 
  $\phi = \{(V,V') \mid \sem{A}(V,V') \land \exists U,U'.\sem{\isOf{a}{\alpha}(B,B)}_{V,V'}(U,U')\}$. 
  It suffices to show \sameCComp{}{M}{M'}{\phi}{P}. By assumption 3, we know that 
  \begin{enumerate}
    \item \evalCost{M}{c}{V}
    \item \evalCost{M'}{c'}{V'}
    \item \eval{P}{\bar{p}}
    \item $p \ge \max(c,c')$
    \item $\sem{A}(V,V')$
  \end{enumerate}
  It remains to show $\exists U,U'.\sem{\isOf{a}{\alpha}.(B,B')}_{V,V'}(U,U')$.
  By Lemma~\ref{lemma:headexp}, \eqComp{N}{N'}{[V/a]B}, which means
  \begin{enumerate}
    \item \eval{N}{U}
    \item \eval{N'}{U'}
    \item $\sem{[V/a]B}(U,U')$
  \end{enumerate}
  It suffices to show $\sem{[V/a]B} = \sem{\isOf{a}{\alpha}.(B,B)}_{V,V'}$. By unicity,  
  $\sem{[V/a]B} = \sem{\isOf{a}{\alpha}.(B,B)}_{V,V}$, so the result follows from Lemma~\ref{lemma:func_per}.
\end{proof}

\begin{lemma}[Subset elimination]\label{lemma:subsete}
  If 
  \begin{enumerate}
    \item \eqVal{V}{V'}{\subsetty{\isOf{a}{A}}{B}}
  \end{enumerate}
  then \eqVal{V}{V'}{A} and there exists $U,U'$ such that \eqVal{U}{U'}{[V/a]B}.
\end{lemma}

\begin{proof}
  By assumption, \sameType{\subsetty{\isOf{a}{A}}{B}}{\subsetty{\isOf{a}{A}}{B}}{\phi}{\tau_{\omega}}. This implies there 
  exists $\alpha,\beta$ such that \sameType{A}{A}{\alpha}{\tau_{\omega}} and \sameTypes{\isOf{a}{\alpha}}{B}{B}{\beta}{\tau_{\omega}}.
  Furthermore, we know that $\phi(V,V')$.
  By definition of $\phi$, $\alpha(V,V')$ and $\beta_{V,V'}(U,U')$ for some $U,U'$. 
  Hence we have \eqVal{V}{V'}{A}. For the second part, we first need to show 
  \sameType{[V/a]B}{[V/a]B}{\rho}{\tau_{\omega}} for some $\rho$. 
  Taking $\rho = \beta_{V,V'}$, this follows from 
  \sameTypes{\isOf{a}{\alpha}}{B}{B}{\beta}{\tau_{\omega}} and symmetry and transitivity of $\tau_{\omega}$.
  The fact that \sameComp{}{U}{U'}{\rho} then holds by definition of $\phi$.
\end{proof}
\subsection{Relations}

Some properties about natural numbers can be lifted from the metatheory. Given a relation 
$r : \mathcal{P}(\N \times \N)$, \rel{r}{M}{N} is the internal representation of $r(m,n)$ when 
\eval{M}{\bar{m}} and \eval{N}{\bar{n}}. For the following lemmas, we just prove the binary versions, 
noting that the proof stays the same for any relation $\mathcal{P}(\N^k)$ of
fixed arity $k$. 

As for subsets, the rules for relations cannot be generalized
to arbitrary contexts in the expected way. The introduction rule
(Lemma~\ref{lemma:reli}) involves identifying metatheoretic numbers that
represent closed theoretic natural numbers. This is not possible for open terms
of type \nat. The elimination rule (Lemma~\ref{lemma:rele}) has a similar
problem. As in the case for subsets, we just instantiate the
semantics involving open judgments for relations. 

\begin{lemma}[Relation Formation]\label{lemma:relf}
  If
  \begin{enumerate}
    \item $r : \mathcal{P}(\N \times \N)$
    \item \eqComp{M}{M'}{\nat}
    \item \eqComp{N}{N'}{\nat}
  \end{enumerate}
  then \eqVal{\rel{r}{M}{N}}{\rel{r}{M'}{N'}}{\univ{i}}.
\end{lemma}

\begin{proof}
  By definition of \sem{\univ{i}},
  we need to show $\tau_i(\rel{r}{M}{N}, \rel{r}{M'}{N'}, \phi)$ for some $\phi$. 
  By definition, $\tau_{i} = \text{Types}(v_{i},\tau_{i})$, so it suffices to show 
  $\text{Rel2}(\tau_{i})(\rel{r}{M}{N}, \rel{r}{M'}{N'}, \phi)$ for some $\phi$.
  By assumption and unicity of $\tau_{\omega}$, we know \sameComp{}{M}{M'}{\omega} and \sameComp{}{N}{N'}{\omega}, 
  so by definition, take $\phi =  \{ (\triv,\triv) \mid \exists m,n.\, r(m,n) \land \sameComp{}{M}{\bar{m}}{\omega} \land \sameComp{}{N}{\bar{n}}{\omega} \}$.
\end{proof}

\begin{lemma}[Relation introduction]\label{lemma:reli}
  If 
  \begin{enumerate}
    \item \eqComp{M}{\bar{m}}{\nat} for some $m$
    \item \eqComp{N}{\bar{n}}{\nat} for some $n$
    \item $r(m,n)$
  \end{enumerate}
  then \eqVal{\triv}{\triv}{\rel{r}{M}{N}}. 
\end{lemma}

\begin{proof}
  First, we need to show \sameType{\rel{r}{M}{N}}{\rel{r}{M}{N}}{\phi}{\tau_{\omega}} for some $\phi$, 
  which holds by Lemma~\ref{lemma:relf}. By unicity,
  \[\phi =\{ (\triv,\triv) \mid \exists m,n.\, r(m,n) \land \sameComp{}{M}{\bar{m}}{\omega} \land \sameComp{}{N}{\bar{n}}{\omega} \}\]
  Next, we need to show $\phi(\triv,\triv)$. 
  It suffices to show $r(m,n)$, \sameComp{}{M}{\bar{m}}{\omega}, and \sameComp{}{N}{\bar{n}}{\omega} for some $m,n$, which 
  follows from the assumptions.
\end{proof}

\begin{lemma}[Relation elimination]\label{lemma:rele}
  If 
  \begin{enumerate}
    \item \eqComp{P}{P'}{\rel{r}{M}{N}}
  \end{enumerate}
  then \eqComp{M}{\bar{m}}{\nat}, \eqComp{N}{\bar{n}}{\nat}, and $r(m,n)$ for some $m,n$.
\end{lemma}

\begin{proof}
  By assumption, \sameType{\rel{r}{M}{N}}{\rel{r}{M}{N}}{\phi}{\tau_{\omega}}, which implies 
  \sameComp{}{M}{M}{\omega} and \sameComp{}{N}{N}{\omega}. By unicity, 
  \[\phi =\{ (\triv,\triv) \mid \exists m,n.\, r(m,n) \land \sameComp{}{M}{\bar{m}}{\omega} \land \sameComp{}{N}{\bar{n}}{\omega} \}\]
  So we have \eval{P}{\triv}, \eval{P'}{\triv}, and $\phi(\triv,\triv)$, so the result follows from definition of $\phi$.
\end{proof}
\subsection{Dependent Product}

\begin{lemma}[Sigma formation]\label{lemma:sigmaf}
  If 
  \begin{enumerate}
    \item \eqComp{A}{A'}{\univ{i}}
    \item \openEqComp{\isOf{a}{A}}{B}{B'}{\univ{i}}
  \end{enumerate}
then \eqVal{\sigmaty{\isOf{a}{A}}{B}}{\sigmaty{\isOf{a}{A'}}{B'}}{\univ{i}}
\end{lemma}

\begin{proof}
  By definition of \sem{\univ{i}}, we need to show
  $\tau_i(\sigmaty{\isOf{a}{A}}{B}, \sigmaty{\isOf{a}{A}}{B}, \phi)$ for some $\phi$.
  By definition, $\tau_{i} = \text{Types}(v_{i},\tau_{i})$, so it suffices to show 
  $\text{Sigma}(\tau_{i})(\sigmaty{\isOf{a}{A}}{B},\sigmaty{\isOf{a}{A'}}{B'},\phi)$ for some $\phi$.
  By assumption, we have \sameType{A}{A'}{\alpha}{\tau_{i}} for some $\alpha$. It remains to show 
  \sameTypes{\isOf{a}{\alpha}}{B}{B'}{\beta}{\tau_{i}} for some $\beta$, which exists 
  by Lemma~\ref{lemma:open_sem}. By definition, take 
  $\phi = \{(\pair{V}{U}, \pair{V'}{U'}) \mid \alpha(V,V') \land \beta_{V,V'}(U,U') \}$.
\end{proof}

\begin{lemma}[Sigma introduction]\label{lemma:sigmai}
  If 
  \begin{enumerate}
    \item \isComp{A}{\univ{i}}
    \item \openComp{\isOf{a}{A}}{B}{\univ{i}}
    \item \eqVal{V}{V'}{A}
    \item \eqVal{U}{U'}{[V/a]B}
  \end{enumerate}
  then \eqVal{\pair{V}{U}}{\pair{V'}{U'}}{\sigmaty{\isOf{a}{A}}{B}}.
\end{lemma}

\begin{proof}
  First, we need to show \sameType{\sigmaty{\isOf{a}{A}}{B}}{\sigmaty{\isOf{a}{A}}{B}}{\phi}{\tau_{\omega}} for some $\phi$.
  This follows from Lemma~\ref{lemma:sigmaf}. By unicity of $\tau_{\omega}$, 
  $\phi = \{(\pair{V}{U}, \pair{V'}{U'}) \mid \sem{A}(V,V') \land \sem{\isOf{a}{\alpha}.(B,B)}_{V,V'}(U,U') \}$.
  Next, we need to show \sameComp{}{\pair{V}{U}}{\pair{V'}{U'}}{\phi}. By assumption, 
  $\sem{A}(V,V')$. It remains to show $\sem{\isOf{a}{\alpha}.(B,B)}_{V,V'}(U,U')$.
  By assumption, $\sem{[V/a]B}(U,U)$, so the result follows from Lemma~\ref{lemma:func_per}.
\end{proof}

\begin{lemma}[Sigma elimination]\label{lemma:sigmae}
  \ \\
  If \eqVal{V}{V'}{\sigmaty{\isOf{a}{A}}{B}}
  then 
  \begin{enumerate}
    \item \eqCComp{\fst{V}}{\fst{V'}}{A}{\bar{1}}
    \item \eqCComp{\snd{V}}{\snd{V'}}{\lsub{\fst{V}}{a}{B}}{\bar{1}}
  \end{enumerate}
\end{lemma}

\begin{proof}
  By assumption, \sameType{\sigmaty{\isOf{a}{A}}{B}}{\sigmaty{\isOf{a}{A}}{B}}{\phi}{\tau_{\omega}}. 
  This implies there exists $\alpha,\beta$ such that \sameType{A}{A}{\alpha}{\tau_{\omega}} and 
  \sameTypes{\isOf{a}{\alpha}}{B}{B}{\beta}{\tau_{\omega}}. By unicity, 
  $\phi = \{(\pair{V}{U}, \pair{V'}{U'}) \mid \alpha(V,V') \land \beta_{V,V'}(U,U') \}$.
  Furthermore, we have $\phi(V,V')$. This means 
  \begin{enumerate}
    \item $V = \pair{V_1}{V_2}$
    \item $V' = \pair{V_1'}{V_2'}$
    \item $\alpha(V_1,V_1')$ \label{fact:3}
    \item $\beta_{V_1,V_1'}(V_2,V_2')$ \label{fact:4}
  \end{enumerate}
  For part 1, it suffices to show \sameCComp{}{\fst{V}}{\fst{V'}}{\alpha}{\bar{1}}. First 
  \sameComp{}{\bar{1}}{\bar{1}}{\omega} by Lemma~\ref{lemma:nati} and unicity.
  We know that 
  \begin{align*}
    \fst{\pair{V_1}{V_2}} &\mapsto V_1\\
    \fst{\pair{V_1'}{V_2'}} &\mapsto V_1'
  \end{align*}
  so by Lemma~\ref{lemma:headexp1}, it suffices to show \sameCComp{}{V_1}{V_1'}{\alpha}{\zero}, 
  which holds by~(\ref{fact:3}).

  For part 2, by Lemma~\ref{lemma:headexp}, it suffices to show 
  \eqCComp{\snd{V}}{\snd{V'}}{[V_1/a]B}{\bar{1}}.
  First, we need to show \sameType{[V_1/a]B}{[V_1/a]B}{\rho}{\tau_{\omega}} for some $\rho$.
   Taking $\rho = \beta_{V_1,V_1'}$, this follows from 
  \sameTypes{\isOf{a}{\alpha}}{B}{B}{\beta}{\tau_{\omega}} and symmetry and transitivity of $\tau_{\omega}$.
  Now it remains to show \sameCComp{}{\snd{V}}{\snd{V'}}{\rho}{\bar{1}}. 
  As above, we have \sameComp{}{\bar{1}}{\bar{1}}{\omega} by Lemma~\ref{lemma:nati} and unicity. 
  We know that 
  \begin{align*}
    \snd{\pair{V_1}{V_2}} &\mapsto V_2\\
    \snd{\pair{V_1'}{V_2'}} &\mapsto V_2'
  \end{align*}
  so by Lemma~\ref{lemma:headexp1}, it suffices to show \sameCComp{}{V_2}{V_2'}{\rho}{\zero}. 
  This follows from~(\ref{fact:4}), provided that $\rho = \beta_{V_1,V_1'}$, which follows by definition of $\rho$.
\end{proof}
\subsection{Funtime}

Now we have all the ingredients to introduce the new dependent function type
constructor ``funtime'' that internalizes the complexity-aware membership
judgment \openCComp{\Gamma}{M}{A}{P}.

\begin{lemma}[Funtime formation]\label{lemma:pitimef}
  If 
  \begin{enumerate}
    \item \eqComp{A}{A'}{\univ{i}}
    \item \openEqComp{\isOf{a}{A}}{B}{B'}{\univ{i}}
    \item \openEqComp{\isOf{a}{A}}{P}{P'}{\nat}
  \end{enumerate}
  then
  \eqVal{\arrtimev{\isOf{a}{A}}{B}{P}}{\arrtimev{\isOf{a}{A'}}{B'}{P'}}{\univ{i}}.
\end{lemma}

\begin{proof}
  By definition of \sem{\univ{i}}, 
  we need to show $\tau_i(\arrtimev{\isOf{a}{A}}{B}{P},
  \arrtimev{\isOf{a}{A'}}{B'}{P'}, \phi)$ for some $\phi$.
  Since $\tau_{i} = \text{Types}(v_{i},\tau_{i})$, it suffices to show 
  $\text{Funtime}(\tau_i)(\arrtimev{\isOf{a}{A}}{B}{P}, \arrtimev{\isOf{a}{A'}}{B'}{P'}, \phi)$ for some $\phi$.
  By assumption, \sameType{A}{A'}{\alpha}{\tau_{i}} for some $\alpha$. By Lemma~\ref{lemma:open_sem}, 
  \sameTypes{\isOf{a}{\alpha}}{B}{B'}{\beta}{\tau_{i}} for some $\beta$. By Lemma~\ref{lemma:open_comp},
  \sameTypes{\isOf{a}{\alpha}}{\nat}{\nat}{\rho}{\tau_{\omega}} and \sameComp{\isOf{a}{\alpha}}{P}{P'}{\rho}.  
  By unicity $\rho = \omega$. Then by definition, we can pick 
  $\phi = \{(\fun{f}{a}{N},\fun{f'}{a'}{N'}) \mid \sameCComp{\isOf{a}{\alpha}}{[\fun{f}{a}{N}/f]N}{[\fun{f}{a}{N'}/f]N'}{\beta}{P}{} \}$.  
\end{proof}

Before we prove the introduction rule, we will need some auxiliary definitions: 
\begin{enumerate}
  \item Let the domain of a PER $\alpha$ be $\dom{\alpha} = \{V \mid \alpha(V,V)\}$.
  \item A complexity description \sameComp{\isOf{a}{\alpha}}{P}{P}{\omega} induces a measure 
    $m_{\isOf{a}{\alpha}.P} : \dom{\alpha} \to \N$, defined as $m_{\isOf{a}{\alpha}.P}(V) = p$, where \eval{[V/a]P}{\bar{p}}.
  \item Given \sameComp{\isOf{a}{\alpha}}{P}{P}{\omega} and $V_1,V_2 : \dom{\alpha}$, write $\ltt{V_1}{\isOf{a}{\alpha}.P}{V_2}$
  when $m_{\isOf{a}{\alpha}.P} V_1 < m_{\isOf{a}{\alpha}.P} V_2$.  
\end{enumerate}

\begin{lemma}
  Given a set $A$ and a function $f : A \to \N$, $(A,\prec)$ is well-founded,
  where $a \prec b \iff f(a) < f(b)$.
\end{lemma}
\begin{proof}
  Suppose $S \subseteq A$ is inductive: $\forall x \in A, (\forall y \in A, y \prec x \implies y
  \in S) \implies x \in S$. It suffices to show $S = A$, and hence it suffices
  to show $A \subseteq S$. 
  Suppose $x \in A$. By assumption, it suffices to show 
  \[\forall y \in A, y \prec x \implies y \in S \iff \forall y \in A, f(y) < f(x)
  \implies y \in S\]  
  Proceed by induction on $f(x)$. If $f(x) = 0$, then the result holds
  vacuously. Otherwise, let $f(x) = m$, and suppose $y \in A$ and $f(y) < m$. We need to show $y \in
  S$. By the assumption that $S$ inductive, it suffices to show 
  \[
    \forall z \in A, f(z) < f(y) \implies z \in S  
  \]
  which holds by induction.
\end{proof}

\begin{corollary}
  Given \sameComp{\isOf{a}{\alpha}}{P}{P}{\omega}, $(\dom{\alpha}, \lessdot_{\isOf{a}{\alpha}.P})$ is a well-founded set.
\end{corollary}

As foreshadowed in the introduction, the proof for the introduction rule
proceeds by induction on the well-founded ordering on the inputs to the
recursive calls. We conjecture that that this rule is sufficient for verifying all
first-order functions. Curried functions and closures present problems in some
situations because the first argument does not necessarily decrease in a partial
function application. One possible solution to this would be to define an
ordering on the resulting closure and demand that partial applications induce a
``smaller'' closure in an appropriate sense. We leave this for future explorations. 

\begin{lemma}[Funtime introduction]\label{lemma:pitimei}
  If 
  \begin{enumerate}
    \item \eqComp{A}{A}{\univ{i}}
    \item \openEqComp{\isOf{a}{A}}{B}{B}{\univ{i}}
    \item \openEqComp{\isOf{a}{A}}{P}{P}{\nat}
    \item \openEqCComp{\isOf{a}{A},\isOf{f}{\arrtimev{\isOf{a}{\subsetty{\isOf{a'}{A}}{\rel{<}{[a'/a]P}{P}}}}{B}{P}}}{N}{N'}{B}{P}
  \end{enumerate}
  then \eqVal{\fun{a}{f}{N}}{\fun{a}{f}{N'}}{\arrtimev{\isOf{a}{A}}{B}{P}}.
\end{lemma}

\begin{proof}
  First, we need to show \sameType{\arrtimev{\isOf{a}{A}}{B}{P}}{\arrtimev{\isOf{a}{A}}{B}{P}}{\phi}{\tau_{\omega}} for some $\phi$.
  This follows from Lemma~\ref{lemma:pitimef}. By unicity, we have that 
  \[\phi =\{(\fun{f}{a}{N},\fun{f'}{a'}{N'}) \mid \sameCComp{\isOf{a}{\sem{A}}}{[\fun{f}{a}{N}/f]N}{[\fun{f}{a}{N'}/f]N'}{\sem{\isOf{a}{\alpha}.(B,B)}}{P} \}\] 
  Next, we need to show $\phi(\fun{a}{f}{N}, \fun{a}{f}{N'})$. It suffices to show 
  \[
    \sameCComp{\isOf{a}{\sem{A}}}{[\fun{a}{f}{N}/f]N}{[\fun{a}{f}{N'}/f]N'}{\sem{\isOf{a}{\alpha}.(B,B)}}{P}
  \]
  Let $\alpha = \sem{A}$, $\beta = \sem{\isOf{a}{\alpha}.(B,B)}$, $F = \fun{a}{f}{N}$ and $F' = \fun{a}{f}{N'}$.
  Let $\mathbb{P} : \dom{\alpha} \to 2$ be define as $\mathbb{P}(V)$ iff
  \[\forall V' : \dom{\alpha}, \alpha(V,V') \implies \sameCComp{}{[V/a,F/f]N}{[V'/a,F'/f]N'}{\beta_{V,V'}}{[V/a]P}\]
  Hence it suffices to show $\mathbb{P}(V)$ for all $V : \dom{\alpha}$. Proceed by induction on $(\dom{\alpha}, \lessdot_{\isOf{a}{\alpha}.P})$.
  Let $V : \dom{\alpha}$. Suppose that $\mathbb{P}(U)$ for all $\ltt{U}{\isOf{a}{\alpha}.P}{V}$. We need to show 
  $\mathbb{P}(V)$. Let $V' : \dom{\alpha}$ such that $\alpha(V,V')$. We need to show 
  \sameCComp{}{[V/a,F/f]N}{[V'/a,F'/f]N'}{\beta_{V,V'}}{[V/a]P}. Since $\beta_{V,V'} = \sem{[V/a]B}$ by Lemma~\ref{lemma:func_per},
  the result follows from assumption 4, given that we show 
  \[\eqVal{\fun{a}{f}{N}}{\fun{a}{f}{N'}}{\arrtimev{\isOf{a}{\subsetty{\isOf{a'}{A}}{\rel{<}{[a'/a]P}{[V/a]P}}}}{B}{P}}\]
  First, we need to show there exists some $\rho$ such that
  \[\sameType{\arrtimev{\isOf{a}{\subsetty{\isOf{a'}{A}}{\rel{<}{[a'/a]P}{[V/a]P}}}}{B}{P}}{\arrtimev{\isOf{a}{\subsetty{\isOf{a'}{A}}{\rel{<}{[a'/a]P}{[V/a]P}}}}{B}{P}}{\rho}{\tau_{\omega}}\]
  Let $A' =\subsetty{\isOf{a'}{A}}{\rel{<}{[a'/a]P}{[V/a]P}}$. 
  By Lemma~\ref{lemma:pitimef}, it suffices to show 
  \begin{enumerate}
    \item \eqComp{A'}{A'}{\univ{i}}
    \item \openEqComp{\isOf{a}{A'}}{B}{B}{\univ{i}}
    \item \openEqComp{\isOf{a}{A'}}{P}{P}{\nat}
  \end{enumerate}
  
  \noindent\underline{\eqComp{A'}{A'}{\univ{i}}}. We need to show that there exists some $\sigma$ such that 
  \[
    \sameType{\subsetty{\isOf{a'}{A}}{\rel{<}{[a'/a]P}{P}}}{\subsetty{\isOf{a'}{A}}{\rel{<}{[a'/a]P}{P}}}{\sigma}{\tau_{\omega}}
  \]
  By Lemma~\ref{lemma:subsetf}, it suffices to show 
  \begin{enumerate}
    \item \eqComp{A}{A}{\univ{i}}
    \item \openEqComp{\isOf{a'}{A}}{\rel{<}{[a'/a]P}{[V/a]P}}{\rel{<}{[a'/a]P}{[V/a]P}}{\univ{i}}
  \end{enumerate}
  The former is assumption 1. For the latter, by Lemma~\ref{lemma:relf}, it suffices to show 
  \begin{enumerate}
    \item \openEqComp{\isOf{a'}{A}}{[a'/a]P}{[a'/a]P}{\nat} 
    \item \openEqComp{\isOf{a'}{A}}{[V/a]P}{[V/a]P}{\nat}
  \end{enumerate}
  Both follow from assumption 3.
  By Lemma~\ref{lemma:open_sem}, let the open semantics of \openEqType{\isOf{a'}{A}}{\rel{<}{[a'/a]P}{[V/a]P}}{\rel{<}{[a'/a]P}{[V/a]P}}
  be 
  \[
    \kappa = \sem{\isOf{a'}{\alpha}.(\rel{<}{[a'/a]P}{[V/a]P},\rel{<}{[a'/a]P}{[V/a]P})}
  \] 
  By Lemma~\ref{lemma:func_per} and unicity, we have the following for all $\alpha(U,U')$:
  \begin{align*}
    \kappa_{U,U'} &= \kappa_{U,U}\\
    & = \sem{\rel{<}{[U/a]P}{[V/a]P}} \\
    & = \{ (\triv,\triv) \mid \exists m,n.\, m < n \land \sameComp{}{[U/a]P}{\bar{m}}{\omega} \land \sameComp{}{[V/a]P}{\bar{n}}{\omega} \}
  \end{align*}
  Hence by unicity, we know the semantics $\sigma$ of the subset type must be the following:
  \[
    \sigma = \{(U,U') \mid \alpha(U,U') \land \exists W,W'.\kappa_{U,U'}(W,W')\}
  \]
  \noindent\underline{\openEqType{\isOf{a}{A'}}{B}{B}}. Suppose \eqVal{U}{U'}{A'}. We need to show 
  \eqType{[U/a]B}{[U'/a]B}. This follows from the assumption 2, given that we know \eqVal{U}{U'}{A}. 
  Note that $\sem{A'} = \sigma$ and $\sem{A} = \alpha$, and $\sigma(U,U') \implies \alpha(U,U')$ by definition of $\sigma$.
  Hence we have \eqVal{U}{U'}{A}.\\

  \noindent\underline{\openEqComp{\isOf{a}{A'}}{P}{P}{\nat}}. Similar to above.\\

  Similarly, by unicity, we know the semantics $\rho$ of the funtime type must be the following:
  \[
    \rho = \{(\fun{f}{a}{N},\fun{f'}{a'}{N'}) \mid \sameCComp{\isOf{a}{\sigma}}{[F/f]N}{[F'/f]N'}{\beta}{P} \}
  \]
  Lastly, we need to show 
  \[
      \rho(\fun{a}{f}{N}, \fun{a}{f}{N'})
  \]
  It suffices to show 
  \[
    \sameCComp{\isOf{a}{\sigma}}{[F/f]N}{[F'/f]N'}{\beta}{P}  
  \]
  Let $\sigma(U,U')$. We need to show 
  \sameCComp{}{[U/a,F/f]N}{[U'/a,F'/f]N'}{\beta_{U,U'}}{[U/a]P}. 
  Note this follows from the assumption that $\mathbb{P}(U)$, given that we show \ltt{U}{\isOf{a}{\alpha}.P}{V}. 
  Let \eval{[U/a]P}{\bar{p'}} and \eval{[V/a]P}{\bar{p}}. It suffices to show $p' < p$. 
  By definition of $\sigma$, we know there exists $W,W'$ such that $\kappa_{U,U'}(W,W')$. By definition of $\kappa$, 
  this means there are $m,n$ such that 
  \begin{enumerate}
    \item $m < n$
    \item \sameComp{}{[U/a]P}{\bar{m}}{\omega}
    \item \sameComp{}{[V/a]P}{\bar{n}}{\omega}
  \end{enumerate}
  Hence $m = p'$, $n = p$, and $p' < p$.
\end{proof}

\begin{lemma}[Funtime Elimination]\label{lemma:pitimee}
  If 
  \begin{enumerate}
    \item \eqVal{\fun{a}{f}{N}}{\fun{a}{f}{N'}}{\arrtimev{\isOf{a}{A}}{B}{P}}
    \item \eqVal{V}{V'}{A}
  \end{enumerate}
  then \eqCComp{\ap{\fun{a}{f}{N}}{V}}{\ap{\fun{a}{f}{N'}}{V'}}{[V/a]B}{\suc{[V/a]P}}.
\end{lemma}

\begin{proof}
  By assumption, \sameType{\arrtimev{\isOf{a}{A}}{B}{P}}{\arrtimev{\isOf{a}{A}}{B}{P}}{\rho}{\tau_{\omega}}, which implies there exists $\alpha,\beta$
  such that \sameType{A}{A}{\alpha}{\tau_{\omega}} and \sameTypes{\isOf{a}{\alpha}}{B}{B}{\beta}{\tau_{\omega}}. 
  Furthermore, we have \sameComp{\isOf{a}{\alpha}}{P}{P}{\omega}. 
  First, we need to show \sameType{[V/a]B}{[V/a]B}{\phi}{\tau_{\omega}} for some $\phi$. 
  By assumption, $\alpha(V,V')$, so this holds by taking $\phi = \beta_{V,V'}$ and apply symmetry and transitivity.
  Next, we need to show \sameCComp{}{\ap{\fun{a}{f}{N}}{V}}{\ap{\fun{a}{f}{N'}}{V'}}{\beta_{V,V'}}{\suc{[V/a]P}}. 
  Let $F = \fun{a}{f}{N}$ and $F' = \fun{a}{f}{N'}$. Then
  \begin{align*}
    \ap{\fun{a}{f}{N}}{V} &\mapsto [V/a,F/f]N\\
    \ap{\fun{a}{f}{N'}}{V'} &\mapsto [V'/a,F'/f]N'
  \end{align*}
  So by Lemma~\ref{lemma:headexp1}, it suffices to show \sameCComp{}{[V/a,F/f]N}{[V/a,F'/f]N'}{\beta_{V,V'}}{[V/a]P}. 
  Since $\rho(F,F')$, by definition, we know: 
  \[
    \sameCComp{\isOf{a}{\alpha}}{[F/f]N}{[F'/f]N'}{\beta}{P}
  \]
  So \sameCComp{}{[V/a,F/f]N}{[V'/a,F'/f]N'}{\beta_{V,V'}}{[V/a]P}, as required.
\end{proof}
\subsection{Weakening}

Cost bounds can be relaxed: 

\begin{lemma}[Cost Weakening]\label{lemma:cost_weak}
  If 
  \begin{enumerate}
    \item \eqCComp{M}{M'}{A}{P}
    \item \eqComp{Q}{Q'}{\rel{\le}{P}{P'}}
  \end{enumerate}
  then \eqCComp{M}{M'}{A}{P'}.
\end{lemma}

\begin{proof}
  By assumption, \sameType{A}{A}{\alpha}{\tau_{\omega}} for some $\alpha$.
  It remains to show \sameCComp{}{M}{M'}{\alpha}{P'}. By assumption, 
  \sameType{\rel{\le}{P}{P'}}{\rel{\le}{P}{P'}}{\phi}{\tau_{\omega}}. By unicity, 
  \[\phi = \{ (\triv,\triv) \mid \exists m,n.\, m \le n \land \sameComp{}{P}{\bar{m}}{\omega} \land \sameComp{}{P'}{\bar{n}}{\omega} \}\]
  which implies\sameComp{}{P'}{P'}{\omega}. By assumption, we know that: 
  \begin{enumerate}
    \item \eval{M}{c}{V}
    \item \eval{M'}{c'}{V'}
    \item $\alpha(V,V')$
    \item \eval{P}{\bar{p}}
    \item $p \ge \max(c,c')$
  \end{enumerate}
  We need to show \eval{P'}{\bar{p'}} and $p' \ge \max(c,c')$. By assumption, 
  \eval{Q}{\triv}, \eval{Q'}{\triv}, and $\phi(\triv,\triv)$, so \eval{P'}{\bar{n}} and $n \ge m = p \ge \max(c,c')$. 
\end{proof}
\section{Machines}\label{sec:machines}

So far, we have considered a straightforward model of computation where each transition of the operation semantics 
is assigned unit cost. Now, we extend the operational semantics with additional primitive forms of computation that are assigned 
unit cost. This is similar to the instruction set of physical machines and represent operations that are considered
to be unit cost with the justification that they are usually implemented efficiently in practice. Without these primitives,
addition and multiplication defined internally would have linear and super-linear complexities, which seemingly goes against 
common intuition. The crux of the issue is that these operations have essentially constant time complexity, under the assumption
that the inputs fit within the \emph{word size} of a given machine. Hence, we index the operational semantics with a 
parameter $w$ representing the word size and introduce primitives that operate on word-sized inputs. 

Transition becomes a family of relations $\_ \mapsto_{w} \_ :
\mathcal{P}(\mathsf{Exp} \times \mathsf{Exp})$ over the parameter $w : \N$,
the maximum value of a number that fits in the word size. 
There are new syntactic forms \op{f}{M} and \arith{f}{M}{N} representing unary and binary arithmetic operators on natural numbers, respectively. 
They have the following associated transition steps:
\begin{mathpar}
\inferrule{
  \step{E}{E'}
}{
  \step{\op{f}{E}}{\op{f}{E'}}
}

\inferrule{
  m < w
}{
  \step{\op{f}{\bar{m}}}{\overline{f(m)}}
}  

\inferrule{
  \step{E_1}{E_1'}
}{
  \step{\arith{f}{E_1}{E_2}}{\arith{f}{E_1'}{E_1}}
}

\inferrule{
  \final{V}\\
  \step{E_2}{E_2'}
}{
  \step{\arith{f}{E_1}{E_2}}{\arith{f}{E_1}{E_2'}}
}

\inferrule{
  m,n < w
}{
  \step{\arith{f}{\bar{m}}{\bar{n}}}{\overline{f(m,n)}}
}  
\end{mathpar}

In the beta rules, inputs to the operator are required to fit in one word. What is ``reasonably'' primitive depends on the 
particular machine one has in mind, but we will postulate that $f$ can be addition, subtraction, multiplication, division, and
modulus.

Note that while we already have a foreign function interface via \cffone{f}{M} and \cfftwo{f}{M}{N}, these are intended to be 
used in the \emph{cost expressions} and not programs. As such, the the foreign function interface has no restriction on the 
inputs to the operators, and when subjected to complexity verification, invocation of the interface should be thought of 
as functions calls to an ``oracle'' that can compute any function in constant time. This is useful when one is interested 
in a more abstract measure of cost, such as the number of recursive calls a function makes as opposed to the actual running
time of said function. 

There is a new associated elimination rule for primitive computations on natural numbers:

\begin{lemma}[Primitive elimination]\label{lemma:arithe}
  If 
  \begin{enumerate}
    \item If \eqVal{V}{V'}{\nat} and \isComp{P}{\rel{<}{V}{\bar{w}}}, then \eqCComp{\op{f}{V}}{\op{f}{V'}}{\nat}{\bar{1}}.
    \item If \eqVal{V}{V'}{\nat}, \eqVal{U}{U'}{\nat}, \isComp{P}{\rel{<}{V}{\bar{w}}}, \isComp{Q}{\rel{<}{U}{\bar{w}}},
    then \eqCComp{\arith{f}{V}{U}}{\arith{f}{V}{U}}{\nat}{\bar{1}}.
  \end{enumerate}
\end{lemma}

\begin{proof}
  Similar to proof of Lemma~\ref{lemma:ffe}.
\end{proof}
\section{Example: greatest common divisor}\label{sec:gcd}

We are now equipped to state and analyze the computational complexity for a general class of programs. 
We focus on analyzing the complexity of Euclid's algorithm for computing the gcd of two natural numbers. 
This is a prototypical example of an algorithm which is \emph{not} structurally recursive, and we demonstrate that 
the semantic framework of \cctt{} is capable of handling such recursive constructions.
\begin{remark}
For presentation purposes, we derive a basic bound and acknowledge that it's
not the tightest characterization. It is certainly possible to carry out a more detailed analysis in \cctt{}. 
\end{remark}

Let the ternary predicate $\textit{gcdProp} : \mathcal{P}(\N \times \N \times \N)$ be defined as $\textit{gcd}(d,m,n)$ if $d$ is the gcd of 
$m$ and $n$. This is then internalized as \tern{gcdProp}{D}{M}{N} for \isComp{D,M,N}{\nat}. It is possible to define this 
predicate internally, but it is convenient to delegate this routine verification to the metatheory. Let the gcd algorithm
be defined as follows:
\begin{align*}
    \textit{gcd} \triangleq\,&\fun{a}{f}{\\&\letcst{\fst{a}}{x}{\\&\letcst{\snd{a}}{y}{
      \\&\ifz{x}{y}{x'}{
        \letcst{\arith{\%}{y}{\suc{x'}}}{z}{\ap{f}{\pair{z}{\suc{x'}}}}}}}}
\end{align*}
Let 
\begin{align*}
    A &= \sigmaty{\subsetty{\isOf{x}{\nat}}{\rel{<}{x}{\bar{w}}}}{\subsetty{\isOf{y}{\nat}}{\rel{<}{y}{\bar{w}}}}\\
    B &= \lsub{\fst{a}}{x}{\lsub{\snd{a}}{y}{\subsetty{\isOf{d}{\nat}}{\tern{gcdProp}{d}{x}{y}}}}\\
    P &= \bar{1} \hat{+} \suc{\lsub{\fst{a}}{x}{\bar{1} \hat{+} \suc{ \lsub{\snd{a}}{y}{\ifz{x}{\bar{1}}{x'}{\suc{\bar{1} \hat{+} \suc{\suc{\bar{8}\hat{\times}\suc{x'}}}}} }}}}
\end{align*}
We will verify the following specification: 
\[
\isVal{\textit{gcd}}{\arrtimev{\isOf{a}{A}}{B}{P}}
\]
\begin{remark}
Note that we deliberately program in a fashion commonly referred to as monadic
form (or sometimes A-normal form)
even though the operational semantics places no 
restriction on the presentation of programs. This is because the proof theory was arranged
to facilitate verification of programs of this form. To verify programs written in other styles, 
different collections of lemmas would be appropriate. For the simplicity of presentation, we choose to
program in ANF because otherwise the rule for sequencing (Lemma~\ref{lemma:seq}) would be essentially duplicated 
in each elimination rule. 
\end{remark}

\begin{proof}
  By Lemma~\ref{lemma:pitimei}, it suffices to show
  \begin{enumerate}
    \item \isComp{A}{\univ{i}}
    \item \openComp{\isOf{a}{A}}{B}{\univ{i}}
    \item \openComp{\isOf{a}{A}}{P}{\nat}
    \item \openCComp{\isOf{a}{A}, \isOf{f}{\arrtimev{\isOf{a}{\subsetty{\isOf{a'}{A}}{\rel{<}{[a'/a]P}{P}}}}{B}{P}}}{N}{B}{P}
  \end{enumerate}
  For the first 3,
  \begin{enumerate}
    \item By Lemma~\ref{lemma:sigmaf}, it suffices to show 
    \begin{enumerate}
      \item \isComp{\subsetty{\isOf{x}{\nat}}{\rel{<}{x}{\bar{w}}}}{\univ{i}}. By Lemma~\ref{lemma:subseti}, it suffices to show 
      \begin{enumerate}
        \item \isComp{\nat}{\univ{i}}. By Lemma~\ref{lemma:natf}.
        \item \openComp{\isOf{x}{\nat}}{\rel{<}{x}{\bar{w}}}{\univ{i}}. By Lemma~\ref{lemma:relf}, it suffices to show 
        \begin{enumerate}
          \item \openComp{\isOf{x}{\nat}}{x}{\nat}. By Lemma~\ref{lemma:hypothesis}.
          \item \openComp{\isOf{x}{\nat}}{\bar{w}}{\nat}. By Lemma~\ref{lemma:nati}.
        \end{enumerate}
      \end{enumerate}
      \item \isComp{\subsetty{\isOf{y}{\nat}}{\rel{<}{y}{\bar{w}}}}{\univ{i}}. As above.
    \end{enumerate}
    \item 
    By Lemma~\ref{lemma:seq}, it suffices to show 
    \begin{enumerate}
      \item \openComp{\isOf{a}{A}}{\fst{a}}{\nat}. This follows by Lemma~\ref{lemma:sigmae},~\ref{lemma:subsete}, and ~\ref{lemma:hypothesis}.
      \item \openComp{\isOf{a}{A},\isOf{x}{\nat}}{\lsub{\snd{a}}{y}{\subsetty{\isOf{d}{\nat}}{\tern{\textit{gcdProp}}{d}{x}{y}}}}{\univ{i}}.
    \end{enumerate}
    For the latter, by Lemma~\ref{lemma:seq}, it suffices to show 
    \begin{enumerate}
      \item \openComp{\isOf{a}{A},\isOf{x}{\nat}}{\snd{a}}{\nat}. This follows by Lemma~\ref{lemma:sigmae},~\ref{lemma:subsete},~\ref{lemma:open_headexp}, and ~\ref{lemma:hypothesis}.
      \item \openComp{\isOf{a}{A},\isOf{x}{\nat},\isOf{y}{\nat}}{\subsetty{\isOf{d}{\nat}}{\tern{\textit{gcdProp}}{d}{x}{y}}}{\univ{i}}.
    \end{enumerate}
    For the latter, let $\Gamma = \isOf{a}{A},\isOf{x}{\nat},\isOf{y}{\nat}$.
    By Lemma~\ref{lemma:subsetf}, it suffices to show 
    \begin{enumerate}
      \item \openComp{\Gamma}{\nat}{\univ{i}}. This holds by Lemma~\ref{lemma:natf}.
      \item \openComp{\Gamma,\isOf{d}{\nat}}{\tern{\textit{gcdProp}}{d}{x}{y}}{\univ{i}}. By Lemma~\ref{lemma:relf}, it suffices to show 
      \begin{enumerate}
        \item \openComp{\Gamma,\isOf{d}{\nat}}{d}{\nat}. This holds by Lemma~\ref{lemma:hypothesis}.
        \item \openComp{\Gamma,\isOf{d}{\nat}}{x}{\nat}. This holds by Lemma~\ref{lemma:hypothesis}.
        \item \openComp{\Gamma,\isOf{d}{\nat}}{y}{\nat}. This holds by Lemma~\ref{lemma:hypothesis}.
      \end{enumerate}
    \end{enumerate}
   \item By Lemma~\ref{lemma:ffe}, it suffices to show 
   \begin{enumerate}
     \item \openComp{\isOf{a}{A}}{\bar{1}}{\nat}. By Lemma~\ref{lemma:nati}. 
     \item \openComp{\isOf{a}{A}}{\suc{\lsub{\fst{a}}{x}{\bar{1} \hat{+} \suc{ \lsub{\snd{a}}{y}{\ifz{x}{\bar{1}}{x'}{\suc{\bar{1} \hat{+} \suc{\suc{\bar{8}\hat{\times}\suc{x'}}}}} }}}}}{\nat}. 
   \end{enumerate}
   For the latter, by Lemma~\ref{lemma:nati}, it suffices to show 
     \[
       \openComp{\isOf{a}{A}}{\lsub{\fst{a}}{x}{\bar{1} \hat{+} \suc{ \lsub{\snd{a}}{y}{\ifz{x}{\bar{1}}{x'}{\suc{\bar{1} \hat{+} \suc{\suc{\bar{8}\hat{\times}\suc{x'}}}}} }}}}{\nat}. 
     \]
     By Lemma~\ref{lemma:seq}, it suffices to show 
     \begin{enumerate}
       \item \openComp{\isOf{a}{A}}{\fst{a}}{\nat}. Similar to above.
       \item \openComp{\isOf{a}{A},\isOf{x}{\nat}}{\bar{1} \hat{+} \suc{ \lsub{\snd{a}}{y}{\ifz{x}{\bar{1}}{x'}{\suc{\bar{1} \hat{+} \suc{\suc{\bar{8}\hat{\times}\suc{x'}}}}} }}}{\nat}. 
     \end{enumerate}
     For the latter, by Lemma~\ref{lemma:ffe}, it suffices to show 
     \begin{enumerate}
       \item \openComp{\isOf{a}{A},\isOf{x}{\nat}}{\bar{1}}{\nat}. Lemma~\ref{lemma:nati}. 
       \item \openComp{\isOf{a}{A},\isOf{x}{\nat}}{\suc{\lsub{\snd{a}}{y}{\ifz{x}{\bar{1}}{x'}{\suc{\bar{1} \hat{+} \suc{\suc{\bar{8}\hat{\times}\suc{x'}}}}} }}}{\nat}. 
     \end{enumerate}
     For the latter, by Lemma~\ref{lemma:nati}, it suffices to show 
     \[
       \openComp{\isOf{a}{A},\isOf{x}{\nat}}{\lsub{\snd{a}}{y}{\ifz{x}{\bar{1}}{x'}{\suc{\bar{1} \hat{+} \suc{\suc{\bar{8}\hat{\times}\suc{x'}}}}}}}{\nat}
    \]
     By Lemma~\ref{lemma:seq}, it suffices to show 
     \begin{enumerate}
       \item \openComp{\isOf{a}{A},\isOf{x}{\nat}}{\snd{a}}{\nat}. Similar to above.
       \item \openComp{\isOf{a}{A},\isOf{x}{\nat},\isOf{y}{\nat}}{\ifz{x}{\bar{1}}{x'}{\suc{\bar{1} \hat{+} \suc{\suc{\bar{8}\hat{\times}\suc{x'}}}}}}{\nat}.
     \end{enumerate}
     Let $\Gamma =\isOf{a}{A},\isOf{x}{\nat},\isOf{y}{\nat}$. For the latter, by Lemma~\ref{lemma:nate1} and~\ref{lemma:open_headexp}, it suffices to show 
     \begin{enumerate}
       \item \openComp{\Gamma,\isOf{z}{\nat}}{\nat}{\univ{i}}. By Lemma~\ref{lemma:natf}.
       \item \openComp{\Gamma}{x}{\nat}. By Lemma~\ref{lemma:hypothesis}. 
       \item \openComp{\Gamma}{\bar{1}}{\nat}. By Lemma~\ref{lemma:nati}. 
       \item \openComp{\Gamma,\isOf{z}{\nat}, \isOf{p}{\eqty{\nat}{\suc{z}}{x}}}{\suc{\bar{1} \hat{+} \suc{\suc{\bar{8}\hat{\times}\suc{z}}}}}{\nat}.
       By repeated application of Lemma~\ref{lemma:nati},~\ref{lemma:ffe}, and~\ref{lemma:hypothesis}.
     \end{enumerate}
  \end{enumerate}
  For the last premise, let 
  \begin{align*}
    C &= \arrtimev{\isOf{a}{\subsetty{\isOf{a'}{A}}{\rel{<}{[a'/a]P}{P}}}}{B}{P}\\ 
    N &= \letcst{\fst{a}}{x}{\\&\letcst{\snd{a}}{y}{
      \\&\ifz{x}{y}{x'}{
        \letcst{\arith{\%}{y}{\suc{x'}}}{z}{\ap{f}{\pair{z}{\suc{x'}}}}}}}
  \end{align*} 
  It suffices to show 
  \[
  \openCComp{\isOf{a}{A}, \isOf{f}{C}}{N}{B}{P}
  \]
  Let 
  \begin{align*}
  N_1 &= \letcst{\snd{a}}{y}{
      \ifz{x}{y}{x'}{
        \letcst{\arith{\%}{y}{\suc{x'}}}{z}{\ap{f}{\pair{z}{\suc{x'}}}}}}\\
  B_1 &= \lsub{\snd{a}}{y}{\subsetty{\isOf{d}{\nat}}{\tern{\textit{gcdProp}}{d}{x}{y}}}\\
  P_1 &=\bar{1} \hat{+} \suc{ \lsub{\snd{a}}{y}{\ifz{x}{\bar{1}}{x'}{\suc{\bar{1} \hat{+} \suc{\suc{\bar{8} \hat{\times} \suc{x'}}}}}} }\\
  \Gamma_1 &= \isOf{a}{A},\isOf{f}{C}
  \end{align*}
  By Lemma~\ref{lemma:seq}, it suffices to show: 
  \begin{enumerate}
    \item \openCComp{\Gamma_1}{\fst{a}}{\subsetty{\isOf{x}{\nat}}{\sigmaty{\rel{<}{x}{\bar{w}}}{\eqty{\nat}{x}{\fst{a}}}}}{\bar{1}}.
    By Lemma~\ref{lemma:subseti}, it suffices to show: 
    \begin{enumerate}
      \item \openCComp{\Gamma_1}{\fst{a}}{\nat}{\bar{1}}. 
        This follows from Lemma~\ref{lemma:subsete} and~\ref{lemma:sigmae}.
      \item For all \eqVal{V}{V'}{A} and \eqVal{F}{F'}{[V/a]C}, there exists $N,N'$ such that 
      \eqComp{N}{N'}{\lsub{\fst{V}}{x}{\sigmaty{\rel{<}{x}{\bar{w}}}{\eqty{\nat}{x}{\fst{V}}}}}.
      By definition, $V = \pair{V_1}{V_2}$ for some \isComp{V_1}{\subsetty{\isOf{x}{\nat}}{\rel{<}{x}{\bar{w}}}} 
      and \isComp{V_2}{\subsetty{\isOf{y}{\nat}}{\rel{<}{y}{\bar{w}}}}.  
      By Lemma~\ref{lemma:headexp}, it suffices to show
      \[\isComp{N}{\sigmaty{\rel{<}{V_1}{\bar{w}}}{\eqty{\nat}{V_1}{\fst{V}}}}\]
      By Lemma~\ref{lemma:sigmai}, it suffices to show 
      \begin{enumerate}
        \item \isComp{\triv}{\rel{<}{V_1}{\bar{w}}}. This follows from Lemma~\ref{lemma:subsete}.
        \item \isComp{\triv}{\eqty{\nat}{V_1}{\fst{V}}}. This follows from Lemma~\ref{lemma:eqi} and~\ref{lemma:headexp}.
      \end{enumerate}
    \end{enumerate}
    \item \openCComp{\Gamma_1,\isOf{x}{\subsetty{\isOf{x}{\nat}}{ \sigmaty{\rel{<}{x}{\bar{w}}}{\eqty{\nat}{x}{\fst{a}}} }}}{N_1}{B_1}{P_1}.
  \end{enumerate}
  For the last part, let 
  \begin{align*}
    N_2 &= \ifz{x}{y}{x'}{
        \letcst{\arith{\%}{y}{\suc{x'}}}{z}{\ap{f}{\pair{z}{\suc{x'}}}}}\\
    B_2 &= \subsetty{\isOf{d}{\nat}}{\tern{\textit{gcdProp}}{d}{x}{y}}\\
    P_2 &= \ifz{x}{\bar{1}}{x'}{\suc{\bar{1} \hat{+} \suc{\suc{\bar{8} \hat{\times} \suc{x'}}}}}\\
    \Gamma_2 &=\Gamma_1,\isOf{x}{\subsetty{\isOf{x}{\nat}}{\sigmaty{\rel{<}{x}{\bar{w}}}{\eqty{\nat}{x}{\fst{a}}}}} 
  \end{align*}
  By Lemma~\ref{lemma:seq}, it suffices to show 
  \begin{enumerate}
    \item \openCComp{\Gamma_2}{\snd{a}}{\subsetty{\isOf{y}{\nat}}{\rel{<}{y}{\bar{w}}}}{\bar{1}}.
    Similar to above.
    \item \openCComp{\Gamma_2,\isOf{y}{\subsetty{\isOf{y}{\nat}}{\sigmaty{\rel{<}{y}{\bar{w}}}{\eqty{\nat}{y}{\snd{a}}}}}}{N_2}{B_2}{P_2}.
  \end{enumerate}
  For the last part, let 
  \begin{align*}
    B_3 &= [n/x]B_2\\
    \Gamma_3 &= \Gamma_2,\isOf{y}{\subsetty{\isOf{y}{\nat}}{\sigmaty{\rel{<}{y}{\bar{w}}}{\eqty{\nat}{y}{\snd{a}}}}}
  \end{align*}
  By Lemma~\ref{lemma:nate2}, it suffices to show 
\begin{enumerate}
    \item \openComp{\Gamma_3,\isOf{n}{\nat}}{B_3}{\univ{i}}. Similar to derivation above.
    \item \openCComp{\Gamma_3}{x}{\nat}{\zero}. By Lemma~\ref{lemma:hypothesis}.
    \item \openComp{\Gamma_3}{\bar{1}}{\nat}. By Lemma~\ref{lemma:nati}. 
    \item \openComp{\Gamma_3,\isOf{n}{\nat}}{\bar{1} \hat{+} \suc{\suc{\bar{8} \times \suc{n}}}}{\nat}. By 
    Lemma~\ref{lemma:ffe} and~\ref{lemma:nati}. 
    \item \openCComp{\Gamma_3}{y}{[\zero/n]B_3}{\zero}.
    Let \eqInst{\gamma}{\gamma'}{\Gamma_3}. We know that 
      \[\gamma(y) = \gamma'(y) = \bar{y}\]
    For some $y$. By Lemma~\ref{lemma:subseti}, suffices to show 
     \begin{enumerate}
       \item \isComp{\bar{y}}{\nat}. This follows from Lemma~\ref{lemma:nati}.
       \item \isComp{\triv}{\tern{\textit{gcdProp}}{\bar{y}}{\zero}{\bar{y}}}. By Lemma~\ref{lemma:reli}, it suffices to show 
       $\textit{gcd}(0,y) = y$, which holds.
     \end{enumerate}
    \item \openCComp{\Gamma_3,\isOf{n}{\nat},\isOf{r}{\eqty{\nat}{\suc{n}}{x}}}{\letcst{\arith{\%}{y}{\suc{n}}}{z}{\ap{f}{\pair{z}{\suc{n}}}}\\}
    {[\suc{n}/n]B_3}{\bar{1} \hat{+} \suc{\bar{8} \hat{\times} \suc{n}}}.
  \end{enumerate}
  For the last part, let 
  \begin{align*}
    \Gamma_4 &= \Gamma_3,\isOf{n}{\nat},\isOf{r}{\eqty{\nat}{\suc{n}}{x}}
  \end{align*}
  By Lemma~\ref{lemma:open_headexp}, 
  \begin{enumerate}
    \item \openEqComp{\Gamma_4}{\lsub{\arith{\%}{y}{\suc{n}}}{z}{[\suc{n}/n]B_3}}{[\suc{n}/n]B_3}{\univ{i}}
    \item \openEqComp{\Gamma_4}{\lsub{\arith{\%}{y}{\suc{n}}}{z}{\suc{\bar{8} \hat{\times} \suc{n}}}}{\suc{\bar{8} \hat{\times}\suc{n}}}{\nat}
  \end{enumerate}
  So by Lemma~\ref{lemma:seq}, it suffices to show
  \begin{enumerate}
    \item \openCComp{\Gamma_4}{\arith{\%}{y}{\suc{n}}}{\subsetty{\isOf{z}{\nat}}{\sigmaty{\rel{<}{z}{\bar{w}}}{\arith{\%}{y}{\suc{n}}}}}{\bar{1}}.
    Let \eqInst{\gamma}{\gamma'}{\Gamma_4}. Then we know 
    \begin{enumerate}
      \item $\gamma(y) = \gamma'(y) = \bar{y}$ for some $y$ such that $y < w$.
      \item $\gamma(x) = \gamma'(x) = \bar{x}$ for some $x$ such that $x = n+1$ and $x<w$.
      \item $\gamma(n) = \gamma'(n) = \bar{n}$ for some $n$.
    \end{enumerate}
    It suffices to show 
    \[
    \isCComp{\arith{\%}{\bar{y}}{\suc{\bar{n}}}}{\subsetty{\isOf{z}{\nat}}{\sigmaty{\rel{<}{z}{\bar{w}}}{\eqty{\nat}{z}{\arith{\%}{y}{\suc{n}}}}}}{\bar{1}}
    \]
    By Lemma~\ref{lemma:subseti}, it suffices to show 
    \begin{enumerate}
      \item \isCComp{\arith{\%}{\bar{y}}{\suc{\bar{n}}}}{\nat}{\bar{1}}. By Lemma~\ref{lemma:arithe}.
      \item \isVal{\pair{\triv}{\triv}}{\lsub{\arith{\%}{\bar{y}}{\suc{\bar{n}}}}{z}{\sigmaty{\rel{<}{z}{\bar{w}}}{\eqty{\nat}{z}{\arith{\%}{y}{\suc{n}}}}}}.
      Let $k = y \% (n+1)$. By Lemma~\ref{lemma:headexp}, it suffices to show
      \[
        \isVal{\pair{\triv}{\triv}}{\sigmaty{\rel{<}{\bar{k}}{\bar{w}}}{\eqty{\nat}{\bar{k}}{\arith{\%}{y}{\suc{n}}}}}
      \]
      By Lemma~\ref{lemma:sigmai}, it suffices to show 
      \begin{enumerate}
        \item \isComp{\triv}{\rel{<}{\bar{k}}{\bar{w}}}. By Lemma~\ref{lemma:reli}, it suffices to show $k < w$, which holds since $k = y \% (n+1) < n+1 < w$.
        \item \isComp{\triv}{\eqty{\nat}{\bar{k}}{\arith{\%}{y}{\suc{n}}}}. By Lemma~\ref{lemma:nati} and~\ref{lemma:headexp}. 
      \end{enumerate}
    \end{enumerate}

    \item \openCComp{\Gamma_4,\isOf{z}{\subsetty{\isOf{z}{\nat}}{\sigmaty{\rel{<}{z}{w}}{\eqty{\nat}{z}{\arith{\%}{y}{\suc{n}}}}}}}{\ap{f}{\pair{z}{\suc{n}}}}{[\suc{n}/n]B_3}{\suc{\bar{8} \hat{\times} \suc{n}}}.
  \end{enumerate}
  For the last part, let 
  \begin{align*}
  \Gamma_5 &= \Gamma_4,\isOf{z}{\subsetty{\isOf{z}{\nat}}{\sigmaty{\rel{<}{z}{w}}{\eqty{\nat}{z}{\arith{\%}{y}{\suc{n}}}}}}
  \end{align*}
  Applying Lemma~\ref{lemma:pitimee} with the following, 
  \begin{enumerate}
    \item \openCComp{\Gamma_5}{f}{\arrtimev{\isOf{a}{\subsetty{\isOf{a'}{A}}{\rel{<}{[a'/a]P}{P}}}}{B}{P}}{\zero}. This follows 
    from Lemma~\ref{lemma:weaken} and Lemma~\ref{lemma:hypothesis}.
    \item \openCComp{\Gamma_5}{\pair{z}{\suc{n}}}{\subsetty{\isOf{a'}{A}}{\rel{<}{[a'/a]P}{P}}}{\zero}.
    Let \eqInst{\gamma}{\gamma'}{\Gamma_5}. We need to show 
    \eqVal{\hat\gamma\pair{z}{\suc{n}}}{\hat{\gamma'}\pair{z}{\suc{n}}}{\hat\gamma \subsetty{\isOf{a'}{A}}{\rel{<}{[a'/a]P}{P}} }. 
    We know that 
    \begin{enumerate}
      \item $\gamma(z) = \gamma'(z) = \bar{z}$ for some $z$.
      \item $\gamma(n) = \gamma'(n) = \bar{n}$ for some $n$.
      \item $\gamma(x) = \bar{x}$ for some $x$ and \isComp{\_}{\rel{<}{\bar{x}}{\bar{w}}}, so $x < w$.
      \item $\gamma(y) = \bar{y}$ for some $y$ and \isComp{\_}{\rel{<}{\bar{y}}{\bar{w}}}, so $y < w$.
      \item $\gamma(a) = \pair{\bar{a_1}}{\bar{a_2}}$.
      \item $\gamma(p) = \triv$, where \isVal{\triv}{\eqty{\nat}{\bar{x}}{\fst{\pair{\bar{a_1}}{\bar{a_2}}}}}, so $x = a_1$.
      \item $\gamma(q) = \triv$, where \isVal{\triv}{\eqty{\nat}{\bar{y}}{\snd{\pair{\bar{a_1}}{\bar{a_2}}}}}, so $y = a_2$.
      \item $\gamma(s) = \triv$, where \isVal{\triv}{\eqty{\nat}{\bar{z}}{\arith{\%}{\bar{y}}{\bar{x}}}}, so $z = y \% x$
      \item $\gamma(r) = \triv$, where \isVal{\triv}{\eqty{\nat}{\suc{\bar{n}}}{\bar{x}}}, so $x = n+1$.
    \end{enumerate}
    By Lemma~\ref{lemma:subseti}, it suffices to show 
    \begin{enumerate}
      \item \isVal{\pair{\bar{z}}{\suc{\bar{n}}}}{\hat\gamma A}. By Lemma~\ref{lemma:sigmai}, it suffices to show 
        \begin{enumerate}
          \item \isVal{\bar{z}}{\subsetty{\isOf{x}{\nat}}{\rel{<}{x}{\bar{w}}}}. By Lemma~\ref{lemma:subseti}, it suffices to show 
          \isVal{\triv}{\rel{<}{\bar{z}}{\bar{w}}}, which holds since $z = y\% x < x < w$.
          \item \isVal{\suc{\bar{n}}}{\subsetty{\isOf{y}{\nat}}{\rel{<}{y}{\bar{w}}}}. By Lemma~\ref{lemma:subseti}, it suffices to show 
          \isVal{\triv}{\rel{<}{\suc{\bar{n}}}{\bar{w}}}, which holds since $n+1 = x < w$.
        \end{enumerate}
      \item \isVal{\triv}{\rel{<}{[\pair{\bar{z}}{\suc{\bar{n}}}/a]\hat\gamma P}{\hat\gamma P}}. There are 
      two possibilities. If $z = 0$, we know that 
      \begin{align*}
        [\pair{\bar{z}}{\suc{\bar{n}}}/a]\hat\gamma P &\mapsto^* \bar{5}\\
        \hat\gamma P &\mapsto^* \overline{8 + 8 \times x}
      \end{align*}
      and it suffices to show \isVal{\triv}{\rel{<}{\bar{5}}{\overline{8 + 8 \times x}}}.
      Since $5 < 8 + 8 \times x$ for any $x$, we have the result by Lemma~\ref{lemma:reli}. 
      Otherwise $z = z'+1$ for some $z'$, and 
      \begin{align*}
        [\pair{\bar{z}}{\suc{\bar{n}}}/a]\hat\gamma P &\mapsto^* \overline{8 + 8\times(z'+1)}\\
        \hat\gamma P &\mapsto^* \overline{8 + 8 \times x}
      \end{align*}
      and it suffices to show \isVal{\triv}{\rel{<}{\overline{8 + 8\times(z'+1)}}{\overline{8 + 8 \times x}}}.
      Since $z' + 1 = z = y \% x < x$, this holds by Lemma~\ref{lemma:reli}. 
    \end{enumerate}
  \end{enumerate}
  We have 
  \[
   \openCComp{\Gamma_5}{\ap{f}{\pair{z}{\suc{n}}}}{[\pair{z}{\suc{n}}/a]B}{\suc{[\pair{z}{\suc{n}}/a]P}}
  \]
  Rewriting using equalities in $\Gamma_5$ and by Lemma~\ref{lemma:open_headexp}, we have 
  \begin{enumerate}
    \item \openEqComp{\Gamma_5}{[\pair{z}{\suc{n}}/a]B}{\subsetty{\isOf{d}{\nat}}{\tern{\textit{gcdProp}}{d}{\arith{\%}{y}{\suc{n}}}{\suc{n}}}}{\univ{i}}
    \item \openEqComp{\Gamma_5}{\suc{[\pair{z}{\suc{n}}/a]P}}{\suc{\bar{1} \hat{+} \suc{\bar{1} \hat{+} \suc{ \ifz{z}{\bar{1}}{z'}{\suc{\bar{1} \hat{+} \suc{\suc{\bar{8}\hat{\times}\suc{z'}}}}} }}}}{\nat}
  \end{enumerate}
  Furthermore, Given 
  \begin{align*}
    &\openComp{\Gamma_5}{D}{\subsetty{\isOf{d}{\nat}}{\tern{\textit{gcdProp}}{d}{\arith{\%}{y}{\suc{n}}}{\suc{n}}}}
  \end{align*}
  we have 
  \begin{align*}
    &\openComp{\Gamma_5}{D}{\subsetty{\isOf{d}{\nat}}{\tern{\textit{gcdProp}}{d}{\suc{n}}{y}}}
  \end{align*}
  because $\textit{gcdProp}(d,n+1,y) \iff \textit{gcdProp}(d,y\% (n+1), n+1)$.
  Let 
  \[
  P_3 = \suc{\bar{1} \hat{+} \suc{\bar{1} \hat{+} \suc{ \ifz{z}{\bar{1}}{z'}{\suc{\bar{1} \hat{+} \suc{\suc{\bar{8}\hat{\times}\suc{z'}}}}} }}}
  \]
  Lastly, by Lemma~\ref{lemma:cost_weak}, it suffices to show 
  \[
    \openComp{\Gamma_5}{\triv}{\rel{\le}{P_3}{\suc{\bar{8} \hat{\times} \suc{n}}}}
  \]
  Let \eqInst{\gamma}{\gamma'}{\Gamma_5}. It suffices to show 
  \isComp{\triv}{\rel{\le}{\hat\gamma P_3}{\suc{\bar{8} \hat{\times} \hat\gamma\suc{n}}}}.
  We know 
  \begin{enumerate}
    \item $\gamma(x) = \bar{x}$ for some $x$.
    \item $\gamma(y) = \bar{y}$ for some $y$.
    \item $\gamma(n) = \bar{n}$ for some $n$ such that $x = n+1$.
    \item $\gamma(z) = \bar{z}$ for some $z$ such that $z = y \% x$.
  \end{enumerate}
  There are two possibilities.  If $\gamma(z) = \zero$, then we have 
  \begin{align*}
    \hat\gamma P_3 &\mapsto^* \bar{6}\\
    \suc{\bar{8} \hat{\times} \suc{n}} &\mapsto \overline{8 \times (n+1)+1}
  \end{align*}
  Since $6 \le 8 \times (n+1)+1$ for any $n$, we have the result by Lemma~\ref{lemma:reli}.
  Otherwise, $\gamma(z) = \suc{\bar{z'}}$ for some $z'$, and 
  \begin{align*}
    \hat\gamma P_3 &\mapsto^* \overline{9 + (8 \times (z'+1))}\\
    \suc{\bar{8} \hat{\times} \suc{n}} &\mapsto \overline{8 \times (n+1)+1}
  \end{align*}
  Since $z'+1 = z = y \% x < x = n+1$ and so  $z'+1 \le n$,
  we have $9 + (8 \times (z'+1)) \le 9 + 8 \times n$, and again the result follows from 
  Lemma~\ref{lemma:reli}.
\end{proof}
\section{Parallelism}\label{sec:parallel}

So far we have developed a computational type theory for analyzing the computational complexity of sequential computations.
To extend the theory to account for parallel programs, we arrange the operational semantics so that
each transition step corresponds to one unit of cost in the \emph{span} of the computation. 
The idea is that sub-terms of a computation that are independent of each other transition simultaneously whenever possible. 
Intuitively, the resulting cost measures
the length of the longest sequence of \emph{dependencies} in the computation. Figure~\ref{fig:parallel_os} contains selected 
rules of the parallel operational semantics of \pl{}.
The omitted rules remain the same from the sequential operational semantics.
The unary sequential composition is generalized to the \emph{binary} sequencing operator \parin{E_1}{E_2}{a}{b}{E_3}, 
which transitions $E_1$ and $E_2$ simultaneously and binds the results to $a$ and $b$ for use in $E_3$. Binary sequencing can 
be taken as a primitive. Here we define it in terms of unary sequencing: 
\[
\parin{E_1}{E_2}{a}{b}{E_3} \triangleq \letcst{\pair{E_1}{E_2}}{c}{\letcst{\fst{c}}{a}{\letcst{\snd{c}}{b}{E_3}}}  
\]
As before, there is the notion of evaluation:
\begin{enumerate}
  \item \evalCost{M}{c}{V} when $\stepIn{M}{c}{V}$ and \final{V}.
  \item \eval{M}{V} when there is a $c$ s.t. \evalCost{M}{c}{V}.
\end{enumerate}
However, $c$ now represents the span, i.e., the parallel complexity of the computation $M$. The construction of the type system 
$\tau_{\omega}$ and the associated proof theory is unaffected by the change of the operational semantics. 

There is an associated derived binary sequencing rule:
\begin{lemma}[Binary sequence]\label{lemma:bin_seq}
  If 
  \begin{enumerate}
    \item \eqCComp{M_1}{M_1'}{A_1}{P_1}
    \item \eqCComp{M_2}{M_2'}{A_2}{P_2}
    \item \openEqCComp{\isOf{a_1}{A_1},\isOf{a_2}{A_2}}{N}{N'}{B}{Q}
  \end{enumerate}
  then 
  \begin{align*}
    \eqCComp{\parin{M_1}{M_2}{a_1}{a_2}{N}}{\parin{M_1'}{M_2'}{a_1}{a_2}{N'}}{\lsub{M_1}{a_1}{\lsub{M_2}{a_2}{B}}\\}{\fop{\max}(P_1,P_2) \fbin{+} 
  \bar{5} \fbin{+} \lsub{M_1}{a_1}{\lsub{M_2}{a_2}{Q}} }
  \end{align*}
\end{lemma}

\begin{proof}
  By assumption, we have \sameCComp{}{M_1}{M_1'}{\sem{A_1}}{P_1} and \sameCComp{}{M_2}{M_2}{\sem{A_2}}{P_2}, which means
  \begin{enumerate}
    \item \evalCost{M_1}{c_1}{V_1}
    \item \evalCost{M_1'}{c_1'}{V_1'}
    \item $\sem{A_1}(V_1,V_1')$
    \item \eval{P_1}{\overline{p_1}}
    \item $p_1 \ge \max(c_1,c_1')$
    \item \evalCost{M_2}{c_2}{V_2}
    \item \evalCost{M_2'}{c_2'}{V_2'}
    \item $\sem{A_2}(V_2,V_2')$
    \item \eval{P_2}{\overline{p_2}}
    \item $p_2 \ge \max(c_2,c_2')$
  \end{enumerate}
  So \eqVal{V_1}{V_1'}{A_1} and \eqVal{V_2}{V_2'}{A_2}. By assumption 3, we have 
  \[\eqCComp{[V_1/a_1,V_2/a_2]N}{[V_1'/a_1,V_2'/a_2]N'}{[V_1/a_1,V_2/a_2]B}{[V_1/a_1,V_2/a_2]Q}\]
  By Lemma~\ref{lemma:headexp}, we have 
  \[\eqCComp{[V_1/a_1,V_2/a_2]N}{[V_1'/a_1,V_2'/a_2]N'}{\lsub{M_1}{a_1}{\lsub{M_2}{a_2}{B}}}{\lsub{M_1}{a_1}{\lsub{M_2}{a_2}{Q}}}\]
  Note that by definition, 
  \begin{align*}
    \parin{M_1}{M_2}{a_1}{a_2}{N} &\mapsto^c \letcst{\pair{V_1}{V_2}}{x}{\letcst{\fst{x}}{a_1}{\letcst{\snd{x}}{a_2}{N}}}\\
    &\mapsto \letcst{\fst{\pair{V_1}{V_2}}}{a_1}{\letcst{\snd{\pair{V_1}{V_2}}}{a_2}{N}}\\ 
    &\mapsto \letcst{V_1}{a_1}{\letcst{\snd{\pair{V_1}{V_2}}}{a_2}{N}}\\ 
    &\mapsto \letcst{\snd{\pair{V_1}{V_2}}}{a_2}{[V_1/a_1]N}\\ 
    &\mapsto \letcst{V_2}{a_2}{[V_1/a_1]N}\\ 
    &\mapsto [V_1/a_1,V_2/a_2]N
  \end{align*}
  Where $c = \max(c_1,c_2)$. So $\parin{M_1}{M_2}{a_1}{a_2}{N} \mapsto^{\max(c_1,c_2)+5} [V_1/a_1,V_2/a_2]N$. Similarly, 
  $\parin{M_1'}{M_2'}{a_1}{a_2}{N'} \mapsto^{\max(c_1',c_2')+5} [V_1'/a_1,V_2'/a_2]N'$. Furthermore, we have 
  \[
    \eval{\fop{\max}(P_1,P_2) \fbin{+} \bar{5}}{\max(p_1,p_2) + 5}
  \]
  and $\max(p_1,p_2) + 5 \ge \max(c_1,c_2) + 5$, so the result holds by Lemma~\ref{lemma:headexp}.
\end{proof}

\begin{figure}
  \begin{mathpar}
\inferrule{
    \step{E_1}{E_1'}\\
    \step{E_1}{E_2'}
}{
    \step{\ap{E_1}{E_2}}{\ap{E_1'}{E_2'}}
}

\inferrule{
    \final{E_1}\\
    \step{E_2}{E_2'}
}{
    \step{\ap{E_1}{E_2}}{\ap{E_1}{E_2'}}
}

\inferrule{
    \final{E_2}\\
    \step{E_1}{E_1'}
}{
    \step{\ap{E_1}{E_2}}{\ap{E_1'}{E_2}}
}

\inferrule{
    \final{E_2}
}{
    \step{(\fun{f}{a}{E}) E_2}{[\fun{f}{a}{E}/f,E_2/a]E}
}

\inferrule{
    \step{E_1}{E_1'}\\
    \step{E_2}{E_2'}
}{
    \step{\pair{E_1}{E_2}}{\pair{E_1'}{E_2'}}
}

\inferrule{
    \final{E_1}\\
    \step{E_2}{E_2'}
}{
    \step{\pair{E_1}{E_2}}{\pair{E_1}{E_2'}}
}

\inferrule{
    \final{E_2}\\
    \step{E_1}{E_1'}
}{
    \step{\pair{E_1}{E_2}}{\pair{E_1'}{E_2}}
}

\inferrule{
  \step{E_1}{E_1'}\\
  \step{E_2}{E_2'}
}{
  \step{\cfftwo{f}{E_1}{E_2}}{\cfftwo{f}{E_1'}{E_2'}}
}

\inferrule{
  \final{E_1}\\
  \step{E_2}{E_2'}
}{
  \step{\cfftwo{f}{E_1}{E_2}}{\cfftwo{f}{E_1}{E_2'}}
}

\inferrule{
  \final{E_2}\\
  \step{E_1}{E_1'}
}{
  \step{\cfftwo{f}{E_1}{E_2}}{\cfftwo{f}{E_1'}{E_2}}
}
\end{mathpar}
\caption{Selected rules for parallel operational semantics of  \pl{}.}
\label{fig:parallel_os}
\end{figure}
\section{Example: Fibonacci number}\label{sec:fib}

To demonstrate the type theory when instrumented with the parallel operational semantics, we analyze the Fibonacci function
when implemented directly as specified by the recursive definition: 
\[
  F(n) =
\begin{cases}
  0 \text{ if } n = 0\\
  1 \text{ if } n = 1\\
  F(n-1) + F(n-2) \text{ o.w. }
\end{cases}
\] 
There is the associated program in \pl{}:  
\begin{align*}
  fib &\triangleq \fun{n}{f}{\ifz{n}{\bar{0}}{n'}{\ifz{n'}{\bar{1}}{n''}{\parin{\ap{f}{n'}}{\ap{f}{n''}}{x}{y}{x \fbin{+} y}}}}
\end{align*}
We analyze the number of recursive calls $fib$ makes as a function of $n$ along the critical path. Let 
\begin{align*}
  P = \ifz{n}{\bar{1}}{n'}{\suc{\ifz{n'}{\bar{1}}{n''}{\suc{\fop{\max}(\bar{8} \fbin{\times} \suc{n'}, \bar{8} \fbin{\times} \suc{n''}) \fbin{+} \bar{5} \fbin{+} \bar{1}}}}}
\end{align*}
We verify the following specification: 
\begin{align*}
  \isVal{fib}{\arrtimev{\isOf{n}{\nat}}{\nat}{P}}
\end{align*}

\begin{proof}
  Let 
  \begin{align*}
    N &= \ifz{n}{\bar{0}}{n'}{\ifz{n'}{\bar{1}}{n''}{\parin{\ap{f}{n'}}{\ap{f}{n''}}{x}{y}{x \fbin{+} y}}}\\ 
  \end{align*}
  By Lemma~\ref{lemma:pitimei}, it suffices to show 
  \begin{enumerate}
    \item \isComp{\nat}{\univ{i}}. This holds by Lemma~\ref{lemma:natf}. 
    \item \openComp{\isOf{n}{\nat}}{\nat}{\univ{i}}. This holds by Lemma~\ref{lemma:natf}. 
    \item \openComp{\isOf{n}{\nat}}{P}{\nat}. By repeated applications of Lemma~\ref{lemma:natf},~\ref{lemma:hypothesis},
    ~\ref{lemma:nate1},~\ref{lemma:ffe2},and ~\ref{lemma:nati}. 
    \item \openComp{\isOf{n}{\nat},\isOf{f}{\arrtimev{\isOf{n}{\subsetty{\isOf{m}{\nat}}{[m/n]P \frel{<} P}}}{\nat}{P}}}{N}{P}.
  \end{enumerate}
  Let 
  \begin{align*}
    \Gamma_1 &= \isOf{n}{\nat},\isOf{f}{\arrtimev{\isOf{n}{\subsetty{\isOf{m}{\nat}}{[m/n]P \frel{<} P}}}{\nat}{P}}\\ 
    N_1 &= \ifz{n'}{\bar{1}}{n''}{\parin{\ap{f}{n'}}{\ap{f}{n''}}{x}{y}{x \fbin{+} y}}\\ 
    P_1 &= \ifz{n'}{\bar{1}}{n''}{\suc{\fop{\max}(\bar{8} \fbin{\times} \suc{n'}, \bar{8} \fbin{\times} \suc{n''}) \fbin{+} \bar{5} \fbin{+} \bar{1}}} 
  \end{align*}
  By Lemma~\ref{lemma:nate2}, it suffices to show 
  \begin{enumerate}
    \item \openComp{\Gamma_1,\isOf{n'}{\nat}}{\nat}{\univ{i}}. By Lemma~\ref{lemma:natf}.
    \item \openCComp{\Gamma_1}{n}{\nat}{\zero}. By Lemma~\ref{lemma:hypothesis}.
    \item \openComp{\Gamma_1}{\bar{0}}{\nat}. By Lemma~\ref{lemma:nati}. 
    \item \openComp{\Gamma_1,\isOf{n'}{\nat}}{P_1}{\nat}. By repeated applications of Lemma~\ref{lemma:hypothesis},
    ~\ref{lemma:nate1},~\ref{lemma:ffe},and ~\ref{lemma:nati}. 
    \item \openCComp{\Gamma_1}{\bar{0}}{\nat}{\bar{0}}. By Lemma~\ref{lemma:nati}. 
    \item \openCComp{\Gamma_1,\isOf{n'}{\nat}, \isOf{p}{\eqty{\nat}{\suc{n'}}{n}}}{N_1}{\nat}{P_1}. 
  \end{enumerate}
  Let 
  \begin{align*}
    \Gamma_2 &= \Gamma_1,\isOf{n'}{\nat}, \isOf{p}{\eqty{\nat}{\suc{n'}}{n}}\\ 
    N_2 &= \parin{\ap{f}{n'}}{\ap{f}{n''}}{x}{y}{x \fbin{+} y}\\
    P_2 &= \fop{\max}(\bar{8} \fbin{\times} \suc{n'}, \bar{8} \fbin{\times} \suc{n''}) \fbin{+} \bar{5} \fbin{+} \bar{1}
  \end{align*}
  By Lemma~\ref{lemma:nate2}, it suffices to show 
  \begin{enumerate}
    \item \openComp{\Gamma_2,\isOf{n''}{\nat}}{\nat}{\univ{i}}. By Lemma~\ref{lemma:natf}.
    \item \openCComp{\Gamma_2}{n'}{\nat}{\zero}. By Lemma~\ref{lemma:hypothesis}.
    \item \openComp{\Gamma_2}{\bar{0}}{\nat}. By Lemma~\ref{lemma:nati}. 
    \item \openComp{\Gamma_2,\isOf{n''}{\nat}}{P_2}{\nat}. By repeated applications of Lemma~\ref{lemma:hypothesis},
    ~\ref{lemma:ffe},and ~\ref{lemma:nati}. 
    \item \openCComp{\Gamma_2}{\bar{1}}{\nat}{\bar{0}}. By Lemma~\ref{lemma:nati}. 
    \item \openCComp{\Gamma_2,\isOf{n''}{\nat}, \isOf{q}{\eqty{\nat}{\suc{n''}}{n'}}}{N_2}{\nat}{P_2}. 
  \end{enumerate}
  Let 
  \begin{align*}
    \Gamma_3 &= \Gamma_2,\isOf{n''}{\nat}, \isOf{q}{\eqty{\nat}{\suc{n''}}{n'}}
  \end{align*}
  By Lemma~\ref{lemma:open_headexp} and~\ref{lemma:bin_seq}, it suffices to show 
  \begin{enumerate}
    \item \openCComp{\Gamma_3}{\ap{f}{n'}}{\nat}{\bar{8} \fbin{\times} \suc{n'}}. 
    In order to apply Lemma~\ref{lemma:pitimee}, we need to show 
    \[
    \openCComp{\Gamma_3}{n'}{\subsetty{\isOf{m}{\nat}}{[m/n]P \frel{<} P}}{\bar{0}}
    \]
    Let \eqInst{\gamma}{\gamma'}{\Gamma_3}. It suffices to show 
    \eqCComp{\gamma(n')}{\gamma'(n')}{\hat\gamma\subsetty{\isOf{m}{\nat}}{[m/n]P \frel{<} P}}{\bar{0}}. 
    We know that 
    \begin{enumerate}
      \item $\gamma(n') = \gamma'(n') = \bar{n'}$ for some $n'$.
      \item $\gamma(n) = \gamma'(n) = \bar{n}$ for some $n$.
      \item \isVal{\gamma(p)}{\eqty{\nat}{\suc{\bar{n'}}}{\bar{n}}}, so $n = n'+1$.
      \item \isVal{\gamma(q)}{\eqty{\nat}{\suc{\bar{n''}}}{\bar{n'}}}, so $n' = n''+1$.
    \end{enumerate}
    By Lemma~\ref{lemma:subseti}, it suffices to show 
    \[
    \isVal{\triv}{[\bar{n'}/m]([m/n]P \frel{<} P)}  
    \]
    There are three cases. If $n' = 0$, then
    \begin{align*}
      [\bar{n'}/m][m/n]P \mapsto^* \bar{1}\\
      P \mapsto^* \bar{2}
    \end{align*}
    since $1 < 2$, the result holds by Lemma~\ref{lemma:reli}. If $n' = 1$, then
    \begin{align*}
      [\bar{n'}/m][m/n]P \mapsto^* \bar{2}\\
      P \mapsto^* \overline{24}
    \end{align*}
    so again the result holds by Lemma~\ref{lemma:reli}. If $n' = m+2$ for some $m$, then 
    \begin{align*}
      [\bar{n'}/m][m/n]P \mapsto^* \overline{8 + 8(m+2)}\\
      P \mapsto^* \overline{8 + 8(m+3)}
    \end{align*}
    and the result holds by Lemma~\ref{lemma:reli}. 

    Now applying Lemma~\ref{lemma:pitimee}, we have 
    \[
      \openCComp{\Gamma_3}{\ap{f}{n'}}{\nat}{[n'/n]P}
    \]
    By Lemma~\ref{lemma:cost_weak}, it suffices to show 
    \[
    \openComp{\Gamma_3}{\triv}{[n'/n]P \frel{\le} \bar{8} \fbin{\times} n'} 
    \]
    Let \eqInst{\gamma}{\gamma'}{\Gamma_3}. We need to show 
    \isComp{\triv}{\hat\gamma[n'/n]P \frel{\le} \bar{8} \fbin{\times} \suc{n'}}.
    We know that $\gamma(n') = \gamma'(n') = \bar{n'}$ for some $n'$.
    Again, there are three cases. If $n' = 0$, then 
    \begin{align*}
      \hat\gamma[n'/n]P \mapsto^* \bar{1}\\
      \hat\gamma\bar{8} \fbin{\times} \suc{n'} \mapsto^* \overline{8}
    \end{align*}
    so the result holds by Lemma~\ref{lemma:reli}. If $n' = 1$, then 
    \begin{align*}
      \hat\gamma[n'/n]P \mapsto^* \bar{2}\\
      \hat\gamma\bar{8} \fbin{\times} \suc{n'} \mapsto^* \overline{16}
    \end{align*}
    so the result holds by Lemma~\ref{lemma:reli}. If $n' = m+2$ for some $m$, then 
    \begin{align*}
      \hat\gamma[n'/n]P \mapsto^* \overline{8 + 8\times(m+2)}\\
      \hat\gamma\bar{8} \fbin{\times} \suc{n'} \mapsto^* \overline{8 \times (m+3)}
    \end{align*}
    since $ 8 + 8\times(m+2) = 8\times(m+3) \le 8 \times(m+3)$, the result holds by Lemma~\ref{lemma:reli}. 
    \item \openCComp{\Gamma_3}{\ap{f}{n''}}{\nat}{\bar{8} \fbin{\times} \suc{n''}}. Similar to above. 
    \item \openCComp{\Gamma_3, \isOf{x}{\nat}, \isOf{y}{\nat}}{x \fbin{+} y}{\nat}{\bar{1}}. By Lemma~\ref{lemma:ffe2}.
  \end{enumerate}
\end{proof}
\section{Related work}\label{sec:related}

There is a substantial amount of literature related to the formulation of computational complexity and 
concrete analysis of algorithmic complexity in the context of both type theories and program logics. 
The following attributes are sometimes helpful in
categorizing a particular approach: higher-order/first-order functions, dependent types/simple types, 
structural/concrete complexity, imperative/functional programs, and automated/manual verification. 
While this article presents a higher-order dependent type theory focusing on concrete complexity and manual verification, 
various combination of the attributes have been explored in prior work. 
In the following, we discuss some of these works and their relations to \cctt{}. \\

\noindent\textbf{Internalizing complexity classes}
An early line of work in the context of \nuprl{}\footnote{\nuprl{} is a proof
refinement framework developed in the 1980s by Constable et
al~\cite{10.5555/10510}. 
It also refers to a family of 
computational type theories (which \cctt{} is based on) and has been fruitful in both mechanizing existing mathematics and 
proving new results~\cite{doi:10.1002/malq.201700057}.} explored the possibility of internalizing the definition of complexity classes in type theory.
In the paper ``A Note on Complexity Measures for Inductive Classes in Constructive Type Theory''~\cite{CONSTABLE1998137}, the author utilizes existing
infrastructure of a computational type theory to define \emph{inductive} classes of functions of type $\N \to \N$, 
examples of which include complexity classes such as \ptime{}. The idea is that a function $f$ is in an inductive class 
$\mathcal{C}$ if it is extensionally equal to a base function or a composition of functions in $\mathcal{C}$. In 
prior work by Bellantoni and Cook~\cite{10.1145/129712.129740},
\ptime{} was given a recursion-theoretic characterization via an inductive class $B$ closed under 
special recursion and composition operators. 
Exploiting this fact, the author defines an internal class of functions $\mathcal{P}$ 
which is extensionally equivalent 
to the class $B$. Crucially, an element of $\mathcal{P}$ is a pair $(f,pf)$, where 
$f : \N \to \N$ is the underlying function and $pf$ is the \emph{proof object} that $f$ is in the inductive class. 
The main results of the paper are the following:
\begin{enumerate}
  \item The programming language is \emph{internally feasible}, meaning that all polytime functions $\N \to \N$ 
  can be implemented internally. 
  \item There is an internally defined complexity measure $\textit{Time}$, such that for each $(f,pf) \in \mathcal{P}$, $\textit{Time}(pf)$ is a function 
  mapping a number $n : \N$ to the number of reduction steps taken by the proof object $pf$ applied to $n$. Furthermore, 
  $Time$ is \emph{faithful} to the external complexity measure $\textit{time}$, meaning
  that $\textit{Time}(pf) = \textit{time}(\textit{term}(pf))$, where 
  $\textit{term}$ is a meta function taking the proof $pf$ to the function it represents. 
\end{enumerate}
This paper presents a novel use of inductive types in type theory to capture intensional aspects of computational behavior, in 
particular computational complexity. 
In particular, the proof object witnessing membership of a complexity class is used to define the internal complexity measure, 
which is necessary because the theory is essentially extensional.
However, unlike \cctt{}, the theory presented does not speak directly about higher-order 
functions, dependent types, or rules for concrete complexity analysis. \\

\noindent In another paper exploring intensional reasoning principles~\cite{CC01}, the authors present a computational type theory 
with a lightweight mechanism for reflecting the term structure and parts of the operational semantics into the object theory.  
This machinery allows for a ``deep embedding'' of the terms of the theory. Among others, there are several key ingredients: 
\begin{enumerate}
  \item A meta level function $\reflect{\cdot}$ which reflects a term $t$ into the internal representation. 
  \item A type $\textit{Term}$ containing all reflected terms in the language.
  \item An operator $\textit{ref}$ on $\textit{Term}$ which computes the denotation of a reflected term.
  \item A type $[A]$, inhabited by terms $t$ such that $\textit{ref}(t) \in A$.
  \item A type $t_1\; \textit{evalto}\; t_2$, inhabited exactly when $\textit{ref}(t_1) \mapsto \textit{ref}(t_2)$.
  \item A type $\textit{iscanon}(t)$, inhabited exactly when \final{\textit{ref}(t)}.
\end{enumerate}
With these notions, the authors define $\textit{time} : \mathcal{P}(\textit{Term} \times \N)$ and $\textit{size} : \textit{Term} \to \N$, where 
\begin{enumerate}
  \item $\textit{time}(t,n)$ iff there is an sequence $(t_i)_{i}$ of length $k$ such that $t = t_0\;\textit{evalto}\;t_1 \dots \textit{evalto}\; t_k$ and $\textit{iscanon}(t_k)$. 
  \item $\textit{size}(t)$ is the size of the reflected term $t$.
\end{enumerate}
The complexity class \ptime{} can then be defined as: 
\[
\ptime{} \triangleq \{ f : [\pity{\isOf{a}{A}}{B}] \mid \exists c, c'.\,\forall \isComp{a}{[A]}.\,
  \textit{Time}(\apabt{f}{a}, c \cdot \textit{size}(a)^{c'})
\}  
\]
Such internal characterizations of complexity classes 
can also be defined in \cctt{} by restricting the funtime function space using the same
technique, given that the input type $A$ has an appropriate \emph{size metric}, which we will discuss shortly. 
While the type theory presented in~\cite{CC01} internalizes a relatively significant portion of the operational semantics, 
\cctt{} only internalizes the notion of cost, which eliminates the technical overhead required to 
instrument the reflection semantics. However, having access to the internal representation of a term enables the possibility 
for more fine-grained analysis. As mentioned above, with a term representation, there is a notion of size at every type 
(in fact for any untyped term). This is why it is possible to define the class \ptime{} for arbitrary dependent functions as 
shown above. 
Although it is convenient to have a default definition for the size of a term,
the relevant \emph{figure of merit} in complexity analysis is often specific to the algorithm in question. For example, in tree 
or graph problems, the relevant size measure ranges from the total number of nodes, the height of the tree, 
the maximal degree, number of connected components, etc. With this in mind, the lack of a default size metric does not appear
to be a significant drawback to us, as it is most likely application dependent. In \cctt{}, the user of the theory 
dictates the relevant metrics for each relevant type $A$ by defining a function $A \to \N$. 
Lastly, while internal term representations allow the theory to differentiate syntactically distinct expressions denoting the 
same value, it is rarely necessary to have this level of detail.

Here, we emphasize the point that \cctt{} is an \emph{extensional theory}
in the sense that functions are equated up to their extensions. The key difference from normal extensional type theories is the \emph{type} at which 
functions are queried for equality. 
Conventionally, a function $f : A \to B$ is completely described by the set of its input-output pairs 
$\{(x,f(x)) \mid x \in A\}$. In \cctt{}, the semantics of the function space is enriched so that it consists of input-output-cost \emph{triples} 
$\{(x,f(x),c(x)) \mid x \in A\}$, where $c : A \to \N$ is the map describing the complexity of $f$ on each input.
In this context, two functions $f,g : \arrtimev{\isOf{a}{A}}{B}{P}$
are equal iff $f(a) = g(a)$ and $c_f(a), c_g(a) \le P$ for all $a \in A$.
Hence this equality principle is morally the same as the usual function extensionality principle, and the
equality of functions is not governed by syntactic rules present in variants of intensional type theories.\\

\noindent\textbf{Recurrence relations}
Aside from defining complexity classes internally, there was also work on
automated concrete cost analysis in the setting of NuPRL. For instance, Benzinger~\cite{BENZINGER200479} describes the framework ACA for 
automated cost analysis of higher-order functional programs. In this work, computational complexity is defined through a cost-annotated 
operational semantics and extended to open terms by symbolic execution on meta terms. A type decomposition scheme is used 
to define the complexity and denotation of higher-order types as compositions of the decomposition at the constituent types. 
Higher-order recurrence relations are generated during symbolic execution and simplified to first-order relations, which 
are solved using a computer algebra system. The framework is used to analyze three different proofs of the pigeonhole principle,
which was used as a lemma in a state minimization algorithm, 
and it was shown that the original proof was the source of an exponential slowdown, and 
the efficient proofs of the lemma yielded the algorithm with the expected runtime. 

As is usual with automated systems, ACA assumes a particular closed-form when solving recurrence relations. 
This differs fundamentally from \cctt{}, where cost bounds are not assumed to have a certain form and
can be arbitrary computations. Another difference is the lack of dependent types in ACA. In general, 
cost analysis may require arbitrarily sophisticated \emph{functional} specifications, 
which means dependent types is also crucial for accurate descriptions of \emph{complexity}. 
The type decomposition function in ACA is only defined for natural numbers and simple function and product types,
and it is unclear how it can be extended to dependent types. 
As noted by the author, since ACA does not handle dependent types such as subsets of $\N$, 
the analysis was unable to detect dead code in one of the proofs of the pigeonhole principle, 
which led to an unexpected exponential cost bound for a polynomial algorithm.
Since \cctt{} features dependent types, such information can be fully leveraged in complexity analysis.\\

\noindent More recently, Kavvos et al.~\cite{10.1145/3371083} presented a framework that put informal algorithmic complexity analysis on formal mathematical footing. 
An important step in informal algorithm analysis is the extraction of a recurrence relation from the program in question that describes the computational complexity of 
the algorithm as a function of the size of the inputs. Often this process is done in a hand-waving fashion with little formal justification. 
In~\cite{10.1145/3371083}, the authors address this problem by presenting an automatic recurrence extraction framework for PCF programs. 
First, the procedure extracts a syntactic recurrence from the source program.
The syntactic recurrence is a program in PCF extended with cost expressions,
denoted as PCF$_{c}$, intended to 
mirror usual mathematical expressions. Next, the syntactic recurrence is interpreted into a semantic domain called ``sized domains'' that models the mathematical recurrences 
used in informal analyses, thus completing the formalization of the informal
process of algorithm analysis. 

The key idea behind the extraction is that each type $A$ is interpreted into a \emph{complexity} $||A||$ consisting of a cost expression $\mathbb{C}$, which accounts for evaluation cost,
and a potential $\pot{A}$, which accounts for the size and cost of future uses of the result of evaluation. 
Expressions $M$ are extracted to a term $||M||$ in PCF$_c$.
Then the soundness of the extraction procedure states that a if \isOf{M}{A} in PCF then \isOf{||M||}{||A||} in PCF$_c$.
A bounding relation $M \le_{||A||} ||M||$ is used to relate the actual complexity of $M$ to the extracted recurrence $||M||$. 
The bounding theorem then ensures that the above always holds given a well-typed term \isOf{M}{A}. 

In~\cite{10.1145/3371083}, Kavvos et al. also consider call by name evaluation by factoring the extraction through the intermediate a call-by-push-value language~\cite{Levy:2006:CDC:1187995.1187999}, 
which we do not consider. Zooming out, their work is focused on formalizing the process of recurrence extraction, while our work with \cctt{} is primarily an exploration of the semantics of a type theory 
with an inherent notion of complexity. While recurrence relations is an established and very powerful method for analyzing algorithms, it is not the only method, and we have defined the semantics of \cctt{} 
as to not exclude this possibility. Although we have no mechanism to automatically extract recurrence relations, it is plausible to verify the correctness of a \emph{given} recurrence relation with respect to 
an algorithm in \cctt{}. \\

\noindent\textbf{Descriptive Complexity}
A related field that lies in the intersection of structural complexity theory and programming languages theory 
is descriptive complexity or implicit computational complexity, where the focus is mainly on syntactic characterizations 
of complexity classes such as \ptime{}~\cite{10.1145/129712.129740},
\logspace{}~\cite{RAMYAA2011247}, \elementary{}~\cite{DANOS2003123}, etc. 
Typical type structures considered in this setting are iterated function spaces generated from $\N$, which means the
specification language is limited compared to dependent type theories. 
One of the most notable results from implicit computational complexity is 
Bellantoni and Cook's recursion-theoretic characterization of \ptime{}~\cite{10.1145/129712.129740}. The authors introduced the 
notion of \emph{recursion on notation} and \emph{safe composition} and divided the inputs of functions 
$\N^k \to \N$ as $\N^m \times \N^n \to \N$, where the first $m$ inputs are \emph{normal} and the latter $n$ are \emph{safe}. 
They then define a class of functions $B$ as the least class containing some initial functions and closed under 
recursion on notation and safe composition. The intuitive idea is that these recursion and composition operators ensure 
that safe inputs can only be used to compute 
a result no larger than an additive constant from the input, while normal inputs can be used in any polytime function.
The main theorem is the equivalence of $B$ and \ptime{}:
\begin{enumerate}
  \item $f : \N^m \to \N$ in \ptime{} implies $\hat f : \N^m \times \N^0 \to \N$, $(x,) \mapsto f(x)$ is in $B$.
  \item $f : \N^m \times \N^n \to \N$ in $B$ implies $\hat f : \N^{m + n} \to \N$, $(x,y) \mapsto f(x,y)$ is in \ptime{}.
\end{enumerate}
And hence $B$ is a complete characterization of \ptime{}.

Clearly, this recursion-theoretic characterization is aimed directly at capturing \emph{structural} computational complexity, 
in the sense that programs in $B$ are guaranteed to be polytime, but there is no way of describing what the bounding polynomial actually is. 
Furthermore, even if an algorithm for a problem is known to be in polytime, for example merge sort for sorting, it is not 
guaranteed that one is able to directly implement that algorithm and prove that it is in $B$. What is guaranteed is that 
there \emph{exists} a polytime algorithm for the \emph{problem} in $B$. The tension between implementing algorithms and characterizing 
complexity classes is partially solved in subsequent works utilizing substructural type systems and amortized analysis. \\

\noindent\textbf{Linear types \& amortized analysis}
Following the syntactic characterization of \ptime{}, various papers~\cite{HOFMANN2000113}\cite{BELLANTONI200017} showed that 
linear types and modal type systems can be used to extend the class $B$ to higher types. In particular, 
Hofmann presented a type system (which we will refer to as SLR for ``structural linear recursion'') 
for non-size-increasing polynomial time functions~\cite{HOFMANN200357} that was 
suited for ordinary programming. The idea is to account for the iterations on inductive data types through 
a new type $\Diamond$ that represents resource available for computation. For example, a term $l$ of a list type
is either nil or a constructor $cons(\blacklozenge,x,xs)$, where the term $\blacklozenge$ represents the sole element of the type $\Diamond$. 
For example, the double function on lists is assigned the following type:
\footnote{Taken from lecture notes by Jan Hoffmann: https://www.cs.cmu.edu/~janh/courses/ra19/assets/pdf/lect05.pdf }
\[
\textit{double} : L(A \times \Diamond) \to L(A) 
\]
which can be interpreted to mean that given an input list $l$ of length $n$, $\textit{double}(l)$ requires $n$ diamonds 
to compute the result. The main results of the paper is that 
\begin{enumerate}
  \item Given any $f : \N \to \N$ that is polytime and linear space, and such that $|f(x)|\le |x|$ for all $x$,
  then $f$ is the denotation of a term $e : \N \to \N$ in SLR. 
  \item Given any term $e : \N \to \N$ in SLR, the denotation of $e$ is polytime computable. 
\end{enumerate}
The main benefit over previous type systems is that SLR respects common patterns of programming such as 
recursion over the results of previous recursive computations, making it possible to actually implement known algorithms. 
However, as in the case for the recursion-theoretic characterization of \ptime{}, 
the cost bounds are implicit, and it was unclear how to express complex cost bounds 
(such as nonlinear functions of the inputs).
In the following decade, the potential for complexity analysis in SLR is developed through a series of papers
~\cite{Hofmann:2003:SPH:604131.604148}\cite{10.1007/978-3-642-11957-6_16}\cite{Jost:2010:SDQ:1707801.1706327}\cite{10.1007/978-3-642-17164-2_13}\cite{10.1145/2362389.2362393}, 
the culmination of which lead up to a versatile and practical automated resource analysis framework for 
functional programs, Resource Aware ML (RaML).\\ 

\noindent\textbf{Resource Aware ML} 
Resource Aware ML~\cite{Hoffmann:2017:TAR:3009837.3009842} is an automated resource analysis framework based on an affine 
type system and linear programming. Similar to the diamond types in system SLR, types in RaML are annotated with a 
measure describing the 
amount of potential stored that is available to pay for computation. 
However, these annotations are only present in types and have no term representation, which relieves the programmer 
from dealing with the allocation of diamonds necessary in SLR.
In the most basic form, the potential annotations are 
nonnegative numbers assigned to each constructor of an inductive type, i.e. a list would be annotated as 
$L^p(A)$, which could be read as a list of type $A$ where each element of the list has $p$ units of potential. 
For the univariate and multivariate polynomial versions, the annotations become more complex, but the idea is similar.
Type derivations are judgments of the form $\Gamma; p \vdash e : (\tau,q)$, which can be read as 
``under context $\Gamma$ and with 
constant potential $p$, $e$ has type $\tau$ with remaining potential $q$''. 
The introduction forms for inductive types request potential from the context, 
while the elimination forms release the stored potential for use in computations. 
During type checking, linear inequalities are generated to ensure the soundness of the derivation, and these constraints 
are collected and passed to an LP solver. If there is a minimal solution, it is used to construct the best worst-case cost bound,
and the soundness theorem ensures that the derived cost bound is an upper bound on the cost of evaluation. 
Recently, RaML has been extended to handle a myriad of computational mechanisms and behaviors, including 
parallelism~\cite{10.1007/978-3-662-46669-8_6}, arrays and references~\cite{lichtman_et_al:LIPIcs:2017:7728}, 
garbage collection~\cite{LPAR-22:Automatic_Space_Bound_Analysis}, and exponential-time algorithms~\cite{10.1007/978-3-030-45231-5_19}. 

As mentioned in the discussion of ACA, while RaML is an \emph{automated} framework for complexity analysis, 
\cctt{} is a semantic type theory for (at the moment) \emph{manual} verification. As a semantic type theory, there can be no 
complete decision procedure for typehood or membership. 
This means that while the complexity of a large class of programs (such as many functions from OCaml's standard library) can be 
directly analyzed using RaML, verification of these programs are completely manual in \cctt{}. On the other hand, the burden of
automation precludes analysis of algorithms whose complexity is not expressible as a polynomial in the input. 
However, in a semantic type theory such as \cctt{}, the verification of a judgment can rely on any mathematical fact or proof strategy,
so the precision of the cost bound can be as detailed as theoretically possible. Of course, in practice it is very convenient
to discharge routine verifications automatically, and one direction for future work is to investigate how to
interpret derivations from RaML as judgments in \cctt{}. \\

\noindent\textbf{Program logics} So far, we have given a brief account of the relevant works in the setting of functional programming and 
type theories. Another approach to analyzing computational complexity can be found in the program logics community, which 
takes a more external view on verification. Broadly speaking, in dependent type theory, types, and thus programs, \emph{is}
the language for specifications. In contrast, program logics are usually stratified in the sense that the logic for specifications
is separate from the programs being analyzed. (In the case that subject of analysis are imperative programs, this distinction 
seems inevitable because imperative programs are not well-behaved as mathematical specifications.) 

Program logics has its origin in Hoare logic~\cite{10.1145/363235.363259}, which
has been extended to separation logic~\cite{10.5555/645683.664578} 
and the Iris program logic~\cite{jung_krebbers_jourdan_bizjak_birkedal_dreyer_2018} to handle highly complex computational phenomena. 
For us, the relevant body of work is the framework of time credits and receipts ~\cite{10.1007/978-3-030-17184-1_1}
built on top of the separation logic framework CFML~\cite{10.1145/2034773.2034828}. Conceptually, the idea of time credits is 
quite simple: in addition to ownership of fragments of memory, predicates on the heap also govern ownership or existence 
of credits that can be used to pay for computation. Because the ambient logic is affine, this is consistent with the intuitive idea
that time credits can be split and joined as necessary in an amortization argument. Time receipts is the dual concept: it is proof 
that a computation requires a certain amount of steps, which can be used to establish lower bounds. This is useful when 
proving that some ``bad'' event (such as overflows) will not happen within a given time frame. The time credits and receipts 
framework has been used to verify sophisticated complexity bounds on algorithms such as union-find~\cite{Chargueraud:2019:VCA:3315672.3315720}. 

Recently, the framework has also been used to formalize \emph{asymptotic} cost analysis~\cite{10.1007/978-3-319-89884-1_19} that 
enables the verification of asymptotic bounds analogous to the Big O notation used in informal analyses. Although 
the use of Big O and other related asymptotic notations (Small O, asymptotic
equivalence, etc.) proliferates algorithmic complexity
literature, it is surprisingly difficult to formalize some generalizations such as multivariate functions, and it has 
been shown that some common intuitions does not hold under naive extensions to univariate Big O~\cite{Howell2008OnAN}. 
The authors mitigate the issue of multivariate Big O through 
the use of \emph{filters}, a technical device that allows multiple interpretations of asymptotic notation depending on 
the situation. Another well-known problem with using asymptotic notation is that the inductive hypothesis 
associated with a Big O bound is not strong enough to push through the verification for recursive functions, where concrete cost 
bounds are necessary. To deal with the tension between the desire to provide modularity through asymptotic bounds and the 
need for concrete cost bounds, the framework separates the verification and publication of the complexity specification. 
First, a concrete cost bound is verified against the target program. This cost bound is then hidden from clients of the 
function, and an abstracted interface composed of the asymptotic bound and some properties of the cost bound (such as 
monotonicity and nonnegativity) is publicized instead. This strikes a balance between using completely concrete cost bounds,
which is non-modular, and asymptotic bounds, which is not powerful enough alone for verification. 

When compared to \cctt{}, the frameworks developed in~\cite{10.1007/978-3-030-17184-1_1} and~\cite{10.1007/978-3-319-89884-1_19} 
offer more systematic methods for complexity verification. In particular, there is support for 1) semi-automatically constructing 
the recurrence relation for the cost expression and 2) guessing a closed form and verifying the solution via the substitution method.
The approach alleviates the user from having to have a concrete cost bound in mind prior to the verification process; all that is 
necessary is a guess for the general shape of the cost expression. 
The proof theory we provided for \cctt{} is designed to guide the verification of a \emph{known} cost bound, but does not help
when one tries to guess and verify a cost bound simultaneously. Another
significant difference is the fact that the time credits framework is able to
handle imperative programs, something we have not considered. It is well-known
that effects are difficult to model in type theory, and it is even
less clear what effects even mean in the presence of dependent types. Although
it is an interesting challenge, complexity analysis for imperative programs
seems to be far out of reach in the kind of theory we have developed.
\section{Conclusion \& future work}\label{sec:concl}

In this paper, we presented a computational type theory \cctt{} for stating and
verifying the time complexity of programs.
While not the first of its kind, \cctt{} features dependent types, higher-order functions, and rules for 
manipulating concrete complexities, a
combination that can serve as a foundation for both program verification and feasible mathematical proofs. 
The main contributions include the formalization of the judgments of \cctt{} and the soundness of the associated proof theory.
Two examples are worked out in detail to demonstrate the practical applicability of the theory.
We also note that the extension to accommodate parallelism only requires a straightforward change to the underlying operational 
semantics, suggesting that the theory is quite flexible with respect to the specific facilities of the programming language and 
the complexity measure. Some interesting directions for future work include developing formal proof theories, 
investigating feasible mathematics, studying different computational behaviors, and modeling abstract complexity analysis. \\

\noindent\textbf{Proof theory}
For the purposes of demonstration, we have presented some useful lemmas and packaged them up into a makeshift proof theory.
It would be interesting to develop a formal proof theory that acts as the syntactic counterpart 
to the semantic theory. One line of approach would be interpreting the type derivations of a formalism such as Resource Aware 
ML as semantic consequences in \cctt{} and prove that all RaML rules are sound. Some cursory attempts at this failed because 
the justifications for termination differ significantly. In \cctt{}, the termination of a function follows from a 
semantic argument: the function introduction rule requires recursive calls to have strictly less cost than the 
current argument and thus induces a well-founded order on the calls to the function. 
In RaML, termination follows from a syntactic argument: 1) variables are linear and 
2) type derivations have finite constant potential, which combined implies there can only be a finite number of function applications. 
We speculate that it is necessary to consider the entire derivation tree of a RaML judgment when interpreting into \cctt{}. 
A different approach would be to recreate RaML's termination proof in \cctt{} as another lemma for introducing function types. \\

\noindent\textbf{Feasible mathematics}
Various works surveyed in the previous section~\cite{BENZINGER200479}\cite{CONSTABLE1998137}\cite{CC01} 
has brought up the subject of \emph{feasible mathematics}, a field that enriches the propositions-as-types 
principle by considering realizers of proofs that are computationally feasible. In one direction, this means constructive proofs 
of theorems can be interpreted as programs realizing a specification. Conversely, this also subjects proofs to analyses
usually applied to programs such as algorithmic complexity.
(Note that this notion should not be confused with the term \emph{proof complexity}, 
which often refers to analyses of lower bounds on the size of proofs in formal systems.)
In terms of formalisms, ideas from bounded quantifiers and bounded linear logic~\cite{GIRARD19921} seem germane.
Additionally, a more thorough treatment of higher-order complexity 
and interactions between proof-relevant and proof-irrelevant components will be important for developing a theory of
\emph{computational} proof complexity in the sense described above. 

One concrete application of computational proof complexity arises
in the fast-developing area of higher-dimensional type theory.
Since Voevodsky posited the idea of univalent
mathematics~\cite{Voevodsky2010UnivalentFP}, the subject has seen rapid progress:
various well-known theorems were reestablished and the cross pollination between homotopy
theory and type theory led to new insights in both fields. One particular
difficult result proved in this setting was the computation of the homotopy
groups of spheres~\cite{Brunerie2016OnTH}. The result was implemented in Cubical
Agda~\cite{10.1145/3341691}, but the proof failed to terminate (even when
given an intuitively reasonable amount of computing resources) before running
out of memory. Given recent work on endowing higher-dimensional proofs with
computational
semantics~\cite{angiuli_et_al:LIPIcs:2018:9673}\cite{10.1145/3290314}, it
becomes possible (at least in principle) to subject the proof to
complexity analysis in a theory similar to \cctt, which could help determine
whether more computational power/optimization is necessary for termination 
or that the proof is inherently intractable. \\

\noindent\textbf{Computational behaviors}
Another direction for future work would be verification techniques for complexity properties when the underlying 
computational system exhibit additional information/behaviors. 
For example, one can consider computational effects such as
randomization~\cite{10.1145/3290377}\cite{wang2020tail}\cite{WANG2019303}. 
A challenge with probabilistic programs is that the operational semantics
becomes nondeterministic, while key lemmas such as head expansion rely on the
determinism of evaluation.
Another direction would be giving a semantics for partial computations/functions and
develop a theory of partial complexity analysis. In order to make sense of
partiality in type theory,
prior work in \nuprl{}~\cite{Constable1987PartialOI} on the construction and
semantics of partial objects via bar types serves as a good starting point. Note that whereas the verification of a
\emph{time} bound inherently forces the function in question to be total,
satisfaction of a \emph{space} bound is orthogonal to the termination behavior of the
function. Given this observation, it would be interesting to analyze
computational resources (concrete or abstract) that a priori have no direct connection to totality. \\

\noindent\textbf{Abstract complexity}
For simplicity, we have arranged every transition step to have unit cost in \cctt{}. The first step towards a less 
concrete measure of cost would be to designate a single expression which effect cost, similar to the ``tick'' construct 
used in RaML. By allowing the users to place ticks in the program, what counts as ``costful'' becomes an extrinsic concept, 
which enables more targeted forms of analysis. 
It is also possible to specify cost measures which are not natural numbers. For example, 
space complexity for functional programs~\cite{10.1145/1411204.1411240} can be defined through a cost semantics
in which a \emph{graph} is the complexity measure. However, such a cost semantics is usually phrased in big-step 
evaluation rules, whereas small-step transition rules are needed to state critical lemmas such as head expansion in type theory.  

An even more abstract account of complexity arises by way of analogy to abstract data types. Mirroring
the protocol of information control between the provider and client of an ordinary abstract data type interface, 
a framework for abstract complexity-aware data types should ensure that the client of the interface is  
only privy to the fact that a particular library function has the specified functional behavior and complexity
and is unaware of the underlying implementation. This mechanism can also be used to conduct cost analysis
assuming that \emph{abstract operations} have certain
complexities, which is a notion that is difficult to describe directly at the level of the operational semantics.

\section*{Acknowledgement}
We are thankful to  Jan Hoffmann, Karl Crary, Jon Sterling, and Carlo
Angiuli for the helpful discussions. 
The authors gratefully acknowledge the support of the 
Air Force Office of Scientific Research through MURI grant FA9550-15-1-0053, 
the National Science Foundation through grant CCF-1901381, and the
Air Force Research Laboratory through the NDSEG fellowship. 
Any opinions, findings and conclusions or recommendations expressed in this
material are those of the authors and do not necessarily reflect the views of
the AFOSR, NSF, or AFRL.

\nocite{*}
\bibliographystyle{plain}
\bibliography{bib}
\end{document}